\documentclass[prx,twocolumn,showpacs,amsmath,amssymb,superscriptaddress]{revtex4-2}
\usepackage{braket}
\usepackage{cases}
\usepackage{color}
\usepackage{graphicx,subfigure}
\usepackage{float}
\usepackage{bm}
\usepackage{hyperref}
\hypersetup{
 colorlinks=true, 
 linkcolor=cyan,
 citecolor=magenta, 
 filecolor=magenta, 
 urlcolor=cyan,
 runcolor=cyan
}
\usepackage[capitalise]{cleveref}
\usepackage{booktabs} 

\usepackage{bibunits}

\usepackage{multirow}

\usepackage{adjustbox}

\usepackage{algorithm}
\usepackage[noend]{algpseudocode}
\usepackage{comment}

\newcommand{\T}{\mathcal{T}}
\newcommand{\Lseq}{\mathcal{\rm L}}
\newcommand{\Gaps}{\mathcal{G}}
\newcommand{\Err}{\mathrm{LogErr}}
\newcommand{\dt}{\mathrm{dt}}

\newcommand{\kyiv} {\texttt{ibm\_kyiv}}
\newcommand{\fez} {\texttt{ibm\_fez}}
\newcommand{\marrakesh} {\texttt{ibm\_marrakesh}}
\newcommand{\aachen} {\texttt{ibm\_aachen}}
\newcommand{\pittsburgh} {\texttt{ibm\_pittsburgh}}

\newcommand{\dxdz}[2]{$(#1,#2)$}

\usepackage{amsthm} 

\newtheorem{mylemma}{Lemma}

\newcommand{\eps} {\varepsilon}
\newcommand{\epstot} {\eps_{\mathrm{tot}}}
\newcommand{\epserr} {\eps_{\mathrm{err}}}
\newcommand{\delerr} {\Delta_{\mathrm{err}}}

\newcommand{\dphi}{\delta\varphi}
\newcommand{\dth}{\delta\theta}
\newcommand{\Herr}{H_{\mathrm{err}}}

\usepackage{commath}
\usepackage{color}
\usepackage{amsmath}
\usepackage{adjustbox}
\usepackage[normalem]{ulem}

\newcommand{\bes} {\begin{subequations}}
\newcommand{\ees} {\end{subequations}}
\newcommand{\beq}{\begin{equation}}
\newcommand{\eeq}{\end{equation}}

\def\b{\beta}

\def\>{\rangle}
\def\<{\langle}
\def\Tr{\mathrm{Tr}}
\def\Pr{\mathrm{Pr}}
\newcommand{\ketbra}[1]{|{#1}\>\!\<#1|}

\newcommand{\ketb}[2]{|{#1}\>\!\<#2|}

\def\XY4{\text{XY}4}

\def\>{\rangle}
\def\<{\langle}
\def\Tr{\mathrm{Tr}}
\def\Pr{\mathrm{Pr}}
\def\bd{\boldsymbol{d}}

\def\Tr{\mathrm{Tr}}

\newcommand{\Fe}{F_{\mathrm{e}}}
\newcommand{\Felow}{F_{\mathrm{e},<}}
\newcommand{\Fehigh}{F_{\mathrm{e},>}}

\makeatletter
\newcommand{\startappendixtoc}{%
 \let\revtex@orig@addcontentsline\addcontentsline
 \renewcommand{\addcontentsline}[3]{%
 \def\APX@toc{toc}%
 \def\APX@tmp{##1}%
 \ifx\APX@tmp\APX@toc
  \revtex@orig@addcontentsline{apx}{##2}{##3}%
 \else
  \revtex@orig@addcontentsline{##1}{##2}{##3}%
 \fi
 }%
}
\newcommand{\stopappendixtoc}{%
 \let\addcontentsline\revtex@orig@addcontentsline
}
\newcommand{\printappendixtoc}{%
 \section*{Appendix Contents}%
 \@starttoc{apx}%
}
\makeatother

\begin{document}
\begin{bibunit}[apsrev4-2]

\title{Surface code scaling on heavy‑hex superconducting quantum processors}

\author{Arian Vezvaee}
\thanks{These authors have contributed equally to this work.}
\affiliation{Department of Electrical \& Computer Engineering, University of Southern California, Los Angeles, California 90089, USA}
\affiliation{Center for Quantum Information Science \& Technology, University of
Southern California, Los Angeles, CA 90089, USA}
\affiliation{Quantum Elements, Inc., Thousand Oaks, California, 91361, USA}

\author{Cesar Benito}
\thanks{These authors have contributed equally to this work.}
\affiliation{Instituto de Fisica Teorica UAM-CSIC, Universidad Autonoma de Madrid, Cantoblanco, 28049, Madrid, Spain}

\author{Mario Morford-Oberst}
\affiliation{Department of Electrical \& Computer Engineering, University of Southern California, Los Angeles, California 90089, USA}
\affiliation{Center for Quantum Information Science \& Technology, University of
Southern California, Los Angeles, CA 90089, USA}

\author{Alejandro Bermudez}
\affiliation{Instituto de Fisica Teorica UAM-CSIC, Universidad Autonoma de Madrid, Cantoblanco, 28049, Madrid, Spain}

\author{Daniel A. Lidar}
\affiliation{Department of Electrical \& Computer Engineering, University of Southern California, Los Angeles, California 90089, USA}
\affiliation{Center for Quantum Information Science \& Technology, University of
Southern California, Los Angeles, CA 90089, USA}
\affiliation{Quantum Elements, Inc., Thousand Oaks, California, 91361, USA}
\affiliation{Department of Physics \& Astronomy, University of Southern California,
Los Angeles, California 90089, USA}
\affiliation{Department of Chemistry, University of Southern California,
Los Angeles, California 90089, USA}

\begin{abstract}
Demonstrating subthreshold scaling of a surface-code quantum memory on hardware whose native connectivity does not match the code remains a central challenge. We address this on IBM heavy-hex superconducting processors by co-designing the code embedding and control: a depth-minimizing SWAP-based ``fold-unfold'' embedding that uses bridge ancillas, together with robust, gap-aware dynamical decoupling (DD). On Heron-generation devices we perform anisotropic scaling from a uniform distance $3$ code to anisotropic distance $(d_x,d_z)=(3,5)$ and $(5,3)$ codes. We find that increasing $d_z$ ($d_x$) improves the protection of $Z$-basis ($X$-basis) logical states across multiple quantum error correction cycles. Even if global subthreshold code scaling for arbitrary logical initial states is not yet achieved, we argue that it is within reach with minor hardware improvements. We show that DD plays a major role: it suppresses coherent $ZZ$ crosstalk and non-Markovian dephasing that accumulate during idle gaps on heavy-hex layouts, and it eliminates spurious subthreshold claims that arise when scaled codes without DD are compared against smaller codes with DD. To quantify performance, we derive an entanglement fidelity metric that is computed directly from $X$- and $Z$-basis logical-error data and provides per-cycle, SPAM-aware bounds. The entanglement fidelity metric reveals that widely used single-parameter fits used to compute suppression factors can mischaracterize or obscure code performance when their assumptions are violated; we identify the strong assumptions of stationarity, unitality, and negligible logical SPAM required for those fits to be valid and show that they do not hold for our data. Our results establish a concrete path to robust tests of subthreshold surface-code scaling under biased, non-Markovian noise by integrating QEC with optimized DD on non-native architectures.
\end{abstract}

\maketitle

\section*{Introduction}

The discovery of quantum error correction (QEC)~\cite{Shor1995PRA,Steane:96a,PhysRevA.54.1098} enables fault-tolerant quantum computing~\cite{DiVincenzo:96,Aharonov:96,10.5555/1972505}, 
which ``fights entanglement
with entanglement''~\cite{Preskill:99} by distributing quantum information among multiple physical qubits, and can 
safeguard long computations against the accumulation of errors. Numerous experiments have implemented 
various QEC codes~\cite{Campbell:2017aa}. 
Such $[[n,k,d]]$ codes employ $n$ physical qubits to encode $k$ logical qubits and can correct any errors on up to $(d-1)/2$ of the physical qubits, where the code distance $d$~\cite{Knill:1997kx} is odd. Topological codes~\cite{Kitaev_2003}, which combine local error-syndrome checks with global information encoding, form a particularly promising class. 
Perhaps the best known example from this class is the surface code~\cite{bravyi1998quantumcodeslatticeboundary,10.1063/1.1499754,Fowler:2012ys}, 
which has recently been the subject of intense experimental attention~\cite{Krinner2022Nature,GoogleAI2023QEC,Bluvstein2023Nature,GoogleAI2024Nature,HetenyiPRXQ2024,Bluvstein2025arxiv}.
The surface code uses $n=d^2$ data qubits and $d^2$ ancilla qubits to encode $k=1$ logical qubits, and leads to a relatively high error threshold~\cite{PhysRevA.89.022321} 
and simple fault-tolerant implementation~\cite{PhysRevA.90.062320} in comparison to other codes. A more general version of the surface code we use in this work has two different distance parameters, $d_x$ and $d_z$ for bit and phase-flip errors, respectively, and $2d_xd_z$ physical qubits.  

When the error rate is below threshold, growing the code distance leads, in theory, to exponential error suppression~\cite{10.1063/1.1499754}.
The consequences of subthreshold distance-scaling were demonstrated only very recently by Google using a superconducting quantum processing unit (QPU): the error rate of a logical qubit was shown to break even as $d$ increased from $3$ to $5$~\cite{GoogleAI2023QEC} and, subsequently, to exhibit a persistent twofold improvement from $d=3$ to $5$ and then from $5$ to $7$~\cite{GoogleAI2024Nature}, in line with the expected exponential improvement. Other
subthreshold demonstrations of surface and color codes followed, using superconducting qubits~\cite{GoogleAI2024arxivDynamicSurfaceCode,GoogleAI2024arxivColorCode} or Rydberg atoms~\cite{Bluvstein2025arxiv}. 
A critical aspect of achieving 
these milestones with Google's QPUs is the chip's qubit connectivity graph: a 2D square lattice of transmon qubits designed to precisely mirror the connectivity of the surface code. This allowed for the parallelization of nearly all operations in the QEC cycles, thus minimizing idle gaps (i.e., periods during which a qubit undergoes no gates) and the associated accumulation of errors.

\begin{figure*}[ht]
{\includegraphics[width=.9\textwidth]{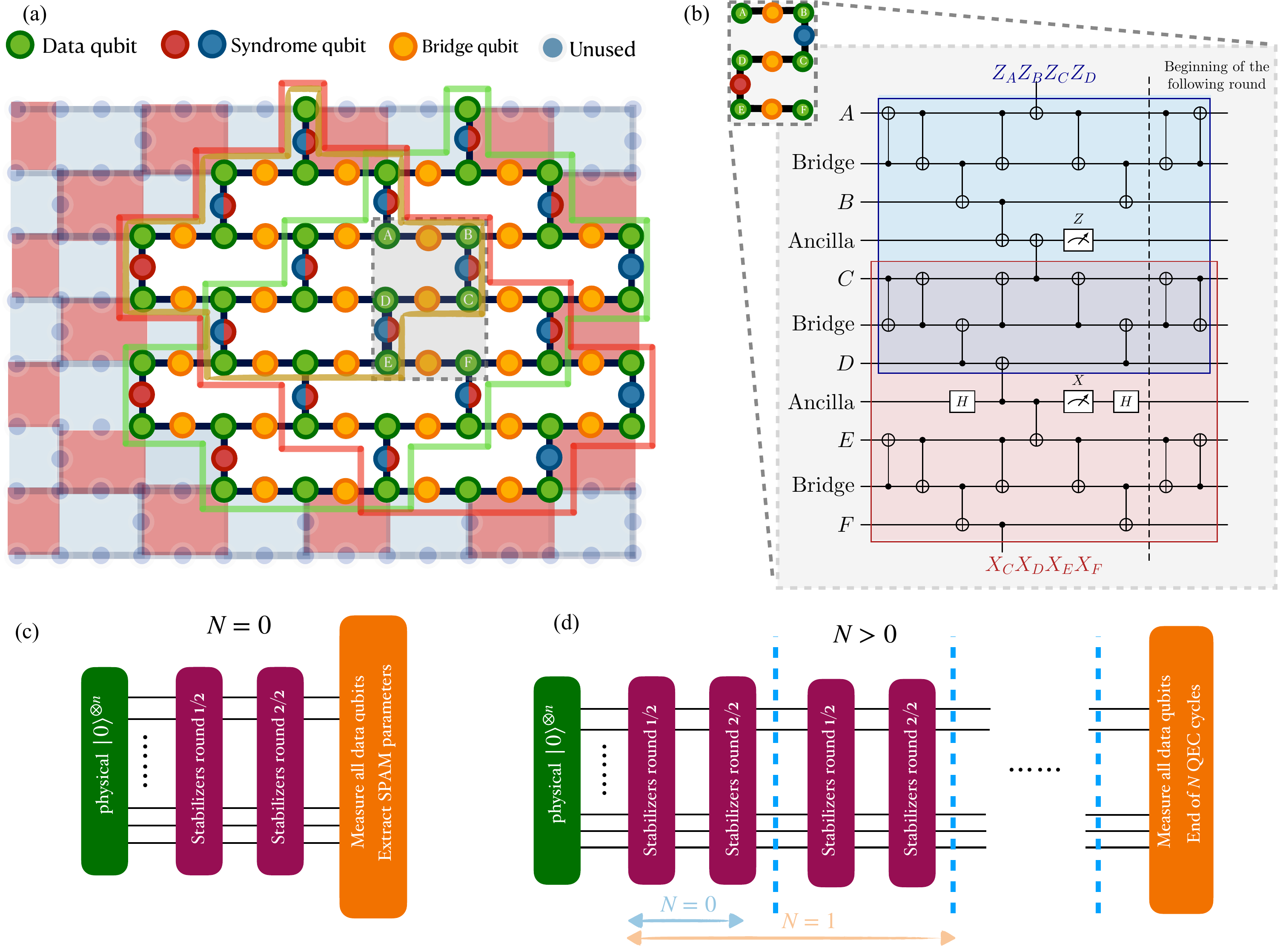}} 
\caption{(a) Schematic of the surface codes implemented in this work on IBM's Heron-generation QPUs. The red boundary defines the \dxdz{3}{5} surface code. The yellow boundary shows one of the possible sublattice \dxdz{3}{3} codes that fits within the larger code. The other two \dxdz{3}{3} sublattices are not shown. The green boundary defines the \dxdz{5}{3} surface code. Similarly, there are three \dxdz{3}{3} sublattices that fit within the \dxdz{5}{3} code (not shown). (b) Stabilizer measurement construction for a heavy-hex surface code. The circuit combines the middle-out and traditional ancilla-based syndrome extraction schemes~\cite{McEwen2023quantum}, adapted to the heavy-hex connectivity~\cite{Benito2025quantum}. Next-nearest-neighbor CNOTs are implemented via bridge qubits using SWAP gates, and the resultant circuit is then simplified. Bridge qubits are not measured. Ancilla qubits are never reset; instead, we apply software processing to measurement outcomes depending on previous measurements. The bridge ancilla is not measured; all measurements are on the check ancilla. (c) $N=0$: logical state preparation via two stabilizer rounds, starting from the physical $\ket{0}^n$ state, followed by a measurement of data qubits. The goal of this cycle is SPAM calibration. (d) $N\ge 1$: QEC cycles, each comprising two stabilizer rounds, followed by a measurement of all the data qubits at the end of the $N$'th cycle. Only half the stabilizers are measured in parallel; a full cycle is two rounds.} 
\label{fig-one}
\end{figure*}

Implementing the surface code on other QPUs with different fixed connectivities is a problem of both fundamental interest and practical importance. In light of 
applications, 
QPU design 
may be driven by considerations other than optimal surface code performance, e.g., to avoid frequency crowding affecting crosstalk and gate performance~\cite{PhysRevLett.107.080502} and to reduce the density of circuit elements, easing thermal management in lithographic fabrication. For example, IBM's superconducting QPUs are based on the heavy-hex lattice [illustrated in \cref{fig-one}(a)], in which transmon qubits are arranged on the sites and links of a honeycomb lattice.
An experimental demonstration of subthreshold scaling on heavy-hex 
devices presents significant challenges. 
In particular, the reduced connectivity can lead to substantial delays to accommodate state transfer between non-neighboring qubits.
The exposure to
noise during these additional idle gaps 
hinders demonstrations of subthreshold scaling~\cite{HetenyiPRXQ2024}, suggesting that it is crucial to 
optimize the embedding strategy to minimize the required circuit depth. The lightweight SWAP-based embedding devised in Ref.~\cite{Benito2025quantum} was shown theoretically, on the basis of a multi-parameter, Markovian (stochastic Pauli channels) circuit-level noise model, to outperform other QEC codes better adapted to the heavy-hex architecture, as well as alternative flag-qubit-based embedding protocols. Thus, we implement this embedding here.

\subsection*{Main results}

Here we demonstrate subthreshold scaling in the heavy-hex lattice as we independently increase the code's distance parameters $d_x$ or $d_z$ for bit- and phase-flip error types. We do not find compelling evidence of basis-independent subthreshold scaling, as the logical  states in a basis orthogonal to the direction of  growth  only suffer from the increased circuit complexity without benefiting from any increased protection. A crucial ingredient in our demonstration is dynamical decoupling (DD)~\cite{Viola:98,Viola1999PRL,Zanardi1999fk,Vitali:99}, a technique that
is particularly well suited for the suppression of coherent and non-Markovian noise, e.g., due to residual coherent coupling between qubits, or due to a temporally correlated environment with a structured spectral noise density~\cite{LidarBrun2013QEC}. 
We use robust DD variants~\cite{Quiroz2013PRA,Genov2017PRL} which we tailor 
for our surface code 
quantum memory, optimizing their integration with the SWAP-based embedding on IBM's $156$-qubit `Heron' QPUs. Although these devices are $6$ qubits short of allowing a fully isotropic scaling from $d=3$ to $d=5$, they do allow for scaling from a $37$-qubit embedding of the $d=3$ surface code to $65$-qubit embeddings of both $(d_x,d_z)=(3,5)$ and $(d_x,d_z)=(5,3)$ codes. Below threshold, these anisotropic code scalings should endow the encoded logical information with improved protection  to either bit- or phase-flip errors, a necessary condition for entering the subthreshold isotropic scaling regime. Our confirmation of this expectation constitutes our main result. 

We report our quantum memory results for up to $10$ full cycles of QEC on~\aachen, reaching depths in excess of $140$, and a total of $2200$ entangling gates. A discussion of similar experiments using other Heron QPUs is given in \cite{supp} (\cref{sec:other_qpus_supp}).
We extract the logical error probabilities 
after encoding logical states in the computational, i.e., $Z$-basis ($\overline{\ket{0}}$, $\overline{\ket{1}}$)
or Hadamard, i.e., $X$-basis ($\overline{\ket{+}}$, $\overline{\ket{-}}$), applying repetitive QEC cycles, and performing a final projective measurement of the data qubits. 
 
Note that, in contrast to the $X$- and $Z$-basis, 
the surface-code preparation of $Y$-basis logical states 
is significantly more complex, as it cannot be achieved directly within the code’s stabilizer structure and instead requires costly state injection~\cite{Fowler:2012ys} 
or  twist defect braiding~\cite{PhysRevX.7.021029,Gidney2024inplaceaccessto}.
The inaccessibility of all six cardinal logical states prevents a complete model-independent estimate of the average channel fidelity~\cite{Nielsen:2002aa}, and other more phenomenological metrics have been proposed. A widely used metric is the suppression factor $\Lambda$, calculated as the ratio of logical errors per QEC cycle of two codes of successive odd distances, each averaged over the available $X$ and $Z$ basis states, typically just $\overline{\ket{0}}$ and $\overline{\ket{+}}$~\cite{GoogleAI2023QEC,GoogleAI2024Nature,GoogleAI2024arxivColorCode,Bluvstein2025arxiv}. 

We critically assess whether this basis-averaged suppression factor, which is furthermore based on a phenomenological one-parameter fit that assumes stationary (cycle-independent) errors, functions as a reliable and robust performance metric. 
In the case of IBM QPUs, we show that this is not the case and provide a first-principles derivation of an alternative metric based on the entanglement fidelity (EF)~\cite{Nielsen:2002aa} that, under fairly general logical error model assumptions (clarified below), allows for a faithful and accurate characterization of our QEC experiments. We use this metric to identify a three-parameter fitting model that separates SPAM from the QEC per-cycle logical error, and use it to demonstrate a reduction in logical error probabilities as we  integrate our SWAP-based embedding with  robust DD while scaling to the two larger surface codes. 

In all QPUs used in this work, we identify non-Markovian dephasing and coherent $ZZ$-crosstalk as the main sources of hardware errors, leading to a dominance of $Z$-type logical errors.
We demonstrate that the incorporation of robust DD in QEC circuits is essential for subthreshold performance. DD suppresses coherent errors, which are underestimated by randomized benchmarking~\cite{tripathi2024DB}.  
Additionally, DD suppresses incoherent non-Markovian errors occurring during the extra idle gaps of the SWAP-embedding circuits, 
which are particularly biased towards single- and two-qubit dephasing. 
DD suppresses these errors to a level where the surface code can correct the remaining Markovian errors, and to the point where the benefits of anisotropic surface code scaling become more important. 

We also find that claiming genuine subthreshold scaling requires careful consideration of DD effects. Without DD, logical error rates can decrease with increasing code distance, apparently demonstrating subthreshold behavior. However, when DD is applied, this trend may reverse: although DD reduces errors across all code distances, it can benefit smaller codes disproportionately, eliminating the apparent scaling advantage of larger codes. This reveals the possibility of \emph{spurious} claims of subthreshold scaling: behavior that appears subthreshold without DD but ceases to be so when DD is properly included.
Thus, our results offer insight into the interplay between QEC and error suppression, and the importance of fully optimizing the latter before a claim of subthreshold scaling is made. 
The non-native heavy-hex architecture is key to unveiling these results, since native surface code implementations parallelize operations and leave no idling gaps that would lead to this subtle interplay of QEC and DD.

The findings presented here constitute a key step in the demonstration of subthreshold surface code scaling in non-native architectures, paving the way for future experiments with near-term QPUs and larger codes.

\section*{Methods}

\subsection{Surface code on heavy-hex QPUs} 
The distance parameter $d$ of an $[[n,k,d]]$ stabilizer code can be characterized as the smallest possible weight of physical Pauli operators appearing in the logical operators ${X}_L,{Z}_L$~\cite{Gottesman:97a}. It is useful to distinguish between error types in this context and, as mentioned above, introduce different distance parameters, e.g., $d_x$ and $d_z$ for Pauli $X$ and $Z$ errors, respectively. As in the surface code, this distinction can be particularly relevant for CSS codes~\cite{Steane:96a,PhysRevA.54.1098}, as these errors are in practice corrected independently. From a practical standpoint, current Heron QPUs 
can support a SWAP-based embedding of exactly three odd-distance such surface code sizes: $(d_x,d_z)=(3,3)$, $(3,5)$, and $(5,3)$ (see \cref{fig-one}), allowing exploration of scaling even though the $(5,5)$ code is out of reach. Henceforth, we use $\bd\equiv (d_x,d_z)$ to denote these codes. 

The operation of the surface code as a quantum memory involves encoding, syndrome extraction, error decoding, and recovery operations. We next describe the syndrome extraction step, i.e., the measurement of plaquette and vertex stabilizers of the surface code. 
Ref.~\cite{Benito2025quantum} devised a SWAP-based protocol of this step on the heavy-hex lattice, which was found to be the best performing strategy compared to several other embeddings and promising native codes. Instead of using a naive CNOT-gate compilation of SWAP operations~\cite{10.5555/1972505}, which must be applied before and after the syndrome-extraction CNOTs to transfer the state of the non-neighboring data qubits that form a specific stabilizer, the SWAP-based protocol uses a folding-unfolding scheme minimizing circuit depth.

The circuits initially follow the middle-out ancilla-free measurement scheme described in Ref.~\cite{McEwen2023quantum}, and start by applying a first layer of CNOT gates between data qubits to fold weight-4 stabilizers into weight-2 operators. The middle-out protocol would then apply another CNOT layer to further fold these 
to one-body operators, which can then be measured. Instead, in light of the heavy-hex layout, it is better to use an ancilla qubit to measure the weight-2 parity check directly. Afterwards, the weight-2 checks are unfolded back into the original weight-4 stabilizers by a sequence of CNOTs. Denoting $A$-$F$ as the bulk data qubits (see \cref{fig-one}), the stabilizers $Z_AZ_BZ_CZ_D$ and $X_BX_CX_EX_F$ are mapped to $Z_BZ_C$ and $X_EX_F$ during the first step, respectively. Then $Z_BZ_C$ and $X_EX_F$ are measured using an ancilla qubit. Finally, the original stabilizers are restored (unfolded) by using consecutive CNOTs. This step can be further optimized by executing the first two steps of the following QEC round with the last ones from the previous round (see \cref{fig-one}), minimizing the overall circuit depth~\cite{Benito2025quantum}. 
We further optimize the circuits  by removing reset gates, instead tracking the previous measurement outcomes in software~\cite{Geher2025npj}. Removing the resets significantly shortens the duration of the syndrome extraction round, thus also reducing idling errors. As most CNOT gates that implement the SWAP operation cancel out, the result is a compact depth-$7$ circuit, shown in \cref{fig-one}(b). 
Note that in this scheme, only half of the stabilizers can be measured in parallel, requiring two rounds to perform a full syndrome extraction cycle. 

We note that Ref.~\cite{HetenyiPRXQ2024} ran syndrome extraction circuits by effectively connecting next-nearest-neighbor qubits, which required translating each CNOT gate into a sequence of four consecutive CNOT gates. This contributes to overhead and partly explains why their surface code experiments did not achieve subthreshold behavior. The lack of optimized DD is likely another major reason. A noteworthy aspect of Ref.~\cite{HetenyiPRXQ2024} is the interesting idea of recycling the spare qubits for an additional code, allowing for novel fault-tolerant entangling gates.

\subsection{Experimental surface code memory with dynamical decoupling}

Investigating subthreshold surface code scaling in the heavy-hex lattice amounts to comparing the largest embeddable code with smaller codes restricted to sublattices of the same superconducting chip. In \cref{fig-one}(a) we show the layouts used 
in this work, which include the \dxdz{5}{3} and \dxdz{3}{5} surface codes, as well as three patches that can host the \dxdz{3}{3} code in a sublattice. 
We note that the \dxdz{3}{5} and \dxdz{5}{3} codes grow in different spatial directions, 
thus involving different sets of transmon qubits and couplers, with varying $T_1$/$T_2$ times and native gate errors.
This prevents us from making a fair comparison between them. Instead, we benchmark each anisotropic code against both the average of the three \dxdz{3}{3} code patches it includes and the top-performing patch.
We also note that it is possible to swap the $X$ and $Z$ stabilizers to implement an alternative, conjugate set of codes. As is clear from \cref{fig-one}, this effectively corresponds to reversing the chirality of each code and using the opposite set of qubits, i.e., using the red layout for the \dxdz{5}{3} code and the green layout for \dxdz{3}{5}. We also performed experiments with these conjugate codes. The results were very similar to those we report below, and are not shown.

We perform a quantum memory experiment by initializing each code in four logical states, $\overline{\ket{\alpha}}$, $\alpha\in\{0,1,+,-\}$, which requires preparing the data qubits in either $\otimes_n\ket{0}$ for $\alpha\in\{0,1\}$ or $\otimes_n\ket{+}$ for $\alpha\in\{+,-\}$. Additionally, we apply the logical operator $X_L$ or $Z_L$ to prepare $\ket{1}$ and $\ket{-}$, respectively.
We proceed by applying $N$ QEC cycles, each consisting of two consecutive rounds of SWAP-based folding/unfolding circuits for all $X$- and $Z$-type stabilizers, yielding a collection of bitstrings that specify which stabilizers have flipped at a particular QEC cycle. This is illustrated in \cref{fig-one}(c) and (d). For the first cycle, which we denote by $N=0$, the $X$-type (for $\overline{\ket{0}}$ and $\overline{\ket{1}}$) or $Z$-type (for $\overline{\ket{+}}$ and $\overline{\ket{-}}$) stabilizers project the prepared qubits into the code space, and their random results are used to define the initial Pauli frame, as well as to extract and calibrate the SPAM errors [\cref{fig-one}(c)]. The subsequent cycles [\cref{fig-one}(d)] project the state into a possibly different stabilizer subspace specified by a bitstring that determines the corresponding $\pm 1$ values of the plaquette and vertex operators, providing an evolving error syndrome that
must then be fed to a decoder, as explained in more detail below. This decoder
estimates the most likely errors
and suggests a recovery correction that can be applied at the software level by a final Pauli-frame update that effectively defines the corrected logical operators. 
After $N$ such QEC cycles, we projectively measure the data qubits in the corresponding basis, counting logical failures in light of the final Pauli frame for the $\overline{\ket{0}}$ and $\overline{\ket{1}}$ ($\overline{\ket{+}}$ and $\overline{\ket{-}}$) states. This gives access to the resulting 
$X$-type ($Z$-type) logical errors, estimating the error probabilities $\overline{p}_{N,0}$ and $\overline{p}_{N,1}$ ($\overline{p}_{N,+}$ and $\overline{p}_{N,-}$) through the relative frequencies. We perform up to $N_{\max}=9$ cycles, totaling approximately $80~\mu$s, which corresponds to $20$ embedded folding-unfolding rounds of surface code stabilizer readout. 

We make heavy use of DD to suppress both coherent and incoherent errors~\cite{Viola:98,Viola1999PRL,Zanardi1999fk,Vitali:99}. The main goal is to reduce the burden on the surface code, leaving it to predominantly handle just the incoherent Markovian errors that DD cannot suppress.
To do so, we optimize performance over a set of DD sequences vetted in the context of IBM QPUs~\cite{Ezzell2021PRApp}, and primarily chosen from the robust pulse sequence family known as ``universally robust'', or UR$_m$~\cite{Genov2017PRL}. These sequences are designed to be resilient to pulse imperfections, including over/under-rotation and axis angle errors. The integer $m$ corresponds to the number of pulses, and as long as only the modeled errors occur, DD can reduce the infidelity to  $\epsilon^m$ by increasing $m$, where $\epsilon$ is a measure of pulse error. In reality, however, longer sequences of imperfect pulses may accumulate unmodeled errors, which can actually overwhelm the suppression of those that are modeled, leading to overall performance degradation. 
Additionally, in contrast to typical QEC experiments where DD is applied only to idling data qubits during readout and resets, our circuits also include idling bridge and ancillary qubits. Accordingly, we apply DD to all qubits, regardless of their role. However, uniform application of DD across all gaps also degrades performance.
Hence, it is important to search for the optimal number of pulses $m$ per sequence, which will typically depend on the specific idle gap in the circuit, as these have different durations and cannot all accommodate a unique DD sequence. We augment the UR$_m$ family with two other robust sequences for the shorter idle gaps, $\mathrm{XY4}$~\cite{Maudsley1986ty,Viola1999PRL} and $\mathrm{RGA8}_a$~\cite{Quiroz2013PRA}, and find that the set $\{\mathrm{UR}_m,\mathrm{XY4},\mathrm{RGA8}_a\}$ with $m\in\{6,8,\dots,18\}$ results in the top-performing DD strategy for our surface code quantum memory. We have devised an iterative algorithm for identifying such a strategy that optimizes the filling of idle gaps from this set in order of decreasing time duration (see \cref{sec:dd_opt_supp} in ~\cite{supp}). As shown below, this approach results in a substantial improvement in surface code performance relative to not using DD, or to using unoptimized DD variants, such as uniformly filling idle gaps with non-robust CPMG sequences~\cite{CPMG1958}.
The inclusion of DD results in the addition of up to $\approx 484$ single-qubit DD pulses per cycle used to fill idle gaps (see~\cite{supp}, \cref{sec:dd_opt_supp} for details).

\subsection{Calibrated non-uniform noise model and decoding }

Our experiments are decoded using Stim \cite{gidney2021stim} and the minimum-weight perfect-matching (MWPM) decoder~\cite{Higgott2025quantum}, which define the logical observable for each measurement shot resulting from the corresponding Pauli frame updates. A logical failure is recorded when this observable disagrees with the measured logical observable, as this signals that the quantum memory has been altered. The logical error probability $\overline{p}_{N,\alpha}$ is then estimated as the fraction of failures over all shots.

In order to predict the logical observable, the decoder must be provided with a circuit-level noise model. In our case, this is a multichannel Markovian noise model that includes depolarizing channels for the gates;
biased dephasing channels that depend on the $T_1$ and $T_2$ values adapted to the different idling times; and bit-flip channels for the measurements (see \cref{sec:noise_model_decod_supp} in~\cite{supp} for details). 
The rates of these channels are obtained from the calibration data of each qubit and qubit pair and display a characteristic non-uniformity (rates depend on per-qubit $T_1/T_2$, gate fidelities, etc.), as each qubit and coupler acquire different specifications during  fabrication.

Given the use of optimized, robust DD sequences that suppress pulse angle and axis errors to a certain order, our decoder assumes ideal single-qubit gates. When our decoding includes DD, we report the performance based on the experimentally-optimized DD for each data point.


\subsection{Suppression factor metric for assessing subthreshold scaling}

To quantify evidence of subthreshold behavior, we begin with the conventional suppression factor between consecutive distances,
$\Lambda_{\varepsilon}^{(d)} \equiv \frac{\varepsilon_{d}}{\varepsilon_{d+2}}$,
where $\varepsilon_{d}$ denotes the average logical error probability per QEC cycle at code distance $d$ \cite{GoogleAI2021Nature,GoogleAI2023QEC,GoogleAI2024Nature,GoogleAI2024arxivColorCode,Bluvstein2025arxiv}. Under standard threshold-model assumptions, one expects an approximately exponential decay of the logical error with distance
$\varepsilon_{d} \approx C e^{-\kappa d}$, implying a supression factor $\Lambda_{\varepsilon}^{(d)} \approx e^{2\kappa}$ that is approximately constant in $d$ and strictly greater than $1$ in the subthreshold regime. Equivalently, a $d$-independent suppression factor $\Lambda_{\varepsilon}>1$ over a range of distances is consistent with exponential suppression of the per-cycle logical error probability with distance, with rate parameter $\kappa=\frac{1}{2}\ln\Lambda_{\varepsilon}$. If such behavior persists to larger $d$, it indicates scalable suppression that, in principle, could enable practically relevant algorithms on  large QPUs (megaquop \cite{Preskill-Megaquop} or even teraquop \cite{gidney2025factor2048bitrsa} regimes).
In current quantum-memory experiments, reported suppression factors  lie in the range $\Lambda_{\varepsilon}^{(d)}\sim [1.56,2.15]$ up to $d=7$ \cite{GoogleAI2024Nature,Bluvstein2025arxiv}. 

To estimate $\varepsilon_d$, the typical approach is to fit the  error probability $\overline{p}^{\bd}_{N,\alpha}$ as a function of the number of QEC cycles $N$ to a phenomenological single-parameter  model 
\beq
\label{eq:single_param_error_model}
\overline{p}^{d}_{N,\alpha} = \frac{1}{2}\big(1 -\bigl(1 -2\varepsilon^{d}_{\alpha}\bigr)^{\!N}\big),
\eeq
and then proceed by averaging the results over possible sublattices and logical basis states with label $\alpha$ to provide a single per-cycle error rate $\varepsilon_{d}$ for at least two consecutive code distances~\cite{GoogleAI2024Nature}. This single-parameter error model can be analytically derived from a purely combinatorial calculation, assuming that the quantum memory has a constant logical error probability per cycle $\varepsilon^{d}_{\alpha}$, and that only odd numbers of logical errors can contribute to a logical failure~\cite{O’Brien2017}. An alternative derivation uses an ad hoc difference equation for the logical error probability update rule after each cycle~\cite{GoogleAI2023QEC}. As we discuss in detail below, \cref{eq:single_param_error_model} is valid only under stationarity (no cycle dependence), unitality (non-increasing purity), and no SPAM errors in the logical error channel. Hence, applying \cref{eq:single_param_error_model} without ensuring that the model assumptions are satisfied may yield unfaithful results.

\subsection{Entanglement fidelity metric for assessing subthreshold scaling}
\label{sec:EF}

As mentioned in the Introduction, a  first-principles quantum memory metric can be derived from the entanglement fidelity (EF) $\Fe = \bra{\phi}
\bigl(\mathcal{I}\otimes\Psi \bigr) \bigl(\ketbra{\phi}\bigr)\ket{{\phi}}$, where $\Psi$ is a quantum channel, i.e., a completely-positive trace-preserving (CPTP) map, 
$\ket{\phi}= \frac{1}{D^{1/2}}\sum_{j=0}^{D-1}\ket{{j}} \otimes\ket{{j}}$
is a maximally entangled state in a $D^2$-dimensional Hilbert space, and $\mathcal{I}$ is the identity channel. $\Fe$ is directly related to the average channel fidelity $F_{\rm ave}=\int d\psi \bra{\psi} \Phi (\ketbra{\psi}) \ket{\psi}=(\Fe+1/D)/(1+1/D)$, and measures how well the channel preserves an arbitrary state. For a qubit, $\Fe = \frac{1}{4}\sum_i \lvert \Tr(E_i)\rvert^2$, where $\{E_i\}$ are $\Psi$'s Kraus operators. For a logical qubit with a Markovian error model described by a CP-divisible map~\cite{PhysRevLett.105.050403}, such that one simply composes the quantum channels for each cycle to obtain the full noisy evolution, calculating the EF requires estimating the logical error probabilities per cycle of the six cardinal states $\{\overline{\ket{0}},\overline{\ket{1}},\overline{\ket{+}},\overline{\ket{-}},\overline{\ket{+{\rm i}}},\overline{\ket{-{\rm i}}}\}$~\cite{Nielsen:2002aa} (see \cref{sec:ent_fid_supp} in~\cite{supp}). 

However, to date, surface code experiments only give access to logical error probabilities in the $X$ and $Z$ basis $\{\overline{p}_{N,0}^{\bd},\overline{p}_{N,1}^{\bd},\overline{p}_{N,+}^{\bd},\overline{p}_{N,-}^{\bd}\}$, corresponding to logical $X$ and $Z$ errors, respectively, since $Y$ basis preparation and readout require significantly more complicated methods~\cite{Gidney2024inplaceaccessto}. Nevertheless, under reasonable assumptions about the logical error channel, we can derive compact expressions for the EF that go beyond a simple averaging of the basis-dependent logical errors and approximate the contributions from the $Y$-basis states.
Our model describes the logical quantum channel as consisting of independent logical Pauli $X$ and $Z$ errors, logical amplitude damping, and logical SPAM errors, and assumes a CP‑divisible (cycle-composable) logical channel across cycles. Under this model, we derive expressions for the EF $\Fe^{\bd}(N)$, as well as rigorous upper and lower bounds so that $F_{{\rm e},<}^{\bd}(N)\leq \Fe^{\bd}(N)\leq F^{\bd}_{{\rm e},>}(N)$, which solely use the decoder-provided (measured) probabilities (see \cref{sec:ent_fid_supp} in~\cite{supp}):

\begin{align}
\label{eq:ent-fidelity}
\Fe^{\bd}(N)
=F_{{\rm e},<}^{\bd }(N) +\Bigl(\frac{1}{\Gamma^{(N)}}-1\Bigr)r(N),
\end{align}
where $r(N)\equiv \frac{\Sigma^{\bd}_x(N) \Sigma^{\bd}_z(N)}{4\Sigma_x^{\bd}(0)\Sigma_z^{\bd}(0)}$ and 
\bes
\begin{align}
\label{eq:Fe_bounds-lower}
F_{{\rm e},<}^{\bd }(N)&=\frac14\left(1+\frac{\Sigma^{\bd}_x(N)}{\Sigma_x^{\bd}(0)}\right)\left(1+\frac{\Sigma_z^{\bd}(N)}{\Sigma_z^{\bd}(0)}\right),\\
\label{eq:Fe_bounds-upper}
F_{\rm e,>}^{\bd }(N) &= F_{\rm e,<}^{\bd}(N) + \left(\frac{\Sigma_z^{\bd}(0)}{\Sigma_z^{\bd}(N)} - 1\right)r(N),
\end{align}
\ees
Here, $\Sigma_x^{\bd}(N)=1-\overline{p}^{\bd}_{N,+}-\overline{p}^{\bd}_{N,-}$ and $\Sigma_z^{\bd}(N)= 1-\overline{p}^{\bd}_{N,0}-\overline{p}^{\bd}_{N,1}$. Their values at $N=0$ account for SPAM renormalization, which we calibrate at the very first cycle. The factor $\Gamma^{(N)}$ accounts for logical amplitude damping and is discussed in detail in \cref{sec:ent_fid_supp} in~\cite{supp}; it cannot be estimated based purely on $X$ and $Z$-basis data, which is why we require the upper and lower bounds.

The EF has a simple interpretation: it lower-bounds the probability that after $N$ steps the cumulative Pauli error is the identity, i.e., $1-\Fe^{\bd}(N)$ upper bounds the probability of a logical error. Further, we note that the lower bound improves upon a baseline four‑state truncation of the EF, $F^{\bd}_{\rm e, trunc}(N)=(1+\Sigma^{\bd}_x(N)/\Sigma^{\bd}_x(0)+\Sigma^{\bd}_z(N)/\Sigma^{\bd}_z(0))/4$~\cite{Nielsen:2002aa}, by incorporating $Y$-basis results modeled as the uncorrelated concatenation of $X$ and $Z$ errors, yielding the above expression of $F^{\rm low}_{\rm e}(N)$. The upper bound $F^{\rm high}_{\rm e}(N)$ assumes the most  non‑unital channel consistent with the $X/Z$ tallies. 

The EF naturally lends itself as a metric for assessing subthreshold scaling: comparing a family codes with different $\bd$ values can be done by computing their respective $\Fe^{\bd}(N)$'s. Let $\tilde{F}_{\mathrm{e}} \equiv \max_{\{\text{DD,noDD}\}}\Fe$ so that each code is separately optimized with respect to the inclusion of DD. If $\tilde{\Fe}^{(d_x,d_z)}(N) > \tilde{\Fe}^{(d'_x,d'_z)}(N)$ for all $N$, where $d_x>d'_x$ or $d_z>d'_z$, or if the ratio of infidelities
\begin{align}
\label{eq:EF-metric}
    \Lambda_F^{(\bd,\bd')}(N) \equiv \frac{1-\tilde{F}_{\mathrm{e}}^{\bd'}(N)}{1-\tilde{F}_{\mathrm{e}}^{\bd}(N)} 
\end{align}
satisfies $\Lambda_F^{(\bd,\bd')}(N)>1$ for all $N$ under the same distance ordering, then this constitutes evidence of subthreshold scaling. Maximization over the inclusion or exclusion of DD ensures that each code is compared in its optimally configured form, an important point that we highlighted in the Introduction and discuss in more detail below. In applying \cref{eq:EF-metric} in the context of individual sublattices, we define the top-performing sublattice code as the code that has the highest fidelity averaged over all cycles. This is consistent with keeping the logical qubit fidelity as high as possible for as long as possible.

We advocate for the use of $\Lambda_F^{(\bd,\bd')}(N)$ in lieu of $\Lambda_{\varepsilon}^{(d)}$, as it is a fitting-model-free metric that, moreover, does not require the assumption of stationarity and can be computed directly from the logical error probabilities output by the decoder. In fact, it is by carefully considering the EF that we can identify the implicit assumptions underlying the suppression factor, and the required generalizations of fitting error models and logical basis averaging.

\subsection{Reduction to fitting models}

As shown in \cref{sec:ent_fid_supp}, when we assume (1) stationarity (the noise channel is independent of the cycle number $N$), (2) unitality (the noise consists purely of Pauli errors), and (3) no logical SPAM errors, the single-parameter error model \cref{eq:single_param_error_model} is recovered from first principles. The EF in this case coincides exactly with the lower bound for $\Sigma_x^{\bd}(0)=\Sigma_z^{\bd}(0)=1$ and, after expanding the EF to leading order in the logical error probabilities, gives $F^{\bd}_{\rm e}(N)=\frac14\big([1+\Sigma^{\bd}_x(N)][1+\Sigma^{\bd}_z(N)]\big)\approx 1-p^{\boldsymbol{d}}_N$, with $\overline{p}^{\boldsymbol{d}}_N=\frac{1}{4}\sum_{\alpha}\overline{p}^{\boldsymbol{d}}_{N,\alpha}$ as the logical error probability after $N$ cycles averaged over the four basis states. Note that each of these error probabilities is described by the single-parameter error model \cref{eq:single_param_error_model}, and that only for $N=1$ one recovers an expression with the per-cycle error rates averaged over basis states $\Fe^{\bd}(1)=1-\epsilon_{\boldsymbol{d}}$, where $\epsilon_{\boldsymbol{d}}=\frac14\sum_\alpha\epsilon^{\bd}_{\alpha}$. For generic $N$, it is not the averaged $\epsilon_{\boldsymbol{d}}$ which determines the EF and the average channel infidelity, but a linear combination of the different exponential decays per basis state, in addition to non-linear contributions due to the $Y$-basis states. 
When the specific assumptions above do not hold, the complete expression \cref{eq:ent-fidelity} gives a more informative metric.

The non-stationary, non-unital, SPAM-inclusive model used to derive \cref{eq:ent-fidelity} is actually crucial in our case, as our QEC experiments  show evidence of cycle-dependence as well as logical amplitude damping and SPAM errors. 
The same first-principles model can also be used to derive a generalization of \cref{eq:single_param_error_model} with three free but stationary (cycle-independent) parameters $\{a_{\alpha},b_{\alpha},\varepsilon^{\bd}_{\alpha}\}$ that can be used to fit the  logical error probability
\beq
\label{eq:3param-model}
\overline{p}^{\bd}_{N,\alpha} = a_{\alpha} + b_{\alpha}\bigl(1 -\varepsilon^{\bd}_{\alpha}/a_\alpha\bigr)^{N},
\eeq
where $\varepsilon^{\bd}_\alpha$ is the per-cycle error rate (a probability) and $a_{\alpha},b_{\alpha}$ account for non-unital and SPAM errors. Interestingly, the same result can also be derived from a physically motivated difference equation for the logical error probabilities, relating the free parameters to specific logical error sources (see \cref{sec:ent_fid_supp} in~\cite{supp}). The aforementioned reduction to the one-parameter model \cref{eq:single_param_error_model} by imposing additional assumptions is particularly clear when
\cref{eq:3param-model} is used as the starting point: neglecting SPAM errors
sets $b_{\alpha}=-a_{\alpha}$, and restricting to unital Pauli errors further fixes $a_{\alpha}=1/2$. For the IBM QPUs, we report different values for $a_{\alpha}$ and $b_{\alpha}$ below as required for a best-fit to \cref{eq:3param-model}, attesting that the logical error model is not a simple Pauli channel.

\section*{Results}

We now present our results for the SWAP-based embedding of the surface code using \aachen, the top-performing QPU in our study, and its integration with robust DD. Recall that the raw decoder outputs are the logical error probabilities after $N$ QEC cycles, i.e., $\overline{p}_{N,\alpha}^{\bd}$ from experiments initialized in the $\overline{\ket{\alpha}}$ state with $\alpha\in\{0,1,+,-\}$, which we analyze in light of the metrics discussed above. We set $N=0$ to denote the SPAM-calibration cycle. Since the $(3,5)$ code increases protection against phase‑flip ($Z$) errors and the $(5,3)$ code increases protection against bit‑flip ($X$) errors, we analyze (i) basis‑resolved $\overline{p}_{N,\alpha}^{\bd}$, (ii) the basis‑specific suppression factor $\Lambda_{\varepsilon,\alpha}^{\bd}$ extracted from fits, and (iii) a basis‑independent entanglement‑fidelity (EF) metric $\Fe^{\bd}(N)$ and its ratio $\Lambda_F$ defined in \cref{eq:EF-metric}. 

Among prior quantum memory experiments, Ref.~\cite{Sundaresan2023NatComm} prepared all four states; Refs.~\cite{GoogleAI2023QEC,GoogleAI2024Nature,GoogleAI2024arxivDynamicSurfaceCode} used only $\alpha\in\{0,+\}$; Ref.~\cite{Bluvstein2025arxiv} also prepared $\alpha\in\{1,-\}$ for neural-network decoder training but did not report memory results for these states.

\setcounter{subsection}{0}

\subsection{Logical error probabilities}

\begin{figure*}[ht]
\includegraphics[width=.9\textwidth]{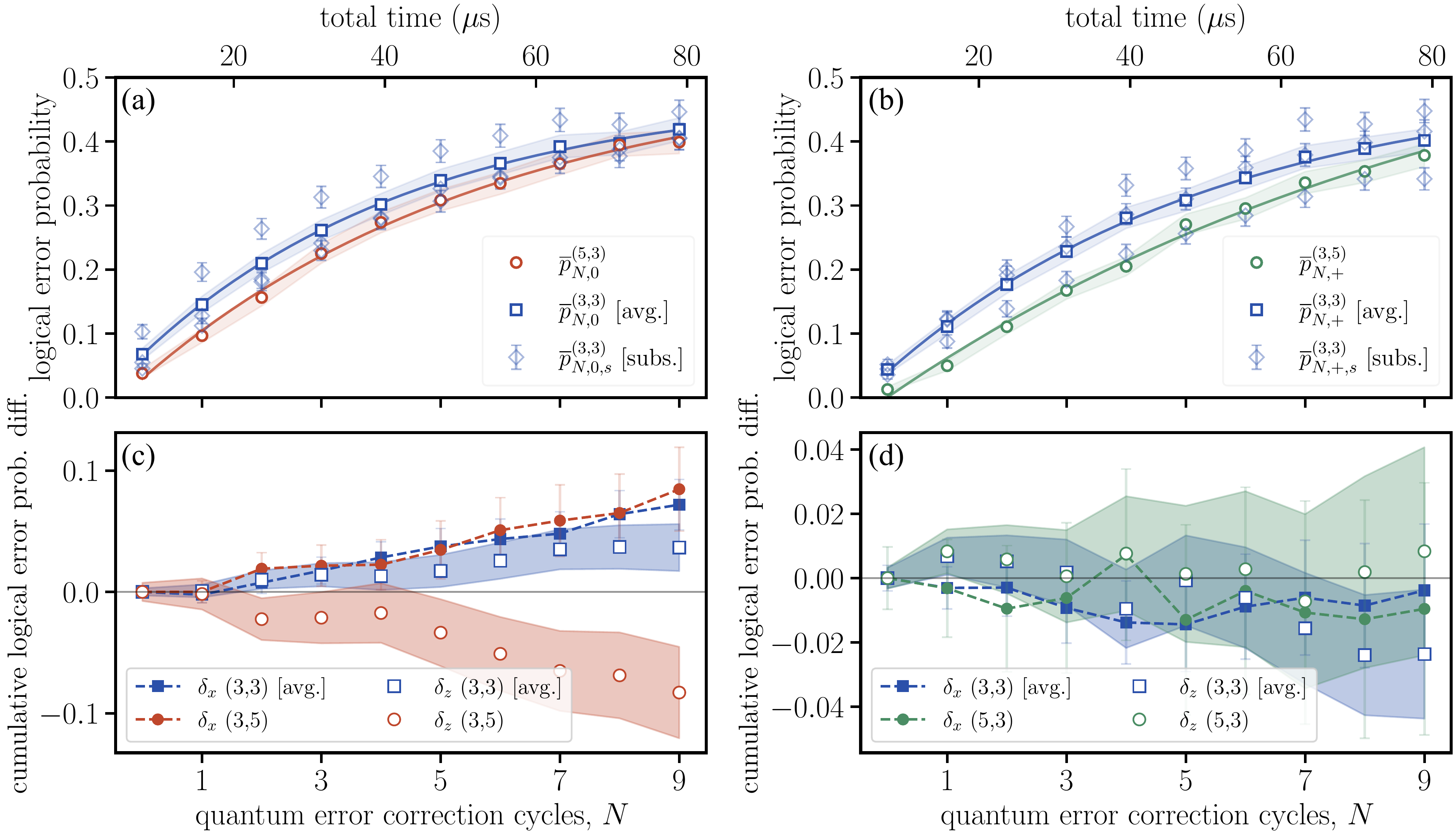}
\caption{Top: Logical error probability $\overline{p}^{\bd}_{N,\alpha}$ versus QEC cycles $N$, for (a) $\overline{\ket{+}}$ and (b) $\overline{\ket{0}}$, including optimized DD. Open blue markers (“subs.”) show the three $(3,3)$ sublattices; filled blue markers (“avg.”) show their arithmetic mean. Note that the $(3,5)$ and $(5,3)$ codes contain different physical $(3,3)$ sublattices, hence the sublattice points differ between (a) and (b). Shaded regions and error bars denote $2\sigma$ confidence intervals. With respect to these intervals, both $(3,5)$ and $(5,3)$ outperform the average $(3,3)$ code, and for the first few cycles also outperform the best individual $(3,3)$ sublattice. Overall, $(3,5)$ shows the strongest improvement. Solid curves are fits to \cref{eq:3param-model}.
Bottom: Cumulative differences $\delta^{\bd}_x=\sum_{N}(\overline{p}^{\bd}_{N,-}-\overline{p}^{\bd}_{N,+})$ vs. $\delta^{\bd}_z=\sum_{N}(\overline{p}^{\bd}_{N,1}-\overline{p}^{\bd}_{N,0})$ for (c) $(3,5)$ and (d) $(5,3)$, along with the average over their respective $(3,3)$ sublattices. Deviations from zero witness non‑unital logical noise, growing with $N$ and more pronounced for the $Z$ eigenstates.}
\label{fig:err-prob}
\end{figure*}

\cref{fig:err-prob} shows our experimental logical error probabilities $\overline{p}^{\bd}_{N,\alpha}$ versus the number of QEC cycles $N$ for the representative states $\alpha\in\{0,+\}$. These probabilities reflect control, measurement, and environment-induced errors, as well as decoder accuracy. For each data point we report the minimum $\overline{p}^{\bd}_{N,\alpha}$ found over a search of decoder noise-model parameters and DD configurations. Because the $(3,3)$ code can be embedded in three distinct sublattices, we report $\overline{p}^{(3,3)}_{N,\alpha,s}$ for each sublattice $s$, along with their average $\overline{p}^{(3,3)}_{N,\alpha}=\frac{1}{3}\sum_{s=1}^3 \overline{p}^{(3,3)}_{N,\alpha,s}$.

\cref{fig:err-prob}(a) shows that $\overline{p}^{(3,5)}_{N,+} < \overline{p}^{(3,3)}_{N,+}$ for all $N$, consistent with subthreshold scaling along the $Z$-error direction: the $(3,5)$ cycles provide increased protection against phase-flip errors relative to the averaged $(3,3)$ case. This is nontrivial because both circuit width and depth increase from $(3,3)$ to $(3,5)$, raising the opportunity for physical faults even as the larger code can correct more phase flips.

\cref{fig:err-prob}(b) shows the analogous behavior for $X$-type errors: $\overline{p}^{(5,3)}_{N,0} < \overline{p}^{(3,3)}_{N,0}$. Thus, up to $9$ QEC cycles, the anisotropically scaled $(5,3)$ and $(3,5)$ codes achieve lower logical error probabilities than the averaged $(3,3)$ code for the error type aligned with their increased distance. Because sublattice averaging can mask best-case performance, we also compare against the best $(3,3)$ sublattice below.

We obtain qualitatively similar results when we include $\overline{p}_{N,-}^{\bd}$ and $\overline{p}_{N,1}^{\bd}$. However, there is valuable additional information in these states: any persistent asymmetry within an $X$ or $Z$ eigenstate pair signals non‑unital logical noise. \cref{fig:err-prob}(c) and (d) shows the cumulative difference of the logical error probabilities for each pair of eigenstates. Although the differences are small, the cumulative effect is not unbiased but grows as the number of cycles increases. This is evidence that the effective logical noise is non-unital as it affects $+/-$ and $0/1$ eigenstates differently, and motivates accounting for non-unital effects in the logical noise model. Moreover, it shows that in our case, the three-parameter \cref{eq:3param-model} must be used instead of the purely unital, SPAM-free \cref{eq:single_param_error_model}.

\subsection{Suppression factor QEC metric}

The solid lines in \cref{fig:err-prob} are the results of using the three-parameter model in \cref{eq:3param-model} to fit our data, showing close agreement. To quantitatively assess whether three or fewer parameters are needed, we also consider two- and one-parameter models, thus accounting for situations in which SPAM or non-unital logical errors do not play a role. We use the Akaike Information Criterion (AIC)~\cite{Akaike:1974aa,Wagenmakers:2004aa}---a standard statistical tool that balances model complexity (the number of fitting parameters) with fit quality.  
We find that the optimal parameter number for our data is indeed three, confirming that both SPAM and logical non-Pauli errors play a role (see \cref{sec:fitting_supp} in~\cite{supp}). 

The fits of the logical error probabilities yield the error rates $\varepsilon_{\alpha}^{\bd}$ of \cref{eq:3param-model} for each basis state and code, as well as the average 
$\varepsilon_{\alpha}^{(3,3)} = \sum_{s=1}^3\varepsilon_{\alpha,s}^{(3,3)}/3$ when we consider the sublattices. From these values (reported in~\cite{supp}, \cref{appsec:AIC-assessment}), we compute the suppression factors $\Lambda^{\bd}_{\varepsilon,\alpha,s}\equiv\varepsilon_{\alpha,s}^{(3,3)}/\varepsilon_{\alpha}^{\bd}$ for each sublattice. The suppression factors obtained in this manner are reported in \cref{tab-res-aachen}, and we find that they lie above unity: $\Lambda_{\varepsilon,\alpha}^{(3,5)},\Lambda_{\varepsilon,\alpha}^{(5,3)}\sim[1.23,1.46]$ for both error scaling directions. This is on par with the suppression-factor ranges reported previously for isotropic scaling~\cite{GoogleAI2024Nature,Bluvstein2025arxiv}. 

Although averaging over sublattices is motivated by the fact that all the corresponding qubits and couplers will also be used in the larger codes, when solving a computational problem of a given size, it is more natural to aim for the best possible performance, i.e., to optimize over the available sublattices. 
We find that the improvement holds even from this perspective. Namely,  with respect to the best sublattice we obtain $\min_s\Lambda^{\bd}_{\varepsilon,\alpha,s}\in[1.01,1.10]$. This shows, for the first time, that current IBM QPUs support subthreshold scaling, albeit only with respect to \emph{anisotropic} growth of the code distance; we clarify below why a global subthreshold advantage is not yet achieved. With these caveats noted, to the best of our knowledge---apart from the results of Ref.~\cite{GoogleAI2023QEC} for roughly the same amount of QPU time ($\sim80\mu$s)---this is currently the only other demonstration of subthreshold scaling in superconducting-qubit QPUs.

We also report in \cref{tab-res-aachen} the suppression factor obtained after averaging over all four input states, namely $\Lambda^{\bd}_{\varepsilon} = \frac{\varepsilon^{(3,3)}}{\varepsilon^{\bd}}$ where $\varepsilon^{(3,3)} = \frac14\sum_{\alpha\in\{+,-,0,1\}}\varepsilon_\alpha^{(3,3)}$ and $\varepsilon^{\bd}=\frac14\sum_{\alpha\in\{+,-,0,1\}}\varepsilon_\alpha^{\bd}$, with $\bd\in\{(3,5),(5,3)\}$. This suppression factor lies well below $1$ for both scaled codes, suggesting  that the anisotropic scaling does not suffice to provide increased protection for all basis states. In other words, the improvement due to QEC in the logical basis parallel to the growth direction does not overcome the degradation of the orthogonal basis caused by the increased circuit complexity without any further protection. However, this basis-state averaging is ad hoc (see~\cite{supp}, \cref{appsec:recover}) and the entire metric is susceptible to model-selection sensitivity. To illustrate the latter, we recomputed $\Lambda_{\varepsilon}^{\bd}$ under the one- and two-parameter fits. The resulting spread, $\bigl(\max\Lambda_{\varepsilon}^{\bd}-\min\Lambda_{\varepsilon}^{\bd}\bigr)/\mathrm{mean}\,\Lambda_{\varepsilon}^{\bd}$, reaches $14.93\%$ for $(3,5)$ with a $\overline{\ket{+}}$ memory and $17.22\%$ for $(5,3)$ with a $\overline{\ket{0}}$ memory. This underscores the need for AIC-guided model choice rather than arbitrary parameter-number selection.

\begin{table}[t]
\centering
\squeezetable
\renewcommand{\arraystretch}{1.2}
\setlength{\tabcolsep}{6pt}
\begin{tabular}{l|cc|cc}
\hline\hline
Suppression& \multicolumn{2}{c}{$\bd=(5,3)$} & \multicolumn{2}{|c}{$\bd=(3,5)$} \\
 factor metric& $\alpha=0$ & $\alpha=1$ & $\alpha=+$ & $\alpha=-$ \\
\hline
\multicolumn{1}{c|}{$\Lambda^{\bd}_{\varepsilon,\alpha}$}       & 1.46 & 1.46 & 1.23 & 1.24 \\
$\min_s\Lambda^{\bd}_{\varepsilon,\alpha,s}$ & 1.10 & 1.07 & 1.04 & 1.01 \\
\hline
\multicolumn{1}{c|}{$\Lambda^{\bd}_{\varepsilon}$}  & \multicolumn{2}{c}{0.67} & \multicolumn{2}{|c}{0.77} \\
\hline\hline
\end{tabular}
\caption{Basis‑specific suppression factors $\Lambda^{\bd}_{\varepsilon,\alpha}$ (scaled code vs. sublattice‑average $(3,3)$) and $\min_s\Lambda^{\bd}_{\varepsilon,\alpha,s}$ (scaled code vs. best $(3,3)$ sublattice). Values $>1$ indicate basis‑specific subthreshold improvement. The bottom row reports $\Lambda^{\bd}_{\varepsilon} = \frac{\varepsilon^{(3,3)}}{\varepsilon^{\bd}}$ where $(3,3)$ is the average over the four initial states of the three sublattice codes, and $\varepsilon^{\bd}=\frac14\sum_{\alpha\in\{+,-,0,1\}}\varepsilon_\alpha^{\bd}$. All uncertainties (omitted) are $<2\times 10^{-3}$.}
\label{tab-res-aachen}
\end{table}

\begin{figure*}[ht]
\includegraphics[width=.9\textwidth]{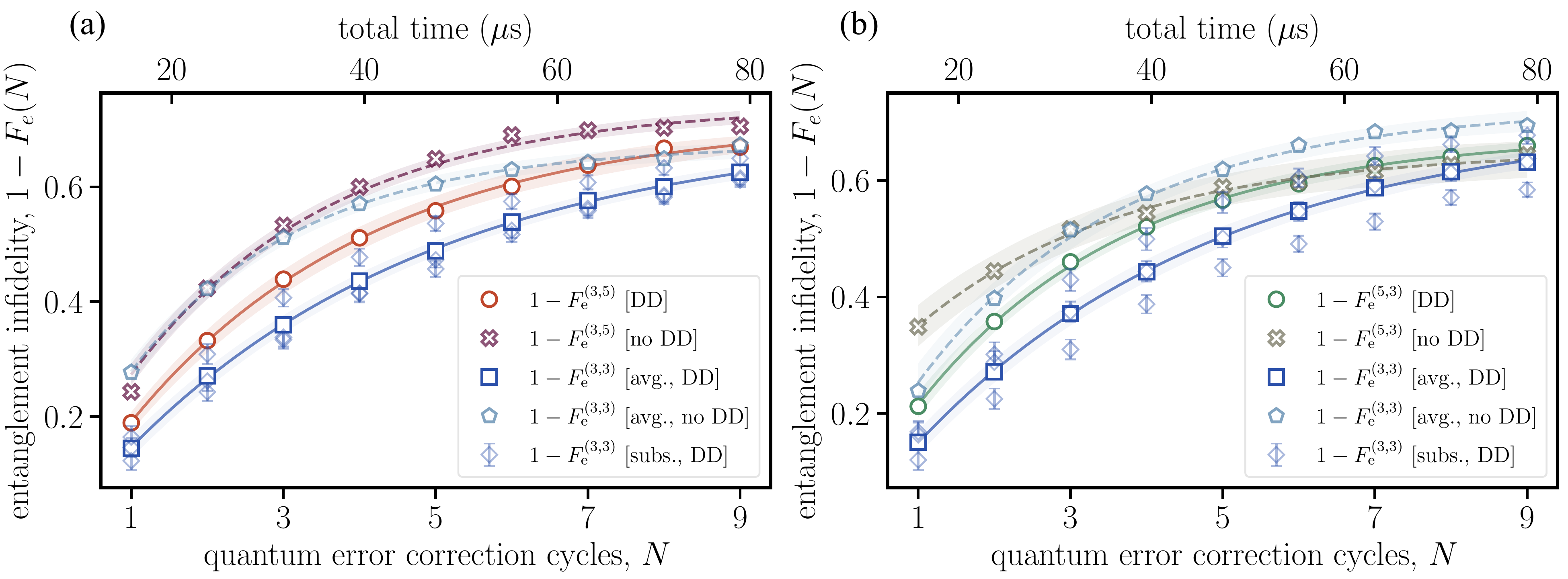}
\caption{Entanglement infidelity $1-\Fe(N)$ for (a) \dxdz{3}{5} and (b) \dxdz{5}{3}. As in \cref{fig:err-prob}, blue open markers are the two sets of three \dxdz{3}{3} sublattices relevant to each scaled code. Filled blue: arithmetic averages over those sublattices. We set $1-\Fe(0)=0$ by removing logical SPAM; $\Fe(N)$ is computed from $X/Z$ data using \cref{eq:ent-fidelity}.} 
\label{fig:ent-fidelity}
\end{figure*}


\begin{figure*}[ht]

\includegraphics[width=.9\textwidth]{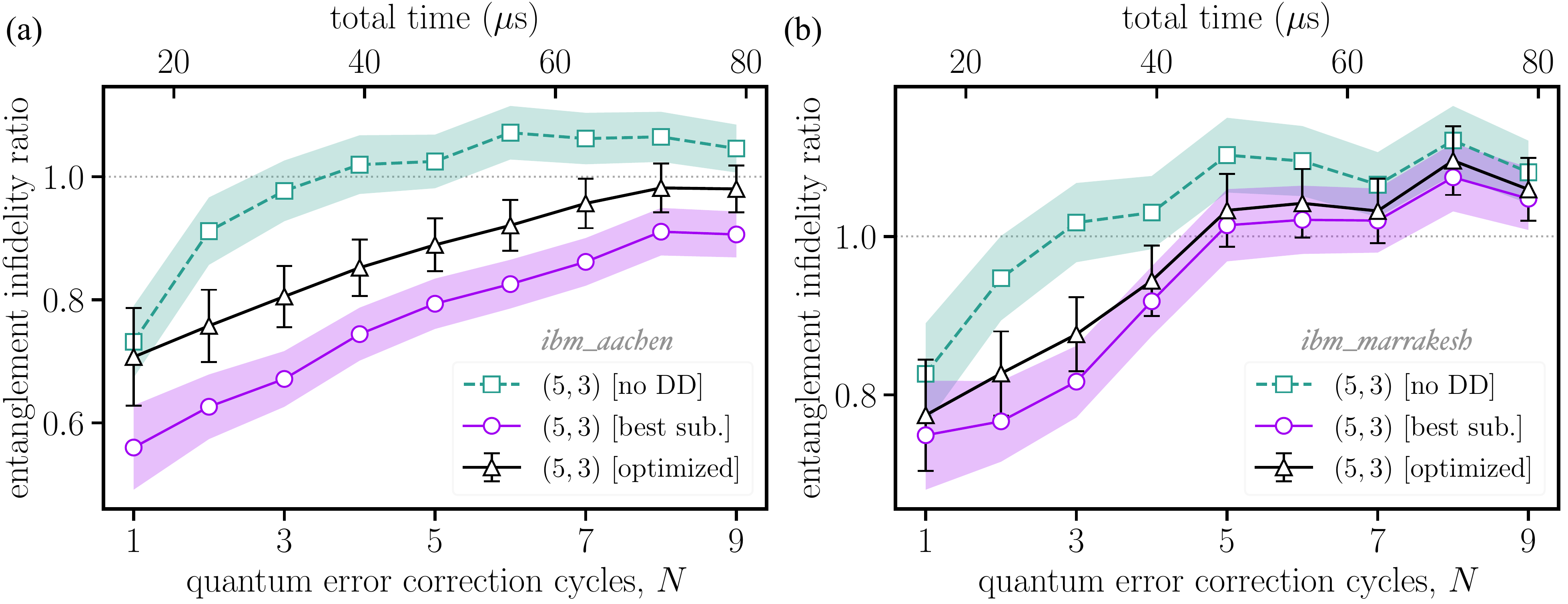}
\caption{Optimized EF metric $\Lambda_F$ [\cref{eq:EF-metric}] relative to the $(3,3)$ sublattice average (black with error bars) and relative to the best $(3,3)$ sublattice (purple, shaded). Also shown is the unoptimized EF metric $\frac{1-F_{\mathrm{e}}^{\bd'}(N)}{1-F_{\mathrm{e}}^{\bd}(N)}$ without DD (green, shaded). All results are for 
the $(5,3)$ code and its sublattices implemented on (a) \aachen\ and (b) \marrakesh. Values $>1$ would indicate basis‑independent subthreshold scaling. Consistent with \cref{fig:ent-fidelity}, this is not observed for \aachen\ once DD is considered. For \marrakesh\ we observe the onset of subthreshold scaling for $N\ge 5$ in terms of both the sublattice-average and the best sublattice, but with low statistical confidence. The no DD $(5,3)$ case exceeds $1$ for both \aachen\ and \marrakesh\ but is spurious, since the correct reference compares DD‑optimized $\max_{\{\text{DD,noDD}\}}\Fe$ for each code, as per \cref{eq:EF-metric}. To select the best sublattice, we compute $\frac{1}{9}\sum_{N=1}^{9}F_{\mathrm{e},s}^{(3,3)}(N)$ and maximize over $s\in\{1,2,3\}$.}
\label{fig:ent-infid-ratio}
\end{figure*}

\subsection{Entanglement fidelity QEC metric}

So far, we have mostly  considered the effect of code scaling on the $\{\overline{\ket{0}},\overline{\ket{1}},\overline{\ket{+}},\overline{\ket{-}}\}$ states separately. For the scaling of the suppression factor, this is motivated by the fact that different fitting parameters are required for each of the states [recall \cref{eq:3param-model}]. However, a genuine quantum memory should exhibit subthreshold scaling independent of the input state, and to assess this, we must use an unambiguous metric that inherently combines the results for the different basis states. 
A common approach based on the logical error rate $\varepsilon^{\boldsymbol{d}}_{\alpha}$ or the suppression factor $\Lambda_{\varepsilon,\alpha}^{\boldsymbol{d}}$, is to average over the basis states~\cite{GoogleAI2024Nature}. As explained in \cite{supp} (see \cref{sec:rate-from-single-point}), this approach can only be rigorously connected to the channel infidelity of the quantum memory after a single $N=1$ QEC cycle, and in the limit of small logical Pauli errors.  For larger cycle numbers, a more complicated expression is needed that would in addition require the measurement of all six cardinal states.

\begin{figure*}[t]

\includegraphics[width=.9\textwidth]{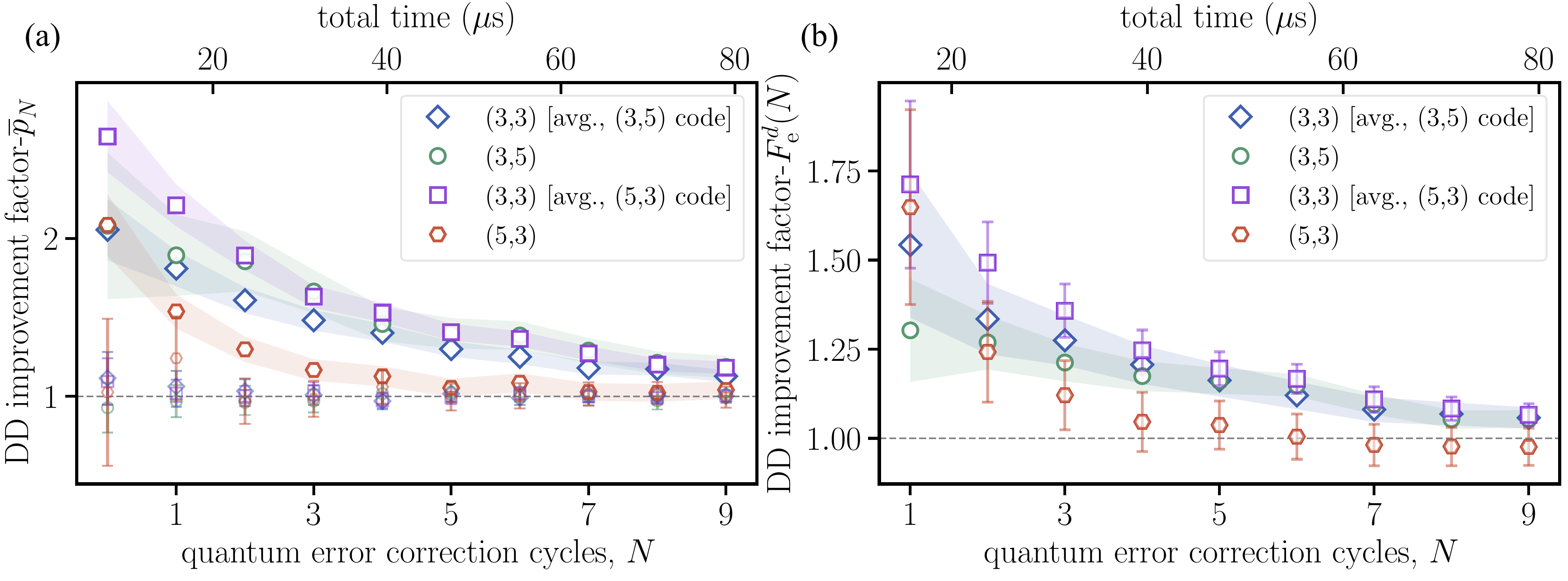}
\caption{The DD improvement factor in terms of (a) ratios of logical error probabilities and (b) entanglement infidelities. In more detail, in (a) we plot $\overline{p}^{\bd,\mathrm{noDD}}_{N,\alpha}/\overline{p}^{\bd}_{N,\alpha}$, for the scaled codes and their respective $(3,3)$ code sublattice averages. A value above $1$ indicates that DD helps. The shaded confidence intervals are for the $\overline{\ket{+}}$ state and the error bars for the $\overline{\ket{0}}$ state. DD results in a significant improvement in the $\overline{\ket{+}}$ case, and has essentially no effect in the $\overline{\ket{0}}$ case. Results for $\overline{\ket{-}}$ (significant improvement) and $\overline{\ket{1}}$ (no improvement) are nearly identical and are not shown. 
In (b) we plot $\frac{1-\Fe^{\bd,\text{noDD}}(N)}{1-\Fe^{\bd}(N)}$, which is the SPAM-free version of (a), and also simultaneously accounts for all four basis states. The $(3,5)$ code (green) is improved more by DD than the $(5,3)$ code (brown), while the opposite is true for their respective $(3,3)$ codes.
}
\label{fig-DD-ratio}
\end{figure*}

The natural alternative to the basis-dependent approach is the EF, which, as we have shown in \cref{sec:EF}, combines all the data provided by the decoder into a single, easily computable and rigorous performance metric. Using the EF [\cref{eq:ent-fidelity}] removes the ambiguity associated with model fitting and, moreover, accounts for SPAM and non-unital errors (as motivated by our three-parameter model), as well as cycle-dependence.
We plot the resulting entanglement infidelities 
in \cref{fig:ent-fidelity}, and their ratios, i.e., the entanglement fidelity metric of \cref{eq:EF-metric}, in \cref{fig:ent-infid-ratio}.

Before interpreting these results, we note that the basis states $\{\overline{\ket{0}},\overline{\ket{1}}\}$ are less protected for the $(3,5)$ code than $\{\overline{\ket{+}},\overline{\ket{-}}\}$ (since its smaller $d_x$ implies less protection against bit-flip than phase-flip errors), and the same argument holds for the $(5,3)$ code with interchanged roles for $X$ and $Z$. Thus, we expect these less protected states to contribute higher logical error probabilities to the EF compared to the $(3,3)$ code, as we increase the circuit complexity by growing the codes anisotropically.
It is unclear \textit{a priori} whether the advantage gained by growing the code distance for the more protected basis states will overcome the  degradation of the EF due to the less protected states. As shown in \cref{fig:ent-fidelity,fig:ent-infid-ratio}, this is not the case: the infidelity of the scaled codes is generally higher than that of the smaller codes, such that there is no overall subthreshold scaling. This holds with DD and also without DD except for the $(5,3)$ code, which we discuss below. We stress that these results do not depend on a fitting model and, in contrast to an average per-cycle error $\epsilon_{\bd}$ and the corresponding suppression factor, are informative at the cycle level well beyond $N=1$.

However, the conditions for subthreshold scaling according to the full EF metric are actually not far from current performance. In fact, using our circuit-level noise model and numerical simulations for the $(5,5)$ code, we find (see \cref{sec:d_5_simulations} of~\cite{supp}) that a $30\%$ reduction of the current experimental noise rates in a slightly larger QPU that has the six missing qubits needed to host this larger code would already achieve genuine subthreshold scaling, i.e., $\Fe^{(5,5)}(N) > \Fe^{(3,3)}(N)$. In \cref{appsec:Willow} of~\cite{supp} we discuss the application of our methodology to the Willow-processor results of Ref.~\cite{GoogleAI2024Nature}. 
We reanalyze the public Willow data with the EF metric, and confirm subthreshold scaling $d:3\to 5\to 7$, with 
a deeper per‑cycle advantage for 
$d:5\to 7$ than 
$d:3\to 5$. 
In contrast, basis‑averaged, single‑parameter suppression factors vary at the 10–15\% level with the chosen fit 
and presuppose stationarity that is not supported by the data.
We therefore advocate EF as a primary, fit‑free metric for an unambiguous assessment of subthreshold scaling across different QPUs with different error mechanisms.

\subsection{Role of DD and spurious subthreshold scaling}
\label{sec:spurious}

As mentioned above, our results require careful integration of optimized DD into embedded surface code circuits. \Cref{fig:ent-fidelity} also contains a comparison between circuits that do or do not incorporate DD, and it is clear at a glance that, in general, DD significantly reduces the entanglement infidelity. Thus, DD plays a critical role in improving the performance of the surface code SWAP-based embedding  in our experiments.

\Cref{fig:ent-infid-ratio} shows that for sufficiently large $N$, the fidelity metric exceeds $1$ in the case of the $(5,3)$ code without DD (green, shaded). This might seem to suggest that subthreshold scaling has been achieved. However, it is essential to first examine the actual entanglement infidelity values and to ensure that they are minimized. In doing so using \cref{fig:ent-fidelity}, we observe that $1-\Fe^{(3,3)}$ is significantly lower with DD than without it. This means that the correct reference is the case with DD, and a conclusion of subthreshold scaling based on a comparison with the noDD reference is spurious. In other words, as already explained in the context of \cref{eq:EF-metric}, each code must be optimized separately with respect to the inclusion of DD. When this is done (black, error bars), we observe that for \aachen, the subthreshold scaling result disappears. For \marrakesh\ we do observe subthreshold scaling for $N\ge 6$, but not until the last cycle ($N=9$) does this result hold with $95\%$ confidence. This represents the onset of subthreshold scaling, i.e., the scaled $(5,3)$ code gradually becomes more effective than the $(3,3)$ code as more errors accumulate. Results for the $(3,5)$ code are well above threshold and are shown in~\cite{supp} (\cref{appsec:B}), along with a complete set of results for \marrakesh.

To further quantify the role played by DD, \cref{fig-DD-ratio} shows the improvement factor we obtain by comparing QEC with and without DD, namely (a) 
$\Lambda_{\mathrm{DD}}^{(\bd)}(N) \equiv \overline{p}^{\bd,\mathrm{noDD}}_{N,\alpha}/\overline{p}^{\bd}_{N,\alpha}$ where the noDD superscript indicates that the corresponding circuits were run without any DD pulses, as well as (b) the SPAM-free entanglement infidelity ratios 
$\Lambda_{F,\mathrm{DD}}^{(\bd)}(N) \equiv \frac{1-\Fe^{\bd,\mathrm{noDD}}(N)}{1-\Fe^{\bd}(N)}$.
As seen in \cref{fig-DD-ratio}(a), while the effect on the $\overline{\ket{0}}$ state is small (symbols with error bars), the logical error probability of the $\overline{\ket{+}}$ state (symbols with shading) is significantly higher without DD for all codes.
This is evidence for a strong bias towards logical $Z$-type errors, including coherent errors due to residual $ZZ$ crosstalk from tunable couplers, and non-Markovian correlated dephasing. DD strongly suppresses these effects, with a particularly pronounced effect for the $(3,3)$ code (purple) relative to its $(5,3)$ extension (brown). The same observation holds after SPAM removal and when accounting for all four basis states, as seen in terms of the infidelity ratios in \cref{fig-DD-ratio}(b). \Cref{fig-DD-ratio} also shows that the $(3,5)$ code benefits more from DD than the $(5,3)$ code. This is consistent with DD's stronger suppression of $Z$-type errors. As for the overall decline in the DD improvement factor with $N$, we attribute this to the aforementioned accumulation of coherent errors (see \cite{supp}, \cref{appsec:coherent-errors} for more details).

\section*{Summary and Conclusions}

We implemented a surface-code quantum memory on IBM heavy-hex QPUs using a lightweight SWAP-based embedding integrated with robust dynamical decoupling (DD). With this hardware-aware co-design we observe directional suppression of logical errors under anisotropic distance scaling: increasing $d_x$ ($d_z$) improves protection against $X$-type ($Z$-type) logical errors, despite increases in circuit width and depth.

To benchmark code scaling without model-fit ambiguities, we introduced an entanglement-fidelity (EF) metric that (i) combines all four accessible basis states; (ii) accounts for SPAM and non-unital effects; and (iii) remains valid under cycle-dependent noise. Using EF, we did not find global, state-independent subthreshold scaling for the full logical channel on current IBM QPUs, with the exception of the onset of such scaling at late QEC cycles for \marrakesh. Our results clarify the limits of inferences drawn from widely used single-parameter, basis-averaged suppression factors, whose assumptions (cycle-independence, unitality, no SPAM) are not fully met in our experiments. 

DD is essential in this setting. It suppresses dominant $Z$-biased mechanisms—including coherent $ZZ$ crosstalk and correlated dephasing, thus pushing the residual noise closer to what the surface code can effectively correct. We showed that threshold assessment must optimize each code instance, before comparison, over using or not using DD: using an unoptimized reference can yield spurious indications of subthreshold scaling. In practice, our recommended metric is the entanglement infidelity ratio computed from the best configuration per code.

Circuit-level simulations calibrated to experiment indicate that a modest, $\sim 30\%$ reduction in current noise rates, together with access to a slightly larger heavy-hex device that can host a $(5,5)$ code, would achieve genuine subthreshold scaling under the EF metric. 

Our results establish a practical path for the demonstration of surface code scaling: (i) minimize depth via connectivity-aware embedding; (ii) integrate robust DD to suppress non-Markovian and coherent error components; and (iii) evaluate scaling using EF-based, SPAM-aware metrics. This combination delivers concrete, device-level targets for realizing below-threshold operation on non-native surface-code architectures, and provides a model-free benchmark for future progress toward fault-tolerant quantum computing.

\acknowledgements
We are grateful to Esperanza L\'{o}pez for useful discussions on optimizing the integration of the DD and QEC circuits and for comments on the manuscript. We acknowledge support from the Office of the Director of National Intelligence (ODNI), Intelligence Advanced Research Projects Activity (IARPA), under the Entangled Logical Qubits program through Cooperative Agreement Number W911NF-23-2-0216. A.B. and C.B. acknowledge support from PID2021-127726NB-I00 and PID2024-161474NB-I00 (MCIU/AEI/FEDER, UE), from the Grant IFT Centro de Excelencia Severo Ochoa CEX2020-001007-S, funded by MCIN/AEI/10.13039/501100011033, from the CSIC Research Platform on Quantum Technologies PTI-001, and from the European Union’s Horizon Europe research and innovation programme under grant agreement No 101114305 ("MILLENION-SGA1” EU Project). A.B and C.B. acknowledge the use of IBM Quantum Credits for the initial optimization of QEC circuits. A.V., M.M.-O., and D.A.L. acknowledge support by the U.S. Army Research Laboratory and the U.S. Army Research Office under contract/grant number W911NF2310255, and by the Defense Advanced Research Projects Agency under Agreement HR00112230006. The views, opinions and/or findings expressed are those of the author(s) and should not be interpreted as representing the official views or policies of the Department of Defense or the U.S. Government. This research was conducted using IBM Quantum Systems provided through the University of Southern California's IBM Quantum Innovation Center. The views expressed are those of the authors and do not reflect the official policy or position of IBM or the IBM Quantum team.

\putbib[biblo]
\end{bibunit}


\clearpage
\appendix

\onecolumngrid
\begin{center}
\large \bf Supplemental Material: Surface code scaling on heavy‑hex superconducting quantum processors 

\end{center}
\twocolumngrid

\startappendixtoc

\printappendixtoc

\begin{bibunit}[apsrev4-2]

\twocolumngrid

\begin{figure*}[ht]
\hspace{0cm}{\includegraphics[width=1\textwidth]{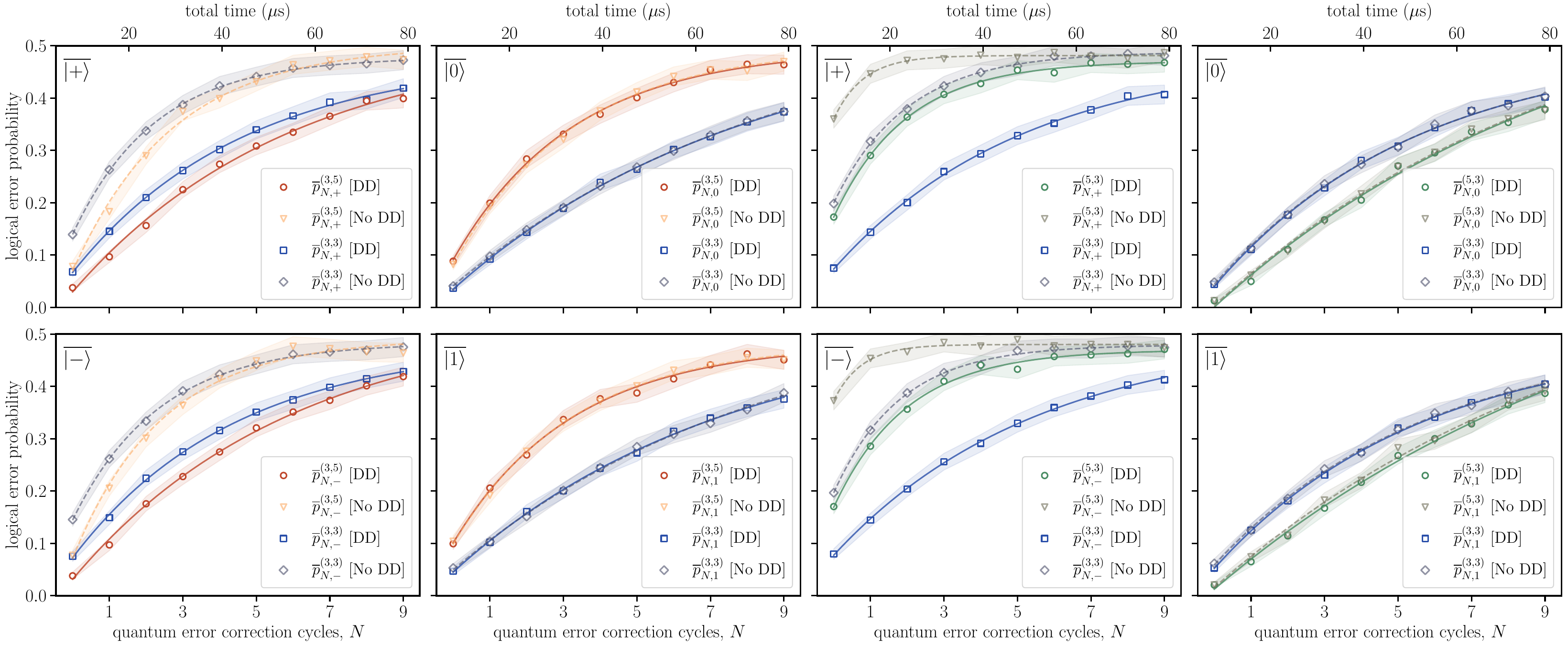}} 
\caption{Full set of empirical data from \aachen. Logical error probabilities are shown for both no DD and best-performing DD, for logical states $\{\overline{\ket{0}},\overline{\ket{1}},\overline{\ket{+}},\overline{\ket{-}}\}$.}
\label{fig-aachen-all}
\end{figure*}


\begin{figure*}[ht]
\hspace{0cm}{\includegraphics[width=1\textwidth]{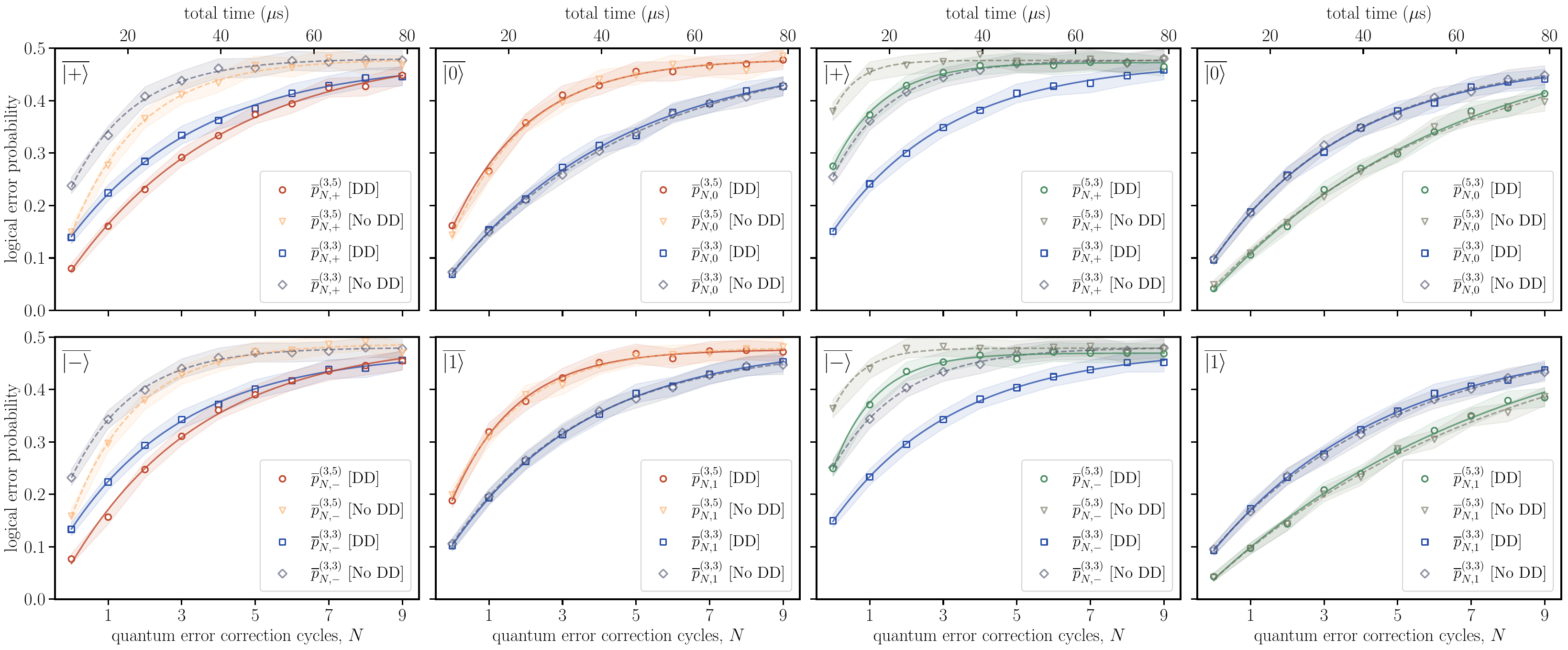}} 
\caption{Full set of empirical data from \marrakesh. Logical error probabilities are shown for both no DD and best-performing DD, for logical states $\{\overline{\ket{0}},\overline{\ket{1}},\overline{\ket{+}},\overline{\ket{-}}\}$.}
\label{fig-marrakesh-all}
\end{figure*}


\begin{figure*}[ht]
\hspace{0cm}{\includegraphics[width=1\textwidth]{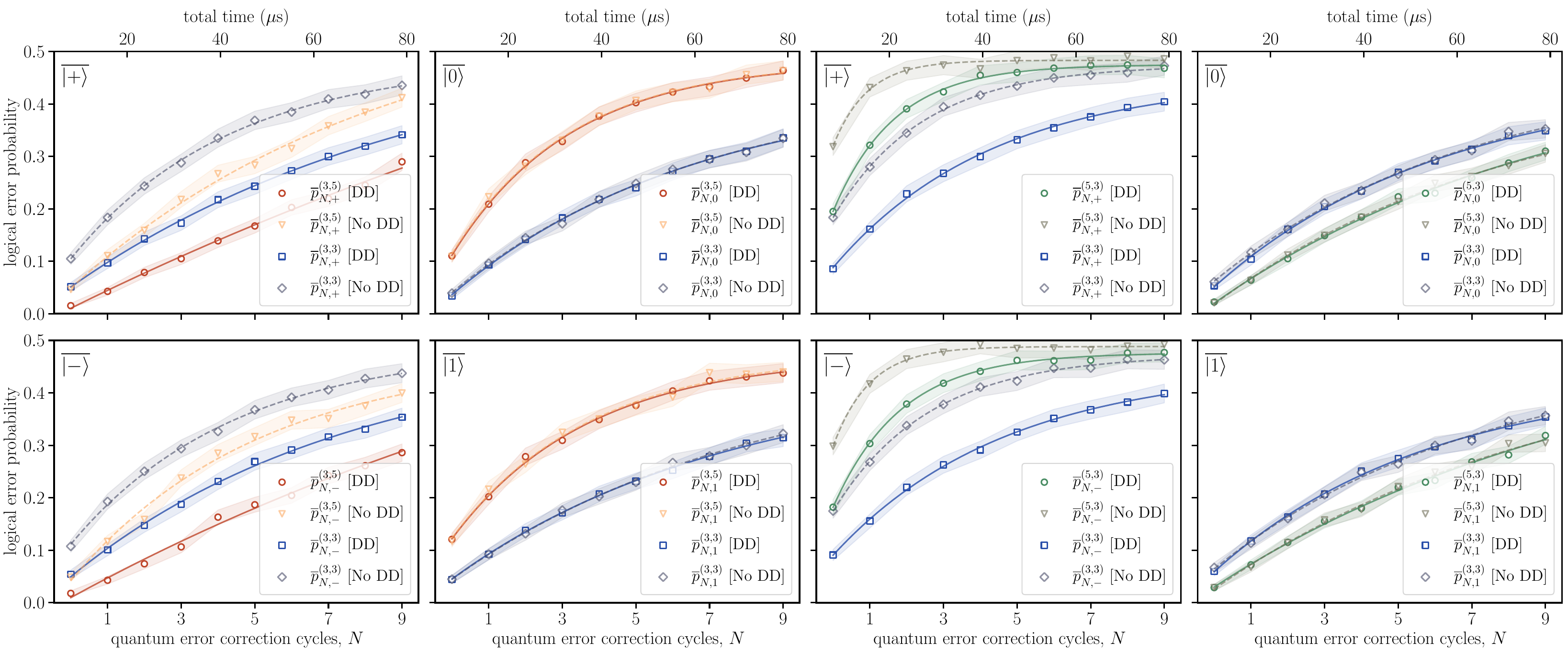}} 
\caption{Full set of empirical data from \pittsburgh. Logical error probabilities are shown for both no DD and best-performing DD, for logical states $\{\overline{\ket{0}},\overline{\ket{1}},\overline{\ket{+}},\overline{\ket{-}}\}$}. 
\label{fig-pittsburgh-all}
\end{figure*}


\begin{figure*}[ht]
\hspace{0cm}{\includegraphics[width=.85\textwidth]{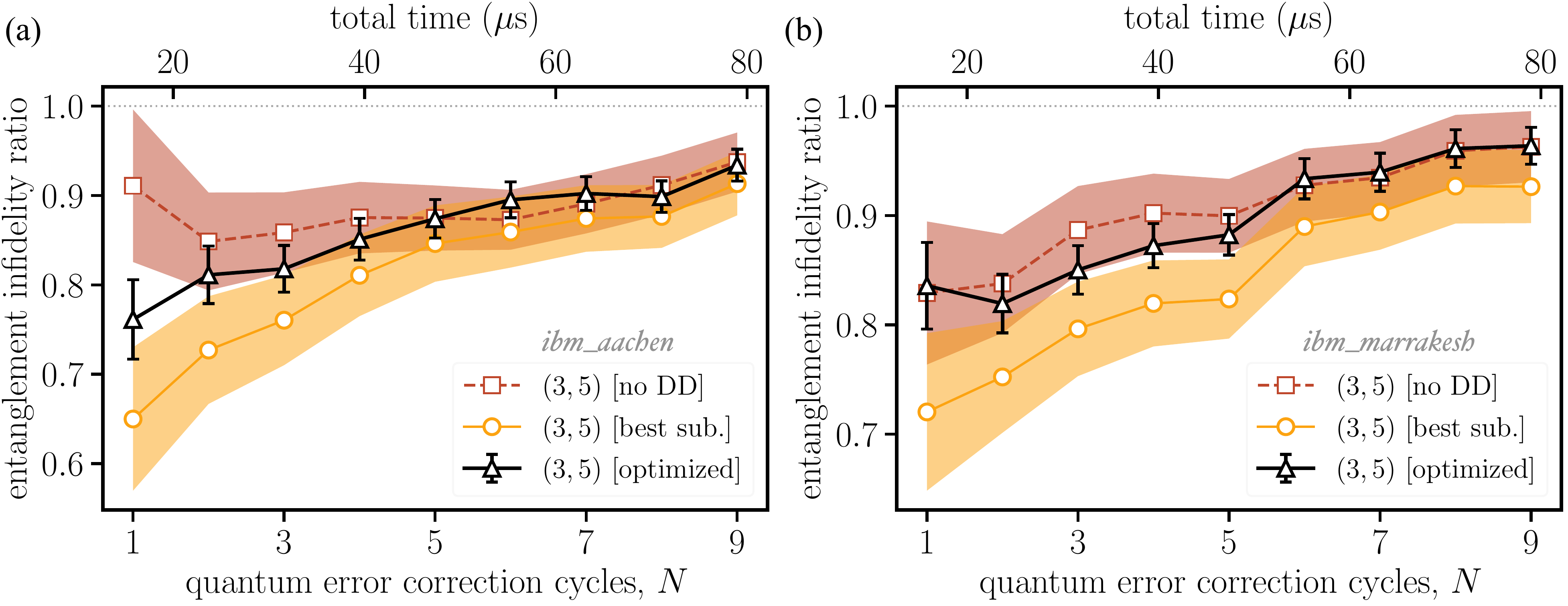}} 
\caption{Exactly as in \cref{fig:ent-infid-ratio} of the main text, but for the $(3,5)$ code. 
Optimized EF metric $\Lambda_F$ [\cref{eq:EF-metric}] relative to the $(3,3)$ sublattice average (black with error bars) and relative to the best $(3,3)$ sublattice (purple, shaded). Also shown is the unoptimized EF metric $\frac{1-F_{\mathrm{e}}^{\bd'}(N)}{1-F_{\mathrm{e}}^{\bd}(N)}$ without DD (green, shaded). All results are for 
the $(3,5)$ code and its sublattices implemented on (a) \aachen\ and (b) \marrakesh. Values $>1$ would indicate basis‑independent subthreshold scaling but are not observed here. To select the best sublattice, we compute $\frac{1}{9}\sum_{N=1}^{9}F_{\mathrm{e},s}^{(3,3)}(N)$ and maximize over $s\in\{1,2,3\}$.}
\label{fig-marrakesh-summary}
\end{figure*}

\begin{figure*}[ht]
\hspace{0cm}{\includegraphics[width=.85\textwidth]{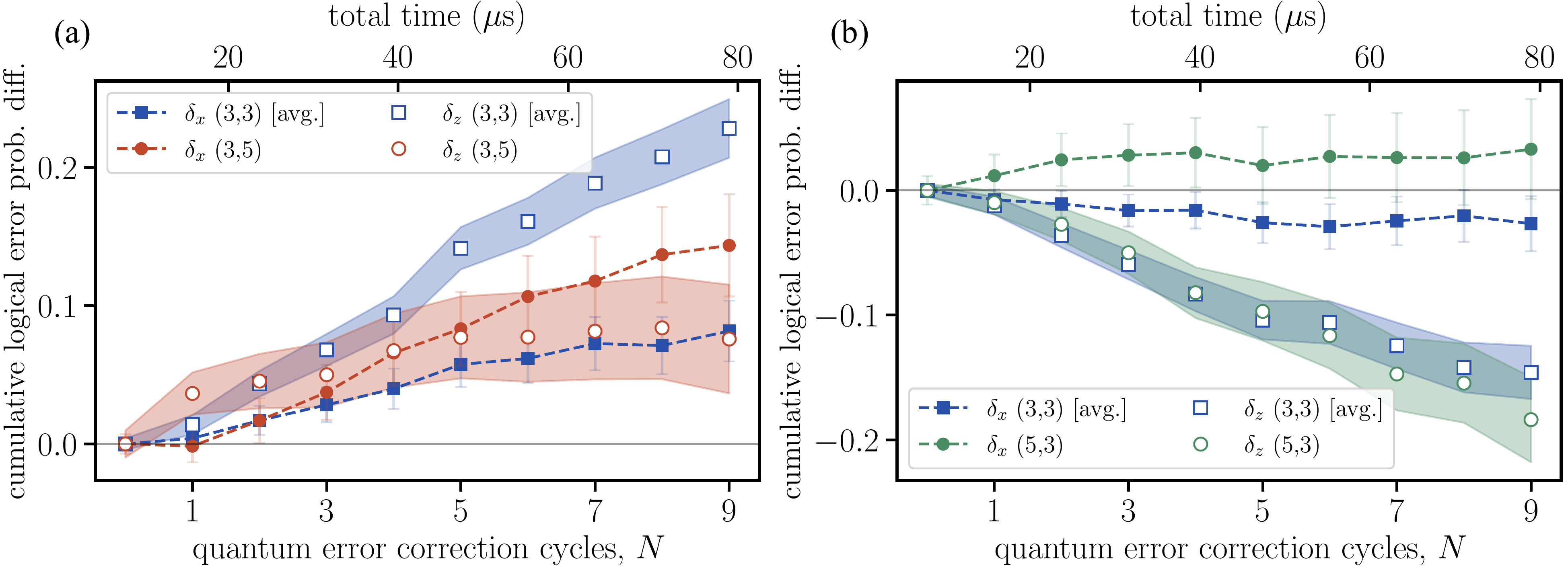}} 
\caption{As in \cref{fig:err-prob}(c,d) of the main text, but for \marrakesh. Cumulative differences $\delta^{\bd}_x=\sum_{N}(\overline{p}^{\bd}_{N,-}-\overline{p}^{\bd}_{N,+})$ vs. $\delta^{\bd}_z=\sum_{N}(\overline{p}^{\bd}_{N,1}-\overline{p}^{\bd}_{N,0})$ for (a) $(3,5)$ and (b) $(5,3)$, along with the average over their respective $(3,3)$ sublattices. Deviations from zero witness non‑unital logical noise, growing with $N$ and more pronounced for the $Z$ eigenstates.}
\label{fig-marrakesh-summary2}
\end{figure*}

\begin{figure*}[ht]
\hspace{0cm}{\includegraphics[width=.85\textwidth]{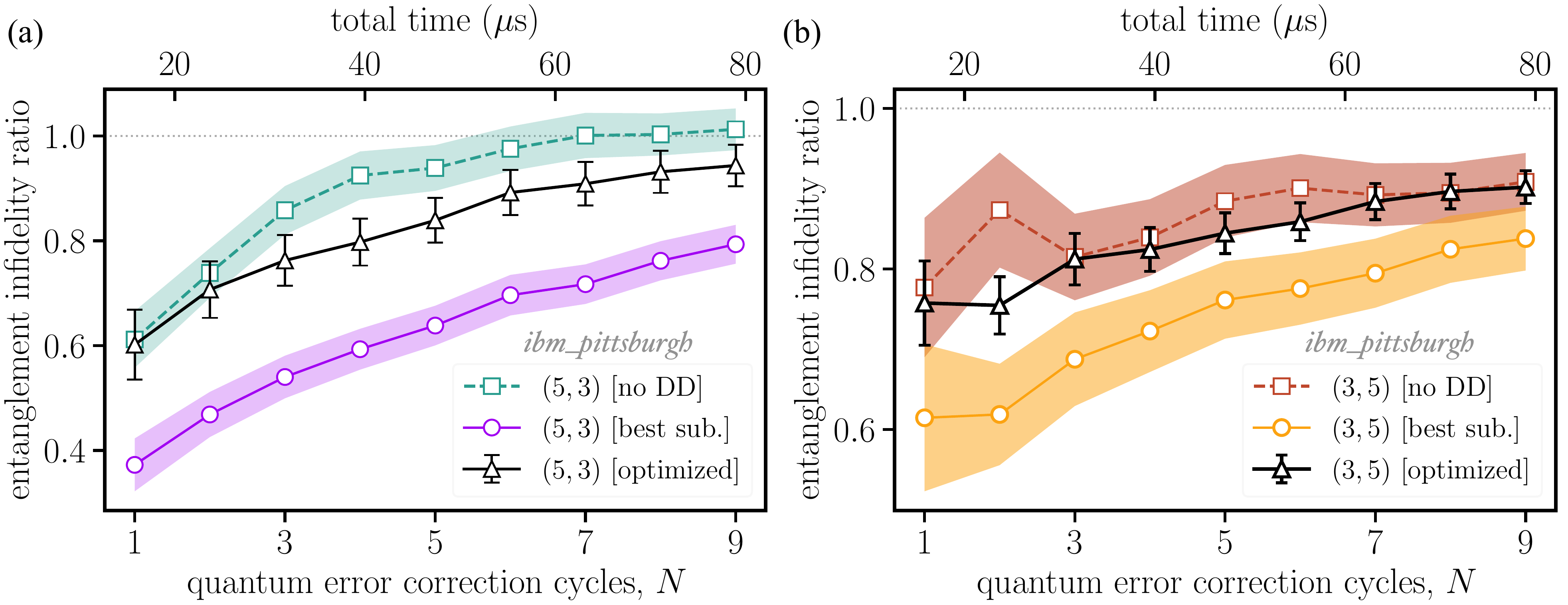}} 
\caption{Exactly as in \cref{fig:ent-infid-ratio} of the main text, but for (a) the $(5,3)$ code and (b) the $(3,5)$ code, both implemented on \pittsburgh. The onset of spurious subthreshold scaling is observed in the $(5,3)$ case for the no DD result at $N\ge 7$.}
\label{fig-pittsburgh-summary}
\end{figure*}

\begin{figure*}[ht]
\hspace{0cm}{\includegraphics[width=.85\textwidth]{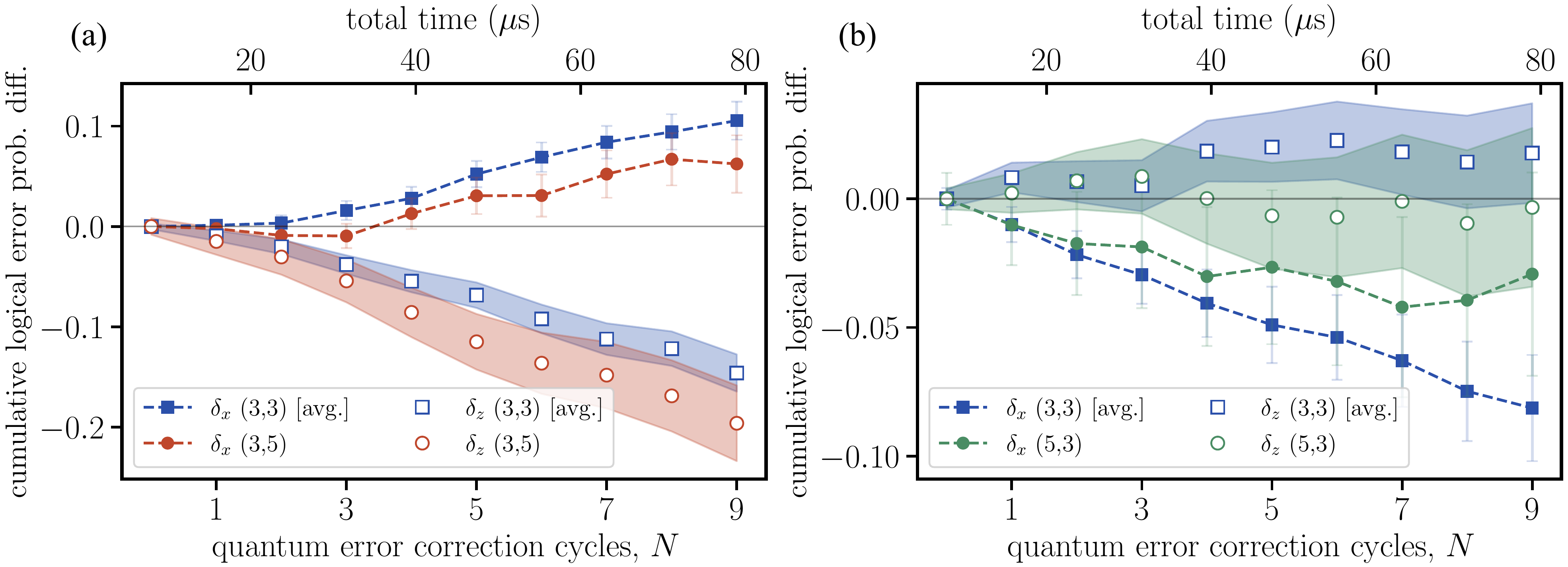}} 
\caption{As in \cref{fig:err-prob}(c,d) of the main text, but for \pittsburgh. Cumulative differences $\delta^{\bd}_x=\sum_{N}(\overline{p}^{\bd}_{N,-}-\overline{p}^{\bd}_{N,+})$ vs. $\delta^{\bd}_z=\sum_{N}(\overline{p}^{\bd}_{N,1}-\overline{p}^{\bd}_{N,0})$ for (a) $(3,5)$ and (b) $(5,3)$, along with the average over their respective $(3,3)$ sublattices. Deviations from zero witness non‑unital logical noise, growing with $N$ and more pronounced for the $Z$ eigenstates for $(3,5)$, but (unlike both \aachen\ and \marrakesh) more pronounced for the $X$ eigenstates for $(5,3)$.}
\label{fig-pittsburgh-summary2}
\end{figure*}


\section{Quantum processing units used in this work}
\label{sec:other_qpus_supp}

We conducted our experiments primarily using two IBM QPUs: \aachen\ and \marrakesh. We also ran experiments on the Heron processors \fez\ and \kyiv\, but these results are not shown, as we found them to underperform relative to \aachen\ and \marrakesh. These QPUs are superconducting quantum devices with $156$ transmon qubits, arranged on a heavy-hexagonal lattice, with each qubit occupying a vertex having three nearest neighbors, while those on the links have two nearest neighbors. This reduced connectivity was initially motivated by the use of nontunable entangling gates in past Eagle devices, in order to avoid unintended resonances, reducing unwanted crosstalk and improving gate performance. In Heron devices 
the heavy-hex connectivity is still important for reducing control complexity while minimizing crosstalk, and has allowed demonstrating a quantum volume of $512$~\cite{mckay2023arxiv}. 

The native gate set consists of \texttt{CZ}, \texttt{RZ}, \texttt{SX}, \texttt{X}, \texttt{ID}, and \texttt{reset}. Qubits are implemented as fixed-frequency transmons. Single-qubit gates are realized by microwave-driven \texttt{X} and \texttt{SX}, namely $\pi$ or $\pi/2$ rotations around the $X$-axis of the Bloch sphere, and \texttt{RZ} gates via the virtual-$Z$ scheme~\cite{McKay2017PRA}, which shifts the phase of subsequent pulses to account for arbitrary-angle rotations around the $Z$-axis of the Bloch sphere. \texttt{ID} stands for a timing delay between gates, used to align schedules for parallel operations. \texttt{reset} is obtained by a dispersive coupling of the qubit to a readout resonator, followed by an \texttt{X} gate conditioned on the measurement result to reset the qubit to the physical $\ket{0}$ state. Heron devices feature tunable couplers between adjacent qubits, which mediate an adjustable qubit-qubit interactions leading to controlled-$Z$ gates \texttt{CZ}, facilitating the suppression of undesired interactions~\cite{Stehlik2021PRL}. 

As mentioned in the main text, we optimize the circuits devised in~\cite{Benito2025quantum} by removing the reset gates. While our no-reset approach has the drawback of allowing leakage errors to accumulate across successive QEC cycles, it significantly reduces the duration of each cycle: by 2.58~$\mu$s per cycle for \marrakesh~and by 2.60~$\mu$s per cycle for \aachen. In addition to shortening the cycle duration, we note that the reset process on the Heron processors is conditioned on a dispersive measurement, which differs from an all-zero state initialization~\cite{IBMQNative}: it applies $X$ if the qubit is found to be in the $\ket{1}$ state and does nothing otherwise. This type of reset does not address leakage to states outside of the qubit manifold. Our choice of the no-reset approach is corroborated by recent results in Ref.~\cite{harper2025arxiv}, which examined the trade-off introduced by the extra time required for resets and concluded that incorporating slow and incomplete resets likely degrades QEC codes performance.


\section{Additional results: \marrakesh, \pittsburgh}
\label{appsec:B}

The main text primarily reports our results from \aachen, the top-performing QPU in our study. Here we report additional results for \marrakesh\ and \pittsburgh. These results are overall consistent with our findings for \aachen, with relatively minor variations. The only significant difference is the observation of the onset of subthreshold scaling for \marrakesh, discussed in the main text.


\Cref{fig-aachen-all,fig-marrakesh-all} show the logical error probabilities for \aachen~and \marrakesh, respectively, in analogy to \cref{fig:err-prob} in the main text. However, we also include the results without any DD (dashed lines), which show a significant deterioration for the $\overline{\ket{+}}$ state, but a negligible effect for the $\overline{\ket{0}}$ state. Overall, the qualitative trends are similar for \aachen\ and \marrakesh, but the logical error probabilities and per-cycle logical error rates are somewhat higher for \marrakesh. This is the reason why we chose to focus on \aachen\ in the main text.

\Cref{fig-marrakesh-summary} shows the EF metric for \aachen\ and \marrakesh\ in the case of the $(3,5)$ code, for which there is no subthreshold scaling (not even spurious). \Cref{fig-marrakesh-summary2} shows the non-unitality witness for \marrakesh, which behaves qualitatively similarly to the \aachen\ case. \Cref{fig-pittsburgh-summary} shows the \pittsburgh\ EF metric for both the $(3,5)$ and $(5,3)$ codes. \Cref{fig-pittsburgh-summary2} shows the non-unitality witness for \pittsburgh.


\section{DD optimization}
\label{sec:dd_opt_supp}

As discussed in the main text, the implementation of a surface code on a heavy-hex lattice inevitably introduces idle gaps due to the lattice's specific connectivity. This stands in contrast to surface code implementations on square-lattice QPUs, such as the Willow processor, where parallelization eliminates idle periods, and DD is employed only during readout and exclusively on data qubits. A straightforward, but naive, strategy for heavy-hex lattices would be to fill \emph{all} idle gaps with the same DD sequence. However, such an approach can actually degrade the fidelity as the idle gaps vary in duration. For example, a four-pulse sequence such as XY4~\cite{Maudsley1986ty} may be well suited for a short idle gap of approximately $400$ ns (given a pulse duration of $32$ ns on \aachen), but in a longer gap of $\sim1\mu$s, the increased spacing between pulses can reduce the effectiveness of the sequence. To address this, longer gaps require DD sequences with a greater number of pulses, but simple periodic repetition of a short sequence can also reduce fidelity~\cite{Ezzell2021PRApp}. For such cases, we employ $n$-pulse universally robust (UR$_n$) sequences~\cite{Genov2017PRL}. However, using excessively large $n$ can result in pulse interference~\cite{Vezvaee2025PRXQ}, again degrading performance. Therefore, it is essential to identify an optimal balance by applying larger UR$_n$ sequences to longer gaps and shorter sequences to smaller gaps, in order to maximize the overall effectiveness of DD. The strategy we employed to implement this optimized approach is described below.

\subsubsection*{Problem Setting and Notation}
\begin{itemize}
 \item Baseline circuit $C_0$ contains a set of idle gaps $\Gaps(C)$ measured in discrete time steps ($\dt$); for comparison, for all three QPUs, an $X$ gate is $32$ns=$8\dt$.
 \item DD library (9 sequences): 
 \\
 $\Lseq=\{\mathrm{UR18},\mathrm{UR16},\mathrm{UR14},\ldots,\mathrm{UR6},\mathrm{XY4},\mathrm{RGA8}_a\}$.
 \item We specifically avoid the Carr-Purcell-Meiboom-Gill (CPMG) type sequences since they are not robust and in fact insert more errors into the system due to the coherent errors of the pulses~\cite{Vezvaee2025arxiv,Ezzell2021PRApp}.
 \item Each DD sequence $s\in\Lseq$ has a fixed duration $L(s)$ (in $\dt$, the system cycle time).
 \item $\Err(C)$ is the decoded logical error probability for that circuit $C$.
 \item An integer grid $\T_{\text{grid}}\subset \mathbb{N}^+$, $\T\in\T_{\text{grid}}$.

\end{itemize}

\subsubsection*{High-Level Description of the Procedure}
Starting from $C_0$, we iteratively perform \emph{passes}. In each pass we:
\begin{enumerate}
 \item Choose a single DD type $s\in\Lseq$ and a single $\T\in\T_{\text{grid}}$. 
 \item Insert at most one instance of $s$ (i.e., no repetition) into each currently empty gap that satisfies $L(s)\le G/\T$. If several candidate sequences are considered \emph{within a pass}, the \emph{longest qualifying} one is inserted. If none qualify, the gap remains empty. 
 \item For a given noise model, vary the parameters of the model to find the lowest logical error from the bare circuit $C_0$. Compare the lowest logical error of all passes to the lowest logical error of the bare circuit and select the instance $(s,\mathcal{T})$ with the largest improvement. Call this circuit $C_1$.
 \item Repeat the process above with $C_1$.
 \item Stop if $C_n$ does not improve over $C_{n-1}$.
 \item Final output: circuit $C_n$ and the list $\{(s_1,\mathcal{T}_1),\cdots,(s_n,\mathcal{T}_n)\}$.
 \end{enumerate}
We summarize this procedure as pseudocode in \cref{alg}.

\begin{algorithm}[H]
\caption{DD Optimization}
\begin{algorithmic}[1]
\Require Baseline circuit $C_0$; DD library $\mathcal{\rm L}$; $\mathcal{T}_{\text{grid}}$; Logical error $\overline{p}_N$
\Ensure Final circuit $C_n$; list of passes $[(s_1,\mathcal{T}_1),\ldots,(s_n,\mathcal{T}_n)]$
\State $C \gets C_0$;\quad $e \gets \min \overline{p}_N$ for $C$;\quad $\mathcal{P}\gets[ ]$
\While{true}
 \State $(e_{\text{best}},s_{\text{best}},\mathcal{T}_{\text{best}},C_{\text{best}})\gets(e,\text{none},\text{none},C)$
 \For{$s\in\mathcal{\rm L}$}
 \For{$\mathcal{T}\in\mathcal{T}_{\text{grid}}$}
  \State $C' \gets \textsc{Sweep}(C,s,\mathcal{T})$ \Comment{insert at most one $s$ per gap if $L(s)\le G/\mathcal{T}$}
  \If{$C' \neq C$}
  \State $e' \gets \min \overline{p}_N$ for $C^\prime$
  \If{$e' < e_{\text{best}}$} \State $(e_{\text{best}},s_{\text{best}},\mathcal{T}_{\text{best}},C_{\text{best}})\gets(e',s,\mathcal{T},C')$ \EndIf
  \EndIf
 \EndFor
 \EndFor
 \If{$e_{\text{best}} \ge e$} \textbf{break} \EndIf
 \State $C \gets C_{\text{best}}$;\quad $e \gets e_{\text{best}}$;\quad $\mathcal{P}\gets \mathcal{P} \Vert [(s_{\text{best}},\mathcal{T}_{\text{best}})]$
\EndWhile
\State \Return $C_n\gets C,\ [(s_1,\mathcal{T}_1),\ldots,(s_n,\mathcal{T}_n)]\gets \mathcal{P}$
\label{alg}
\end{algorithmic}
\end{algorithm}

As an example, suppose the optimal strategy found by the procedure returns the passes
\[
\mathcal{P} = [(\mathrm{UR18}, 7),\;(\mathrm{RGA8}_a, 10),\;(\mathrm{XY4}, 8)].
\]
Operationally, this means: (1) apply UR18 with $\T=7$ to fill the largest qualifying gaps; (2) on the remaining gaps, apply RGA8$_a$ with $\T=10$; (3) then XY4 with $\T=8$; (4) stop, as any further DD increases error.

For a toy illustration of the $\T$ ratio consider the following: Let a single gap have length $G=600 \dt$ and available options
XX $(100 \dt)$, XY4 $(200 \dt)$, UR6 $(300 \dt)$. Then:
\begin{itemize}
 \item $\T=2$: $G/\T=300$; qualifying lengths $\{100,200,300\}$; insert UR6 $(300 \dt)$.
 \item $\T=3$: $G/\T=200$; qualifying $\{100,200\}$; insert XY4 $(200 \dt)$.
 \item $\T=4$: $G/\T=150$; qualifying $\{100\}$; insert XX $(100 \dt)$.
 \item $\T>6$: $G/\T<100$; no sequence qualifies; insert nothing.
\end{itemize}

\cref{fig-circ-example} 
shows the temporal layout of the circuit for the first QEC cycle of the \dxdz{3}{3} surface code, the result of the optimization procedure described above.
We note that the DD configurations we find are not unique and, in addition, may vary across calibration cycles. Therefore, they may need to be re-optimized with each new experiment. 

Finally, we comment on the additional DD pulses required. The single-qubit gates in the bare QEC circuits are transpiled into $\sqrt{X}$ and VZ gates; thus, any extra $X$ gates appearing after Algorithm 1 arise solely from the added DD. \cref{fig-dd-count} illustrates the number of added pulses per QEC cycle. Each cycle is a standalone circuit optimized independently. We observe that the optimal number of pulses needed to suppress $Z$-type errors for the $\overline{\ket{0}}$ state is substantially smaller than for the $\overline{\ket{+}}$ state when suppressing $X$-type errors.

\begin{figure}
\hspace{0cm}{\includegraphics[width=1\columnwidth]{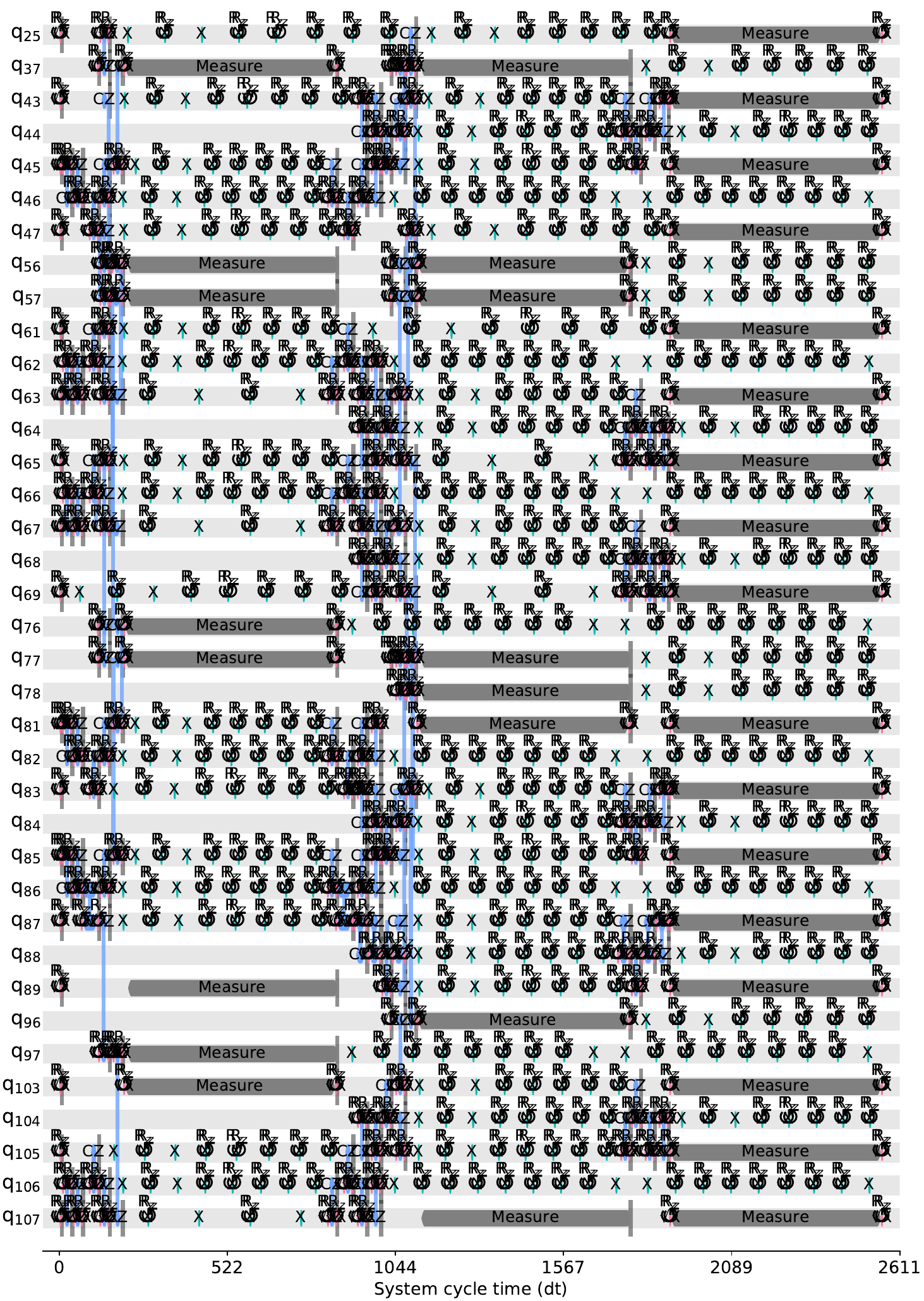}} 
\caption{A temporal depiction of one of the circuits corresponding to a single QEC cycle of the \dxdz{3}{3} surface code is shown. Blue components are the two-qubit gates (Controlled-$Z$), red components are $\sqrt{X}$ gates, and green components are $X$ gates. Instantaneous virtual $Z$ rotations are shown with the symbol $R_Z$. Various idle gaps are present throughout the circuit, which we pad with DD pulses. Our strategy prioritizes inserting longer DD sequences into the longest gaps first, followed by padding any remaining shorter gaps with appropriately shorter sequences. } 
\label{fig-circ-example}
\end{figure}

\begin{figure}
\hspace{0cm}{\includegraphics[width=.9\columnwidth]{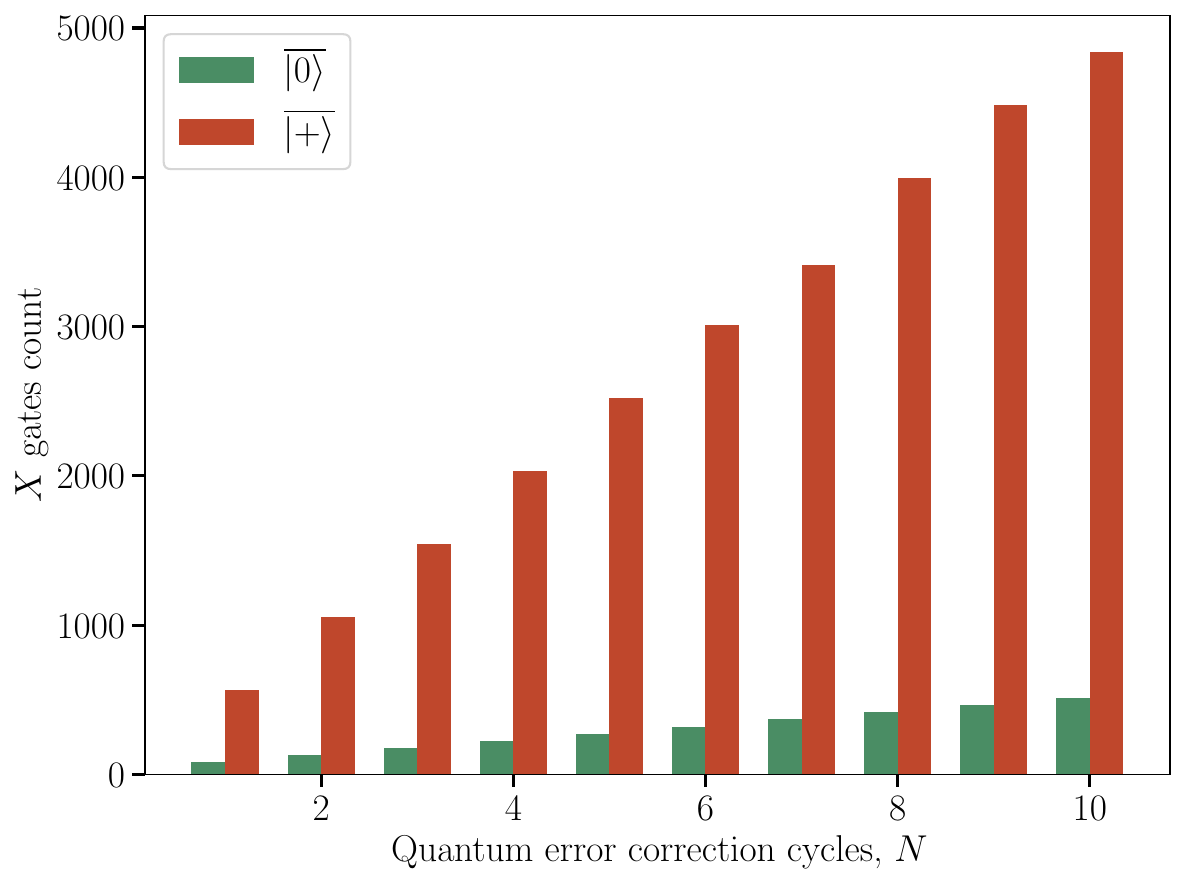}} 
\caption{Example showing the number of added single-qubit DD pulses for various QEC cycles of the $(3,3)$ code. For the $\overline{\ket{0}}$ state, the optimal number of pulses needed to suppress $Z$-type errors is much smaller than for the $\overline{\ket{+}}$ state when suppressing $X$-type errors. Similar behavior is observed for the $\overline{\ket{-}}$ and $\overline{\ket{1}}$ states (not shown).} 
\label{fig-dd-count}
\end{figure}


\section{Circuit-level noise and decoding}
\label{sec:noise_model_decod_supp}

A faithful noise model is required to realistically simulate code performance numerically and make comparative assessments of different possible strategies. In addition, a faithful noise model is also useful as input to the decoder of the experimental data. We generate the so-called detector error model (DEM) for decoding using Stim, a high-performance stabilizer circuit simulator~\cite{gidney2021stim}. 
Stim builds the DEM by calculating how the faults of each elementary gate, idle period, or measurement, all sampled stochastically from a Markovian Pauli channel, lead to specific flips of the code stabilizers. Since Pauli operators can be propagated very efficiently through Clifford circuits, these independent elementary faults yield a weighted graph where nodes are detectors (syndrome changes in space-time), edges represent the specific pairs of detectors that are flipped through the actual noisy circuit simulation, and the weights are connected to the probabilities of the circuit-level Pauli channel through the error propagation. The DEM allows for more efficient minimum-weight matching decoding, as the Pauli channel can be fed with actual calibration data, including inhomogeneous error rates, instead of assuming uniform uncorrelated flips.
A model that accurately captures the noise in the real device can thus boost the logical fidelity, as it can be more effective in identifying the correct recovery operations. We use a non-uniform circuit-level noise model with the following error channels:
\begin{itemize}
\item \textit{Two-qubit depolarizing channel} with rate $p_{2q}^{i,j}$ after any two-qubit gate between qubits $i$ and $j$.
\item \textit{Bit-flip channel} with rate $1-\sqrt{1-p_\text{m}^i}$ before measuring qubit $i$.
\item \textit{Classical bit flip channel} with rate $1-\sqrt{1-p_\text{m}^i}$ for measurements, flipping the measurement outcome without affecting the measured qubit. Thus, the total probability of an error happening while measuring qubit $i$ is $p_m^i$.
\item \textit{Biased dephasing channel} during idle periods, obtained by Pauli twirling the amplitude and phase damping errors, with rates~\cite{Geller2013PRA}
\begin{equation}
p_{\text{id},x}^i=p_{\text{id},y}^i=\frac{t_\text{id}^i}{4T_1},\quad p_{\text{id},z}^i=\frac{t_\text{id}^i}{2}\left(\frac{1}{T_2}-\frac{1}{2T_1}\right),
\label{eq:biasdephasing}
\end{equation}
where $t_\text{id}^i$ is the time during which qubit $i$ is idling.
\item \textit{No error channel} for single-qubit gates: since they are much more accurate than measurements and CNOTs, we consider them ideal. 
\end{itemize}

Aside from DD pulses, we only use single-qubit gates to perform $X$-basis measurements and to transpile the CNOT gate into a native two-qubit gate; the durations and fidelities of those operations are minimally impacted by the extra single-qubit gates. Moreover, given the use of optimized robust sequences that suppress pulse area/axis errors to a certain order, our decoder assumes ideal single-qubit gates in the DD sequences. 
Future improvements of the decoder may consider the higher-order errors of the single-qubit gates, which could be used for further optimization of the DD sequences considered in this work, provided that the decoder is provided with a noise model adapted to the non-uniform and cycle-dependent DD. 

To determine idle errors, the circuit is sliced into layers, ensuring that at most one gate is applied to each qubit inside a layer. These layers are preserved when the circuit is transpiled for execution in real hardware. The length of each layer equals the duration of the slowest gate within it. If the duration of the gate applied to qubit $i$ is $t_\text{op}^i$, the qubit will idle for a time equal to
\begin{equation}
t_\text{id}^i=\max_j\left(t_\text{op}^j\right)-t_\text{op}^i,
\end{equation}
undergoing errors according to \cref{eq:biasdephasing}.

Error rates $p^i_\text{op}$ and durations $t^i_\text{op}$ are obtained individually for every qubit using calibration data from the QPU, extracted from the IBM API at the time of job execution.

In order to numerically simulate larger-distance codes such as the $\bd=(5,5)$ code, which cannot be embedded on the existing QPUs using the SWAP-based approach, we replace the error probabilities and operation times of individual qubits with the median values across all qubits in the device:
\begin{equation}
p^i_\text{op}\rightarrow\underset{j}{\operatorname{med}}(p^j_\text{op}),\quad t^i_\text{op}\rightarrow\underset{j}{\operatorname{med}}(t^j_\text{op}).
\label{eq:mediannoise}
\end{equation}

Using the generated DEM and the syndrome data obtained experimentally, the decoder makes a prediction of which logical operators have been flipped; the experimental measurement outcomes for logical operators are corrected accordingly. A logical error occurs when the corrected measurement disagrees with the value of the logical operator in the absence of noise. Experimentally, the logical error probability $\overline{p}_N$ can be estimated by repeating the experiment $N_\text{shots}$ times, counting the number of failures $N_\text{fail}$, and calculating the relative frequencies
\begin{equation} \label{eq-def-pn}
\overline{p}_N=N_\text{fail}/N_\text{shots}.
\end{equation}
We use $N_{\rm{shots}}=3000$ shots in all our experiments. The confidence interval of $\overline{p}_N$ can be determined using the standard deviation of the sample mean over the total of $N_\text{shots}$ shots:
\begin{equation}
\Delta_{\overline{p}_N}=\frac{\sigma_{\overline{p}_N}}{\sqrt{N_\text{shots}}}=\sqrt{\frac{\overline{p}_N\left(1-\overline{p}_N\right)}{N_\text{shots}}}.
\end{equation}

Additionally, we bootstrap each $\overline{p}_N$ dataset 
by passing it into the Gaussian Process Regression (GPR) bootstrap function process \cite{Rasmussen2004}. GPR treats these points as noisy observations of an underlying smooth function, defined by a mean and a covariance (kernel) function. A single sample from the GPR posterior results in a synthetic replicate that preserves the inferred correlations among the $\overline{p}_N$.




\section{CPTP map model for the logical error probability and entanglement fidelity}
\label{sec:ent_fid_supp}

Here we rigorously derive an expression for the entanglement fidelity starting directly from the general completely positive trace-preserving (CPTP) map setting. We obtain the phenomenological fitting model as a special case and, moreover, go beyond the basis- and parameter-dependent fitting results to provide a more robust metric to assess the performance of QEC.

\subsection{General Model}
We assume a model where a single qubit is subject to a sequence of CPTP maps $\{\Phi_n\}_{n=1}^{N}$, each corresponding to a QEC cycle and described by Kraus operators $\{E^{(n)}_{i}\}_i$ that act on the \emph{logical} code subspace, satisfying the completeness condition $\sum_i E^{(n)\dagger}_{i} E^{(n)}_{i}=\mathbb{I}$.
This is a reasonable model since the system is effectively projected into a stabilizer subspace after each syndrome readout, and is effectively re-projected to the code space via the final Pauli-frame update after $N$ cycles. 
Hence, the initial system-bath state at the beginning of each cycle is a product state, and as long as slow temporal correlations such as parameter drift can be ignored between cycles, the evolution is described by a CP-divisible map~\cite{PhysRevLett.105.050403} at the beginning of each new cycle. Then 
\beq
\Psi_{N}\equiv\Phi_{N}\circ\dots\circ\Phi_{2}\circ\Phi_{1}
\eeq
is the overall composed map after $N$ QEC cycles. Note that we do not assume that the individual CPTP maps are equal, nor that they are of a particular Pauli type with a constant error rate per cycle.

\subsection{Entanglement fidelity}

The average channel fidelity of a CPTP $\Phi$ with Kraus operators $\{E_i\}$ is~\cite{Nielsen:2002aa} 

\beq
F_{\mathrm{ave}} \equiv \int d\psi \bra{\psi} \Phi (\ketbra{\psi}) \ket{\psi}\\
=\frac{\Fe +1/D}{1+1/D} ,
\eeq
where $d\psi$ is the Haar measure on a Hilbert space of dimension $D$, and $\Fe$ is the entanglement fidelity:
\beq
\label{eq:Fe-def}
\Fe = 
 \bra{\phi}
\bigl(\mathcal{I}\otimes\Phi \bigr) \bigl(\ketbra{\phi}\bigr)
 \ket{\phi} \\
 = 
 \frac{1}{D^{2}}
 \sum_i
  \Bigl|
   \Tr \bigl( E_i \bigr)
  \Bigr|^{2},
\eeq
where 
$\ket{\phi}
 = \frac{1}{\sqrt{d}}\sum_{j=0}^{D-1}
   \ket{j} \otimes\ket{j}$
is the maximally entangled purification of the fully mixed state. 

The average channel fidelity has an intuitive interpretation as the probability of preserving a pure state $\ket{\psi}$ after the action of the channel $\Phi$, averaged over all pure states. 
In the main text we used the entanglement fidelity as a metric for assessing code performance that is independent of a fitting model and is appropriate under biased noise. Here we derive an expression that relates it to the logical error probabilities output by the decoder. We also show that the entanglement fidelity has a simple operational meaning in the Pauli-noise, SPAM-free limit: it equals the probability that the cumulative logical Pauli is the identity.

\subsection{General expression for the entanglement fidelity}
We use the Pauli-transfer matrix (PTM) representation of the channel in the normalized logical Pauli basis
$A_0=I/\sqrt{2}$, $A_i=\sigma_i/\sqrt{2}$, so that
$\Tr(A_i^\dagger A_j)=\delta_{ij}$, $i,j\in\{x,y,z\}\equiv\{1,2,3\}$.
For a single-qubit state $\rho=(I+\mathbf{r} \cdot \boldsymbol\sigma)/2$ with Bloch vector
$\mathbf{r}=(r_x,r_y,r_z)^T$ (we use a superscript $T$ to denote the transpose) and components $r_i = \Tr(\sigma_i \rho)$, the PTM $S^{(n)}$ of the CPTP channel of the $n$'th QEC cycle $\Phi_n$ is the unique real $4\times4$ matrix satisfying
\beq
\begin{pmatrix}1\\ \mathbf{r}'\end{pmatrix}
=
S^{(n)}
\begin{pmatrix}1\\ \mathbf{r}\end{pmatrix},
\qquad
S^{(n)}=
\begin{pmatrix}
1 & 0\\
\mathbf{t}^{(n)} & T^{(n)}
\end{pmatrix},
\eeq
so that $\mathbf{r}'=\mathbf{t}^{(n)}+T^{(n)}\mathbf{r}$.
The matrix elements are given explicitly by
\beq
\label{eq:T}
T^{(n)}_{ij}=\Tr \bigl[A_i \Phi_n(A_j)\bigr],
\qquad
t^{(n)}_i=\Tr \bigl[A_i \Phi_n(A_0)\bigr].
\eeq
Trace preservation enforces $S^{(n)}_{00}=1$ and $S^{(n)}_{0i}=0$ for all $i$, giving the first row $[ 1\ 0\ 0\ 0 ]$. Unital channels further satisfy $\Phi_n(I)=I$ and hence $\mathbf{t}^{(n)}=\mathbf{0}$. The $3\times3$ real (not necessarily symmetric) block $T^{(n)}$ captures contractions, rotations, and cross-axis couplings. In our context, the off-diagonal entries $T^{(n)}_{ij}$ ($i\neq j$) are cross-logical terms.

After a total of $N$ cycles, the composed PTM is
\beq
\label{eq:S^N}
S^{(N)}_{\mathrm{eff}}=S^{(N)}\cdots S^{(2)} S^{(1)}
=\begin{pmatrix}1&0\\ \mathbf{t}^{(N)}_{\mathrm{eff}}&T^{(N)}_{\mathrm{eff}}\end{pmatrix},
\eeq
where
\beq
\label{eq:T-prod}
T^{(N)}_{\mathrm{eff}}=T^{(N)}\cdots T^{(2)} T^{(1)},
\eeq
with affine shift
\beq
\label{eq:Teff}
\mathbf{t}^{(N)}_{\mathrm{eff}}
=\sum_{n=1}^{N}
\Bigl(\prod_{j=n+1}^{N} T^{(j)}\Bigr)\mathbf{t}^{(n)}.
\eeq
With the initial Bloch vector being $\mathbf{r}_0$, the composed map yields
\beq
\label{eq:Bloch_N}
\mathbf{r}_N=\mathbf{t}_{\mathrm{eff}}^{(N)}+T_{\mathrm{eff}}^{(N)}\mathbf{r}_0.
\eeq

We use a result due to Ref.~\cite{Bowdrey:2002aa} (see also Ref.~\cite{Nielsen:2002aa}):
\begin{mylemma}
For a given PTM, the entanglement fidelity depends only on the $3\times3$ block 
\beq
\Fe=\frac{1+\Tr(T)}{4}.
\eeq
\end{mylemma}

\begin{proof}
Using \cref{eq:T} we have 
$T_{jj}=\Tr \bigl[A_j \Phi(A_j)\bigr] = \Tr \bigl[A_j \sum_i E_i A_j E_i^\dagger\bigr]$, so that
\beq
\Tr(T) = \sum_i \Tr \bigl[\sum_{j=1}^3 A_j E_i A_j E_i^\dagger\bigr].
\eeq
Using the Pauli completeness identity $\sum_{j=x,y,z}\sigma_j E \sigma_j = 2\Tr(E)I - E$ (valid for any $2\times 2$ matrix $E$), and recalling that $A_j = \sigma_j/\sqrt{2}$, we find
\bes
\begin{align}
\Tr(T) &=\frac12\sum_i \left(2 \bigl|\Tr(E_i)\bigr|^2-\Tr(E_i^\dagger E_i)\right)\\
&=\sum_i\bigl|\Tr(E_i)\bigr|^2-1,
\end{align}
\ees
where we used trace preservation $\sum_iE_i^\dagger E_i=I$ so $\sum_i\Tr(E_i^\dagger E_i)=\Tr(I)=2$.
Therefore, using \cref{eq:Fe-def} with $d=2$,
\beq
\Fe=\frac{1}{d^{2}}\sum_i\bigl|\Tr(E_i)\bigr|^{2}
=\frac{1}{4}\sum_i\bigl|\Tr(E_i)\bigr|^{2}
=\frac{1+\Tr(T)}{4}.
\eeq
\end{proof}
In particular, for any composition of $N$ noise channels, the final entanglement fidelity is given by 
\beq
\label{eq:Fe(N)}
\Fe(N)=\frac{1+\Tr(T^{(N)}_{\mathrm{eff}})}{4},
\eeq
with $T^{(N)}_{\mathrm{eff}}$ given by \cref{eq:T-prod}.

\subsection{General connection between entanglement fidelity and decoder output}
In general, each $T^{(n)}$ in \cref{eq:T-prod} is a real matrix, i.e., can have up to $9$ independent matrix elements. We do not have access to data that fully characterizes all $9$ elements per QEC cycle, as this would requires full process tomography at the logical level and repeated over $N$ cycles. 
However, if we assume that each $T^{(n)}$ is diagonal with entries $\eta^{(n)}_i \equiv[T^{(n)}]_{ii}$ ($i\in\{x,y,z\})$, then the final entanglement fidelity is given by 
\beq
\label{eq:F_e-diag}
\Fe(N)=\frac14 \left[1+\prod_{n=1}^{N}\eta^{(n)}_x+\prod_{n=1}^{N}\eta^{(n)}_y+\prod_{n=1}^{N}\eta^{(n)}_z\right].
\eeq
Next, we explain how to relate these expressions for the entanglement fidelity to the data we obtain from our experiments, and clarify the assumptions required to have such a diagonal $T^{(n)}$.

Experimentally, we have direct access to the decoder-reported logical error probabilities after $N$ QEC cycles of four cardinal logical states:
\beq
\overline{p}_{N,\alpha}\ \text{from }\ \overline{\ket{\alpha}}\ \text{runs},\quad \alpha\in\{0,1,+,-\}.
\eeq
Note that encoding logical $\overline{\ket{\pm {\rm i}}}$ states, and measuring in the corresponding $\overline{Y}$ basis is not straightforward in the surface code; including those states on an equal footing in a QEC metric lies beyond the reach of current hardware capabilities.

The PTM formalism uses \emph{per-cycle} quantities, whereas the decoder returns \emph{cumulative} (after $N$ cycles) logical error probabilities $\overline{p}_{N,\alpha}$. Note that in the main text we use the notation $\overline{p}_{N,\alpha}^{\bd}$ for $\overline{p}_{N,\alpha}$; here we drop the code distance parameters $\bd$, as they do not play a role in our calculations.

\begin{mylemma}
The connection between the experimental quantities $\overline{p}_{N,\alpha}$ and the diagonal PTM elements is
\bes
\label{eq:eta}
\begin{align}
\eta^{(N)}_x&=\frac{1-\overline{p}_{N,-}-\overline{p}_{N,+}}{1-\overline{p}_{N-1,-}-\overline{p}_{N-1,+}}\\
\eta^{(N)}_z&=\frac{1-\overline{p}_{N,1}-\overline{p}_{N,0}}{1-\overline{p}_{N-1,1}-\overline{p}_{N-1,0}},
\end{align}
\ees
with the convention $\overline{p}_{0,\alpha}=0$ so that the $N=1$ case reduces to $\eta_x^{(1)}=1-\overline{p}_{1,-}-\overline{p}_{1,+}$ and $\eta_z^{(1)}=1-\overline{p}_{1,1}-\overline{p}_{1,0}$.
\end{mylemma}

\begin{proof}
Assume the composed $N$-cycle logical channel has diagonal Pauli-transfer form
$T_{\mathrm{eff}}^{(N)}=\mathrm{diag} \bigl([T_{\mathrm{eff}}^{(N)}]_{xx},[T_{\mathrm{eff}}^{(N)}]_{yy},[T_{\mathrm{eff}}^{(N)}]_{zz}\bigr)$,
with an arbitrary affine shift $\mathbf{t}_{\mathrm{eff}}^{(N)}=(t^{(N)}_x,t^{(N)}_y,t^{(N)}_z)^T$.

Consider the initial state $\rho_{0,\pm j} =\frac{I\pm\sigma_j}{2}$, i.e., the logical qubit is prepared in the $\pm 1$ eigenstate of $\sigma_j$. Let $\Psi_N=\Phi_N\circ\cdots\circ\Phi_1$ and let $\rho_{N,\pm j}=\Psi_N(\rho_{0,\pm j})$ denote the resulting state after $N$ cycles. Define the ``conditional magnetization''
\beq
\label{eq:mag}
m_i^{(N)}(\pm j) \equiv \Tr (\sigma_i\rho_{N,\pm j}),
\eeq
i.e., the expectation value of $\sigma_i$ after $N$ QEC cycles, given that the initial state is the $\pm 1$ eigenstate of $\sigma_j$.
The corresponding PTM is given by \cref{eq:S^N}, and the composed map yields $\mathbf{r}_N=\mathbf{t}_{\mathrm{eff}}^{(N)}+T_{\mathrm{eff}}^{(N)}\mathbf{r}_0$ [\cref{eq:Bloch_N}]. Recall that the $i$'th Bloch vector component is $r_{N,i} = [\mathbf{r}_N]_i = \Tr(\sigma_i \rho_N)$, and note that this is also the magnetization defined in \cref{eq:mag}. Thus, for $i\in\{x,y,z\}$:
\bes
\label{eq:cond-mag_ij}
\begin{align}
m_i^{(N)}(\pm j) &= \Tr (\sigma_i\rho_{N,\pm j}) = t^{(N)}_i\pm [T_{\mathrm{eff}}^{(N)}]_{ij} \\
&= t^{(N)}_i\pm [T_{\mathrm{eff}}^{(N)}]_{ii}\delta_{ij},
\end{align}
\ees
since by assumption the PTM blocks $T^{(n)}$ are diagonal.

Let $G_N\in\{I,X,Y,Z\}$ be the net logical Pauli returned by the decoder after $N$ cycles. From the corresponding tallies, we compute $\overline{p}_{N,\alpha}$, the error probability of \emph{not} ending in $\ket{\alpha}$ after $N$ cycles, having initialized in $\ket{\alpha}$, for $\alpha\in\{0,1,+,-\}$. Thus, 
$\overline{p}_{N,+}=\Pr[G_N\in\{Z,Y\}]$ for $\ket{+}$ and $\overline{p}_{N,-}=\Pr[G_N\in\{Z,Y\}]$ for $\ket{-}$ runs, and 
$\overline{p}_{N,0}=\Pr[G_N\in\{X,Y\}]$ and $\overline{p}_{N,1}=\Pr[G_N\in\{X,Y\}]$ for $\ket{0}$ and $\ket{1}$ runs, respectively.

For $i\in\{x,y,z\}$ define
\beq
C_i(G_N)=
\begin{cases}
+1, & [G_N,\sigma_i]=0\ \ (G_N\in\{I,\sigma_i\}),\\
-1, & \{G_N,\sigma_i\}=0\ \ (G_N\in\{\sigma_j:\ j\neq i\}).
\end{cases}
\eeq
For a logical qubit prepared in the $\pm$ eigenstate of $\sigma_i$, define the binary random variable
\beq
A_{i,\pm}(G_N)\equiv \pm C_i(G_N).
\eeq
Then the conditional magnetization after $N$ cycles is
\bes
\label{eq:mi-ideal}
\begin{align}
& m_i^{(N)}(\pm i) = \mathbb E \bigl[A_{i,\pm}(G_N)\bigr]\\
&\quad = \pm\Bigl(\Pr([G_N,\sigma_i]=0)-\Pr(\{G_N,\sigma_i\}=0)\Bigr)\\
 &\quad = \pm\bigl(1-2 \Pr(\{G_N,\sigma_i\}=0)\bigr).
 \end{align}
\ees

In the case of the four basis states used in our experiments:
\bes
\label{eq:all-mxmz}
\begin{align}
m_x^{(N)}(+x)&=1-2 \Pr[G_N\in\{Z,Y\}]
 = 1-2 \overline{p}_{N,+},\\
m_x^{(N)}(-x)&=-\bigl(1-2 \Pr[G_N\in\{Z,Y\}]\bigr)
 = 2 \overline{p}_{N,-}-1 \\
m_z^{(N)}(+z)&=1-2 \Pr[G_N\in\{X,Y\}]
 = 1-2 \overline{p}_{N,0},\\
m_z^{(N)}(-z)&=-\bigl(1-2 \Pr[G_N\in\{X,Y\}]\bigr)
 = 2 \overline{p}_{N,1}-1 .
\end{align}
\ees

Comparing \cref{eq:cond-mag_ij,eq:all-mxmz}, we obtain
\bes
\label{eq:eff-decoder}
\begin{align}
\label{eq:eff-decoder-1}
[T_{\mathrm{eff}}^{(N)}]_{xx}&=1-\overline{p}_{N,-}-\overline{p}_{N,+}\\
\label{eq:eff-decoder-2}
[T_{\mathrm{eff}}^{(N)}]_{zz}&=1-\overline{p}_{N,1}-\overline{p}_{N,0}\\
\label{eq:eff-decoder-3}
t^{(N)}_{\mathrm{eff},x}&=\overline{p}_{N,-}-\overline{p}_{N,+}\\
\label{eq:eff-decoder-4}
t^{(N)}_{\mathrm{eff},z}&=\overline{p}_{N,1}-\overline{p}_{N,0}.
\end{align}
\ees
These identities are purely combinatorial statements about the decoder’s Pauli-frame tallies and do not assume any particular noise model.

Since $T^{(n)}$ is diagonal for all $n$, we have $[T_{\mathrm{eff}}^{(N)}]_{jj}=\prod_{n=1}^{N}\eta_j^{(n)}$, yielding 
\beq
\label{eq:per-cycle-etas-from-four}
\eta_x^{(N)}=\frac{[T_{\mathrm{eff}}^{(N)}]_{xx}}{[T_{\mathrm{eff}}^{(N-1)}]_{xx}},
\quad \eta_z^{(N)}=\frac{[T_{\mathrm{eff}}^{(N)}]_{zz}}{[T_{\mathrm{eff}}^{(N-1)}]_{zz}},
\eeq
whenever the denominators are nonzero. Substituting \cref{eq:eff-decoder-1,eq:eff-decoder-2} yields \cref{eq:eta}, as claimed. Note that an entirely analogous relation holds for $\eta_y^{(N)}$ if $\overline{\ket{\pm i}}$ runs are available.
\end{proof}

With access restricted to $\overline{\ket{\pm 1}}$ and $\overline{\ket{0/1}}$ experiments, a diagonal-$T$ model is the minimal assumption that makes $[T_{\mathrm{eff}}^{(N)}]_{xx}$ and $[T_{\mathrm{eff}}^{(N)}]_{zz}$ identifiable from our data; off-diagonal PTM elements would require $\overline{\ket{\pm i}}$.

\subsection{Independent $X$, $Z$, and amplitude-damping model, without SPAM}

We now consider a model in which the logical noise channel comprises a combination of logical bit-flips, phase-flips, and a non-unital logical contribution in the form of amplitude damping (AD). We defer to the next subsection the incorporation of SPAM and generalized AD (i.e., when $\overline{\ket{0}}$ is not the equilibrium state), which more fully accounts for all sources of error observed in our experiments. 

Fix a cycle index $n\in\{1,\dots,N\}$ and let
\beq
p_x(n),\quad p_z(n),\quad \gamma_n\in[0,1]
\eeq
denote, respectively, the per-cycle logical $X$-flip, $Z$-flip, and amplitude-damping parameters. Note that these are different from the decoder outputs $\bar{p}_{N,\alpha}$, and we will establish the relation between them below.

Independent $X$/$Z$ flips [with $Y$ occurring with probability $p_x(n)p_z(n)$] have Kraus operators given by 
\bes
\label{eq:Kraus-ops}
\begin{align}
E^{(N)}_0&=\sqrt{[1-p_x(N)][1-p_z(N)]} I,\\
E^{(N)}_1&=\sqrt{p_x(N)[1-p_z(N)]} X,\\
E^{(N)}_2&=\sqrt{p_z(N)[1-p_x(N)]} Z,\\
E^{(N)}_3&=\sqrt{p_x(N)p_z(N)} Y ,
\end{align}
\ees
forming the Pauli channel $\Phi^{\mathrm P}_n$.
Zero-temperature amplitude damping ($\ket{1} \to \ket{0}$ with probability $\gamma_n$, but $\ket{0} \to \ket{0}$ with probability $1$) is described by the Kraus operators
\bes
\label{eq:AD-kraus}
\begin{align}
A^{(n)}_0&=\ketb{0}{0}+\sqrt{ 1-\gamma_n }\ketb{1}{1},\\
A^{(n)}_1&=\sqrt{\gamma_n}\ketb{0}{1},
\end{align}
\ees
forming the channel $\Phi^{\mathrm{AD}}_n$.
We take the per-cycle logical channel as the composition
\beq
\label{eq:cycle-comp}
\Phi_n = \Phi^{\mathrm P}_n\circ\Phi^{\mathrm{AD}}_n,
\eeq
so a valid Kraus set for $\Phi_n$ is $\{ E^{(n)}_\mu A^{(n)}_k\}$ with $\mu=0,1,2,3$ and $k=0,1$.

Define 
\beq
\label{eq:xi-zeta-def}
\xi_n \equiv 1-2p_x(n),\quad
\zeta_n \equiv 1-2p_z(n),
\eeq
both of which are in $[-1,1]$, i.e., they are Bloch-$x$ and $z$ contractions.
The independent-$X/Z$ Pauli channel is unital and diagonal, and in the normalized Pauli basis $\{A_i\}_{i=0}^3$ it has PTM form
\beq
T^{\mathrm P}_n=\mathrm{diag}(\zeta_n, \zeta_n\xi_n, \xi_n),\quad
\mathbf{t}^{\mathrm P}_n=\mathbf{0},
\eeq
while amplitude damping has PTM form
\bes
\label{eq:AD-PTM}
\begin{align}
 T^{\mathrm{AD}}_n&=\mathrm{diag}(\sqrt{1-\gamma_n}, \sqrt{1-\gamma_n}, 1-\gamma_n)\\
\mathbf{t}^{\mathrm{AD}}_n&=(0,0,\gamma_n)^T.
\end{align}
\ees
For the composed map \cref{eq:cycle-comp} we have:
\beq
\label{eq:per-cycle-PTM}
T^{(n)}=T^{\mathrm P}_n T^{\mathrm{AD}}_n,\quad
\mathbf{t}^{(n)}=\mathbf{t}^{\mathrm P}_n+T^{\mathrm P}_n \mathbf{t}^{\mathrm{AD}}_n,
\eeq
and $\mathbf{t}^{\mathrm P}_n=\mathbf{0}$, so explicitly
\bes
\label{eq:Tn-tn-explicit}
\begin{align}
T^{(n)}&=\mathrm{diag} \bigl( \zeta_n\sqrt{1-\gamma_n}, \zeta_n\xi_n\sqrt{1-\gamma_n}, \xi_n(1-\gamma_n) \bigr)\\
\mathbf{t}^{(n)}&=(0,0,\xi_n\gamma_n)^T.
\end{align}
\ees

We now consider the composition over $N$ cycles.
Using \cref{eq:S^N,eq:T-prod,eq:Teff} and the diagonal form \cref{eq:Tn-tn-explicit}, we obtain
\beq
\label{eq:Teff-products}
T^{(N)}_{\mathrm{eff}}=T^{(N)}\cdots T^{(1)}=\mathrm{diag} \bigl(A_x^{(N)},A_y^{(N)},A_z^{(N)}\bigr),
\eeq
where
\bes
\label{eq:Teff-diag}
\begin{align}
\label{eq:Teff-diag-1}
A_x^{(N)}\equiv \bigl[T^{(N)}_{\mathrm{eff}}\bigr]_{xx}
&=\prod_{n=1}^{N}\ \zeta_n \sqrt{1-\gamma_n},\\
\label{eq:Teff-diag-2}
A_y^{(N)}\equiv \bigl[T^{(N)}_{\mathrm{eff}}\bigr]_{yy}
&=\prod_{n=1}^{N}\ \zeta_n\xi_n \sqrt{1-\gamma_n},\\
\label{eq:Teff-diag-3}
A_z^{(N)}\equiv \bigl[T^{(N)}_{\mathrm{eff}}\bigr]_{zz}
&=\prod_{n=1}^{N}\ \xi_n (1-\gamma_n),
\end{align}
\ees
and
\bes
\label{eq:teff-sum}
\begin{align}
\label{eq:teff-sum-1}
\mathbf{t}^{(N)}_{\mathrm{eff}}&=(0,0,t^{(N)}_{\mathrm{eff},z})^T\\
\label{eq:teff-sum-2}
t^{(N)}_{\mathrm{eff},z}&=\sum_{n=1}^{N} \xi_n \gamma_n \prod_{j=n+1}^{N} \xi_j (1-\gamma_j).
\end{align}
\ees
We also define
\beq
\label{eq:betaProd}
\Gamma^{(N)}\equiv \prod_{n=1}^{N}(1-\gamma_n),\quad
\xi^{(N)}\equiv \prod_{n=1}^{N}\xi_n ,
\eeq
so that 
\beq
\label{eq:AyN}
A_y^{(N)}=A_x^{(N)} \xi^{(N)}=\frac{A_x^{(N)} A_z^{(N)}}{\Gamma^{(N)}}.
\eeq
Equating \cref{eq:eff-decoder-1,eq:eff-decoder-2} with \cref{eq:Teff-diag-1,eq:Teff-diag-3}, respectively, and repeating \cref{eq:eff-decoder-4}, we obtain
\bes
\label{eq:RS-tz-from-data}
\begin{align}
\label{eq:RS-tz-from-data-1}
A_x^{(N)}&=1-\bigl( \overline{p}_{N,+}+ \overline{p}_{N,-}\bigr)\\
\label{eq:RS-tz-from-data-2}
A_z^{(N)}&=1-\bigl( \overline{p}_{N,0}+ \overline{p}_{N,1}\bigr)\\
\label{eq:RS-tz-from-data-3}
t^{(N)}_{\mathrm{eff},z}&= \overline{p}_{N,1}- \overline{p}_{N,0}.
\end{align}
\ees
This establishes the promised relation between the decoder outputs and the model parameters. 

Now, using \cref{eq:Fe(N)}, we obtain the entanglement fidelity after $N$ QEC cycles
\bes
\label{eq:Fe-AD-Pauli}
\begin{align}
\Fe(N)&=\frac14\left[ 1 + A_x^{(N)}+A_y^{(N)}+A_z^{(N)} \right] \\
&=\frac14\Bigl(1+A_x^{(N)}\Bigr)\Bigl(1+A_z^{(N)}\Bigr)\notag\\
&\quad +\frac{1}{4}\Bigl[\frac{1}{\Gamma^{(N)}}-1\Bigr]A_x^{(N)}A_z^{(N)},
\end{align}
\ees
where in the second equality we used \cref{eq:AyN}.
Thus, using \cref{eq:RS-tz-from-data}, we finally obtain the entanglement fidelity in terms of the decoder outputs:
\begin{align}
\label{eq:Fe-cycle-dep-exact-main}
&\Fe(N)
=\left(1-\frac{\overline{p}_{N,0}+\overline{p}_{N,1}}{2}\right)
 \left(1-\frac{\overline{p}_{N,+}+\overline{p}_{N,-}}{2}\right) \notag \\
&\quad+\frac{1}{4} \left[\frac{1}{\Gamma^{(N)}}-1\right]
 \bigl(1-\overline{p}_{N,0}-\overline{p}_{N,1}\bigr) 
 \bigl(1-\overline{p}_{N,+}-\overline{p}_{N,-}\bigr) ,
\end{align}
which is a precursor to \cref{eq:ent-fidelity} of the main text.

Since $\Gamma^{(N)}\le 1$, the second term is nonnegative [assuming $p_x(n),p_z(n)\le 1/2$, which is satisfied in our experiments] and we have the lower bound
\beq
\label{eq:bound_fe_AD-general}
\Fe(N)\ \ge\ 
\left(1-\frac{\overline{p}_{N,0}+\overline{p}_{N,1}}{2}\right)
\left(1-\frac{\overline{p}_{N,+}+\overline{p}_{N,-}}{2}\right).
\eeq

Let us now show how we can also express $\Gamma^{(N)}$ in terms of the decoder outputs. 
Using \cref{eq:RS-tz-from-data}, both $A_z^{(N)}$ and $t^{(N)}_{\mathrm{eff},z}$ are obtained directly from the decoder output. It follows from \cref{eq:Teff-diag-3} that
\beq
\label{eq:lambdaN}
\lambda_N\equiv\frac{A_z^{(N)}}{A_z^{(N-1)}}=\xi_N(1-\gamma_N),
\quad A_z^{(0)}\equiv 1.
\eeq
Using \cref{eq:teff-sum-2}, we can write
\bes
\begin{align}
t^{(N)}_{\mathrm{eff},z}&=\xi_N\gamma_N+\bigl[\sum_{n=1}^{N-1}\xi_n\gamma_n \prod_{j=n+1}^{N-1} \xi_j(1-\gamma_j)\bigr] \xi_N(1-\gamma_N)\\
&=\xi_N\gamma_N+t^{(N-1)}_{\mathrm{eff},z}\lambda_N .
\end{align}
\ees
Therefore,
\beq
a_N\equiv t^{(N)}_{\mathrm{eff},z}-\lambda_N t^{(N-1)}_{\mathrm{eff},z}
=\xi_N\gamma_N,
\quad t^{(0)}_{z}\equiv 0.
\label{eq:aN}
\eeq
\Cref{eq:lambdaN,eq:aN} allow us to express $\gamma_N$ in terms of the two decoder-data-derived quantities $a_N$ and $\lambda_N$:
\beq
1-\gamma_N=\frac{1}{a_N/\lambda_N+1} = \frac{1}{t^{(N)}_{\mathrm{eff},z}/\lambda_N-t^{(N-1)}_{\mathrm{eff},z}+1}.
\label{eq:gamma-from-data}
\eeq
Therefore,
\bes
\label{eq:DN-from-data}
\begin{align}
\Gamma^{(N)}&=\prod_{n=1}^{N}(1-\gamma_n)
=\prod_{n=1}^{N}\frac{1}{a_n/\lambda_n+1}\\
&=\prod_{n=1}^{N}\frac{1}{t^{(n)}_{\mathrm{eff},z}\frac{A_z^{(n-1)}}{A_z^{(n)}}- t^{(n-1)}_{\mathrm{eff},z}+1} .
\end{align}
\ees
Upon substituting \cref{eq:RS-tz-from-data-2,eq:RS-tz-from-data-3}, this expresses $\Gamma^{(N)}$ entirely in terms of the decoder outputs at each cycle $n$. Hence, the entanglement fidelity in \cref{eq:Fe-cycle-dep-exact-main} is also given entirely in terms of experimental quantities.

Note that if we choose the opposite ordering convention from \cref{eq:cycle-comp}, i.e., let $\Phi_n=\Phi^{\mathrm{AD}}_n\circ\Phi^{\mathrm P}_n$, the $T^{(n)}$ block in \cref{eq:Tn-tn-explicit} is unchanged (both factors are diagonal), while the affine shift becomes $\mathbf{t}^{(n)}=(0,0,\gamma_n)^T$. In this case, we have $a_N=\gamma_N$ instead, so that $\Gamma^{(N)}=\prod_{n=1}^N (1-a_n)$.

\subsection{Independent $X$, $Z$, and generalized amplitude-damping model, with SPAM}

The problem with the model in the previous subsection is that it predicts that $\overline{p}_{N,-}=\overline{p}_{N,+}$ [combine \cref{eq:eff-decoder-3,eq:teff-sum-1}]. However, this is not what we observe, as can be seen in \cref{fig:err-prob}(c) and (d) of the main text. Therefore, in this subsection we present a model that agrees with the AIC-based conclusion that we need the three-parameter fitting model, and is capable of explaining our data: we account for SPAM, and also replace zero-temperature amplitude damping by nonzero-temperature generalized amplitude damping (GAD). We then derive the entanglement fidelity and show that it can still be computed entirely in terms of the decoder outputs.

\subsubsection{Generalized Amplitude Damping}

We start by replacing zero-temperature AD by nonzero-temperature GAD with
\beq
\gamma_n\in[0,1],\quad u_n\in[-1,1],
\eeq
where $\gamma_n$ is the damping strength and $u_n$ is the equilibrium Bloch-$z$ polarization. The equilibrium excited-state population $p^{\infty}_n=\frac{1+u_n}{2}$; $u_n=+1$ recovers zero-temperature AD toward $\ket{0}$, $u_n=-1$ is ``inverted'' AD toward $\ket{1}$; $|u_n|<1$ corresponds to temperature $T>0$, and $u_n=0$ is the infinite-temperature limit of zero polarization.

A Kraus set for GAD is~\cite{Nielsen:2002aa}
\beq
\begin{aligned}
&K^{(n)}_0=\sqrt{p^{\infty}_n}\begin{pmatrix}1&0\\0&\sqrt{1-\gamma_n}\end{pmatrix},\quad
K^{(n)}_1=\sqrt{p^{\infty}_n}\begin{pmatrix}0&\sqrt{\gamma_n}\\0&0\end{pmatrix},\\
&K^{(n)}_2=\sqrt{q^{\infty}_n}\begin{pmatrix}\sqrt{1-\gamma_n}&0\\0&1\end{pmatrix},\quad
K^{(n)}_3=\sqrt{q^{\infty}_n}\begin{pmatrix}0&0\\\sqrt{\gamma_n}&0\end{pmatrix} ,
\end{aligned}
\eeq
where $q^{\infty}_n\equiv 1-p^{\infty}_n$.
In the normalized Pauli basis $A_0=I/\sqrt2$, $A_{1,2,3}=\sigma_{x,y,z}/\sqrt2$, GAD has PTM form
\bes
\label{eq:GAD-PTM}
\begin{align}
 T^{\mathrm{GAD}}_n&=\mathrm{diag}(\sqrt{1-\gamma_n}, \sqrt{1-\gamma_n}, 1-\gamma_n)\\
\mathbf t^{\mathrm{GAD}}_n&=(0,0,\gamma_n u_n)^T .
\end{align}
\ees

We retain the composition ordering $\Phi_n=\Phi^{\mathrm P}_n\circ\Phi^{\mathrm{GAD}}_n$, so that $T^{(n)}=T^{\mathrm P}_n T^{\mathrm{GAD}}_n$ and $\mathbf{t}^{(n)}=T^{\mathrm P}_n \mathbf{t}^{\mathrm{GAD}}_n$; then
\bes
\label{eq:GAD-per-cycle-PTM}
\begin{align}
T^{(n)}&=\mathrm{diag} \bigl(\zeta_n\sqrt{1-\gamma_n},\zeta_n\xi_n\sqrt{1-\gamma_n},\xi_n(1-\gamma_n)\bigr),\\
\mathbf t^{(n)}&=(0,0, \xi_n\gamma_n u_n)^T.
\end{align}
\ees
After $N$ cycles we again obtain \cref{eq:Teff-products,eq:Teff-diag}, but now \cref{eq:teff-sum-2} becomes
\beq
\label{eq:GAD-Teff}
t^{(N)}_{\mathrm{eff},z}=\sum_{n=1}^{N}\ \xi_n\gamma_n u_n\ \prod_{j=n+1}^{N}\xi_j(1-\gamma_j) .
\eeq
We remark that GAD differs from AD only in the appearance of the factor $u_n$ in \cref{eq:GAD-PTM,eq:GAD-per-cycle-PTM}; setting $u_n=1$ recovers the AD model, specifically \cref{eq:AD-PTM,eq:Tn-tn-explicit}.

\subsubsection{SPAM}
\label{appsec:SPAM}

Next, we model SPAM errors as a non-unital channel that modifies the logical state before QEC cycles, affecting the observed PTM:
\bes
\label{eq:SPAM-affine}
\begin{align}
&t^{(N)}_{\mathrm{obs},i}=t^{(N)}_{\mathrm{eff},i}+b_i[T^{(N)}_{\mathrm{eff}}]_{ii}\\
&[T^{(N)}_{\mathrm{obs}}]_{ii}=a_i[T^{(N)}_{\mathrm{eff}}]_{ii} ,
\end{align}
\ees
with $|a_i|\le 1$ and $|b_i|\le 1-|a_i|$. Imperfect initial-state preparation means $b_i\ne 0$ at $N=0$. Note that in our experimental setup the very first cycle corresponds to logical state preparation starting from the physical all-$\ket{0}$ state, which is why we label our cycles starting at $N=0$; this way $N=0$ is the cycle which we use to calibrate the logical SPAM parameters, and for which by definition we normalize the EF to $1$. Subsequent cycles ($N\ge 1$) are the ones we use to assess fidelity decay.

With $h(\alpha)\in\{x,z\}$ and $s(\alpha)\in\{+1,-1\}$ for $\alpha\in\{0,1,+,-\}$, let $h(+)=h(-)=x$, $h(0)=h(1)=z$, $s(+)=s(0)=+1$, and $s(-)=s(1)=-1$. 
We can then write the four cases in \cref{eq:all-mxmz} as a single identity:
\begin{equation}
\label{eq:pbar-unified}
\overline{p}_{N,\alpha} = \frac{1 - s(\alpha) m_{h(\alpha)}^{(N)} \bigl(s(\alpha) h(\alpha)\bigr)}{2} .
\end{equation}
It is convenient to define the sum/difference combinations
\bes
\label{eq:sum-diff}
\begin{align}
\label{eq:sum-diff-1}
\Sigma_x(N)&\equiv 1-(\overline{p}_{N,+}+\overline{p}_{N,-})\\
\label{eq:sum-diff-2}
\Sigma_z(N)&\equiv 1-(\overline{p}_{N,0}+\overline{p}_{N,1})\\
\label{eq:sum-diff-3}
\Delta_x(N)&\equiv \overline{p}_{N,-}-\overline{p}_{N,+}\\
\label{eq:sum-diff-4}
\Delta_z(N)&\equiv \overline{p}_{N,1}-\overline{p}_{N,0}.
\end{align}
\ees
Note that these quantities are determined purely by the decoder outputs, i.e., they are directly measurable.

Using \cref{eq:pbar-unified} we then have, for $i\in\{x,z\}$:
\bes
\label{eq:sig-del-xz}
\begin{align}
\Sigma_i(N)&=\frac12 \left[m_i^{(N)}(+i)-m_i^{(N)}(-i)\right]\\
\Delta_i(N)&=\frac12 \left[m_i^{(N)}(+i)+m_i^{(N)}(-i)\right] .
\end{align}
\ees
Using \cref{eq:SPAM-affine} and $m_i^{(N)}(\pm i)=t^{(N)}_{\mathrm{obs},i}\pm [T^{(N)}_{\mathrm{obs}}]_{ii}$ [\cref{eq:cond-mag_ij}], we obtain
\bes
\label{eq:SPAM2-cancel}
\begin{align}
\label{eq:SPAM2-cancel-1}
\Sigma_i(N)&=a_i [T^{(N)}_{\mathrm{eff}}]_{ii}=a_i A_i^{(N)}\\
\label{eq:SPAM2-cancel-2}
\Delta_i(N)&=t^{(N)}_{\mathrm{eff},i}+b_i[T^{(N)}_{\mathrm{eff}}]_{ii}= t^{(N)}_{\mathrm{eff},i}+b_iA_i^{(N)}
\end{align}
\ees
Thus, even if the channel has $t^{(N)}_{\mathrm{eff},x}=0$ (true for the Pauli-AD model above), a nonzero SPAM error $b_x$ yields $\Delta_x(N)\ne 0$ and hence
$\overline{p}_{N,-}\neq\overline{p}_{N,+}$. Likewise, finite-temperature GAD ($u_n\neq 0$) gives $t^{(N)}_{\mathrm{eff},z}\neq 0$ [\cref{eq:GAD-Teff}], which yields $\Delta_z(N)\ne 0$ and hence
$\overline{p}_{N,0}\neq\overline{p}_{N,1}$; any additional nonzero SPAM error $b_z$ adds to that splitting. 
Both effects are what we observe in \cref{fig:err-prob}(c) and (d) of the main text.

Since $[T^{(0)}_{\mathrm{eff}}]_{ii}=1$ and $t^{(0)}_{\mathrm{eff},i}=0$, combining \cref{eq:SPAM-affine} and \cref{eq:sig-del-xz} gives us that the SPAM coefficients $a_i$ and biases $b_i$ can be calibrated at $N=0$ as
\beq
\label{eq:SPAM-cal}
a_i=\Sigma_i(0),\quad b_i=\Delta_i(0) ,
\eeq
and as already implied by \cref{eq:SPAM-affine}, we assume that these values are cycle-independent.

\subsubsection{Non-unitality witnesses}

Let us define the SPAM-free ``non-unitality witnesses''
\beq
\label{eq:non-unitality-witness}
\widetilde{\Delta}_i(N) \equiv\ \Delta_i(N)-\frac{b_i}{a_i} \Sigma_i(N),\quad i\in\{x,z\}.
\eeq
Then, by construction, using \cref{eq:SPAM2-cancel,eq:SPAM-cal}, where $A_i^{(N)}=[T^{(N)}_{\mathrm{eff}}]_{ii}$ are the diagonal PTM contractions
and $t^{(N)}_{\mathrm{eff},i}$ are the affine (non-unital) shifts after $N$ cycles,
\beq
\widetilde{\Delta}_i(N)=t^{(N)}_{\mathrm{eff},i}.
\eeq
Thus, any statistically significant deviation of $\widetilde{\Delta}_i(N)$ from zero
is direct evidence of non-unitality along the $i=x,z$ axis. 

The witnesses in \cref{eq:non-unitality-witness} can be computed directly in terms of the decoder-reported logical error probabilities. Rather than reporting these witnesses, to boost the SNR we report their cumulative sums
\beq
\delta_i(N)\equiv\sum_{n=1}^{N}\widetilde{\Delta}_i(n) .
\eeq
A systematic drift away from zero is evidence of non-unitality. This is what we plot in \cref{fig:err-prob}(c) and (d).

\subsubsection{Entanglement fidelity with GAD and SPAM}

The entanglement fidelity is still given, via $\Fe(N)=\frac{1}{4}\bigl(1+\Tr T^{(N)}_{\mathrm{eff}}\bigr)$, by \cref{eq:Fe-AD-Pauli}. 
Using \cref{eq:AyN,eq:SPAM2-cancel-1}, we can write this entirely in terms of the decoder outputs and the two SPAM coefficients $a_x$ and $a_z$:
\bes
\label{eq:Fe-GAD-SPAM}
\begin{align}
\label{eq:Fe-GAD-SPAM-1}
\Fe(N)&=\frac14 \left[ 1+\frac{\Sigma_x(N)}{a_x}+\frac{\Sigma_z(N)}{a_z}
+\frac{\Sigma_x(N) \Sigma_z(N)}{a_x a_z \Gamma^{(N)}}\right] \\
\label{eq:Fe-GAD-SPAM-2}
&=\frac14\Bigl(1+\frac{\Sigma_x(N)}{a_x}\Bigr)\Bigl(1+\frac{\Sigma_z(N)}{a_z}\Bigr)\notag\\
&\quad +\frac{1}{4}\Bigl[\frac{1}{\Gamma^{(N)}}-1\Bigr]\frac{\Sigma_x(N) \Sigma_z(N)}{a_x a_z},
\end{align}
\ees
where the decoder outputs appear via the $\Sigma_i(N)$ [\cref{eq:sum-diff-1,eq:sum-diff-2}]. This is \cref{eq:ent-fidelity} of the main text.

Throughout, $\Fe(N)$ is the entanglement fidelity of the logical
channel after logical state preparation. We calibrate the 
SPAM coefficients at the logical start ($N=0$) via $a_i=\Sigma_i(0)$ and
$b_i=\Delta_i(0)$, and divide them out, so $\Fe(0)=1$ by construction. 

Since $\Gamma^{(N)}\le 1$ and $\Sigma_i(N)\ge 0$ in the physically relevant regime ($p_i(n)\le \frac12$), the second term in \cref{eq:Fe-GAD-SPAM-2} is nonnegative and
\beq
\label{eq:Fe-GAD-SPAM-ineq}
\Fe(N)\ \ge\ \frac14\Bigl(1+\frac{\Sigma_x(N)}{a_x}\Bigr)\Bigl(1+\frac{\Sigma_z(N)}{a_z}\Bigr) \equiv \Felow(N),
\eeq
i.e., the Pauli-only fidelity lower-bounds the Pauli-GAD fidelity. This is \cref{eq:Fe_bounds-lower} of the main text.

Let us now show how we can express $\Gamma^{(N)}$ (almost entirely) in terms of the decoder outputs. We have four cycle-dependent unknowns: $\{\xi_N,\zeta_N,\gamma_N,u_N\}$. At first sight, we have enough information to determine all, since the decoder returns the four probabilities $\{\overline{p}_0,\overline{p}_1,\overline{p}_+,\overline{p}_-\}$. However, in fact we can only form three SPAM-free quantities from these, as can be seen from 
\cref{eq:SPAM2-cancel-2}: since $t_{\mathrm{eff},x}^{(N)}=0$ for the Pauli-GAD channel, we find that $\Delta_x(N)=b_x$, i.e., is pure SPAM. The other three sum/difference combinations, $\{\Sigma_x(N),\Sigma_z(N),\Delta_z(N)\}$, can form SPAM-free quantities that we can use to extract $\{\xi_N,\zeta_N,\gamma_N\}$, as we show next.

As in \cref{eq:lambdaN}, define
\bes
\label{eq:lambdaN-muN-SPAM}
\begin{align}
\label{eq:lambdaN-SPAM}
\lambda_N&\equiv \frac{\Sigma_z(N)}{\Sigma_z(N-1)} =\frac{A_z^{(N)}}{A_z^{(N-1)}}
= \xi_N(1-\gamma_N),\quad A_z^{(0)}\equiv 1\\
\label{eq:muN-SPAM}
\mu_N&\equiv \frac{\Sigma_x(N)}{\Sigma_x(N-1)} =\frac{A_x^{(N)}}{A_x^{(N-1)}}
= \zeta_N\sqrt{1-\gamma_N},\quad A_x^{(0)}\equiv 1
\end{align}
\ees
where in both lines, in the second equality we used \cref{eq:SPAM2-cancel-1} and \cref{eq:Teff-diag-3} in the third equality. Note that the SPAM coefficients $a_z$ and $a_x$ cancel out, so both $\lambda_N$ and $\mu_N$ are SPAM-free quantities.

Using \cref{eq:GAD-Teff}, we can write
\bes
\begin{align}
t^{(N)}_{\mathrm{eff},z}&=\xi_N\gamma_N u_N \\
&\quad +\bigl[\sum_{n=1}^{N-1}\xi_n\gamma_n u_n \prod_{j=n+1}^{N-1} \xi_j(1-\gamma_j)\bigr] \xi_N(1-\gamma_N)\notag\\
&=\xi_N\gamma_N u_N+t^{(N-1)}_{\mathrm{eff},z}\lambda_N , \quad t^{(0)}_{\mathrm{eff},z}\equiv 0 .
\end{align}
\ees
Therefore,
\bes
\label{eq:aN-SPAM}
\begin{align}
a_N&\equiv \tilde{\Delta}_z(N)-\lambda_N\tilde{\Delta}_z(N-1)\\
&=t^{(N)}_{\mathrm{eff},z}-\lambda_N t^{(N-1)}_{\mathrm{eff},z} =\xi_N\gamma_N u_N ,
\end{align}
\ees 
Clearly, $a_N$ is also SPAM-free.

We can solve for $\gamma_N$ using \cref{eq:lambdaN-SPAM,eq:aN-SPAM}, which allows us to express it in terms of the two decoder-data-derived quantities $\lambda_N$ and $a_N$:
\beq
\label{eq:gamma-from-data-SPAM}
1-\gamma_N=\frac{1}{a_N/(\lambda_N u_N)+1}.
\eeq
Thus,
\beq
\label{eq:GammaN}
\Gamma^{(N)}=\prod_{n=1}^{N}\frac{1}{1+a_n/(\lambda_n u_n)},
\eeq
which (apart from $u_n$) expresses $\Gamma^{(N)}$ in terms of the decoder outputs at each cycle $n$: each $\lambda_n$ is given by \cref{eq:lambdaN-SPAM} in terms of the ratio of $\Sigma_z$'s, which are entirely determined by the decoder output via \cref{eq:sum-diff-2}. Likewise, each $a_n$ is given by \cref{eq:aN-SPAM} in terms of the $\Delta_z$'s, which are quantities that are entirely determined by the decoder output via \cref{eq:sum-diff-4}. The SPAM parameters $a_z$ and $b_z$ needed to compute $a_n$ are determined using \cref{eq:SPAM-cal}. 

\begin{figure*}[ht]
\hspace{0cm}{\includegraphics[width=0.9\textwidth]{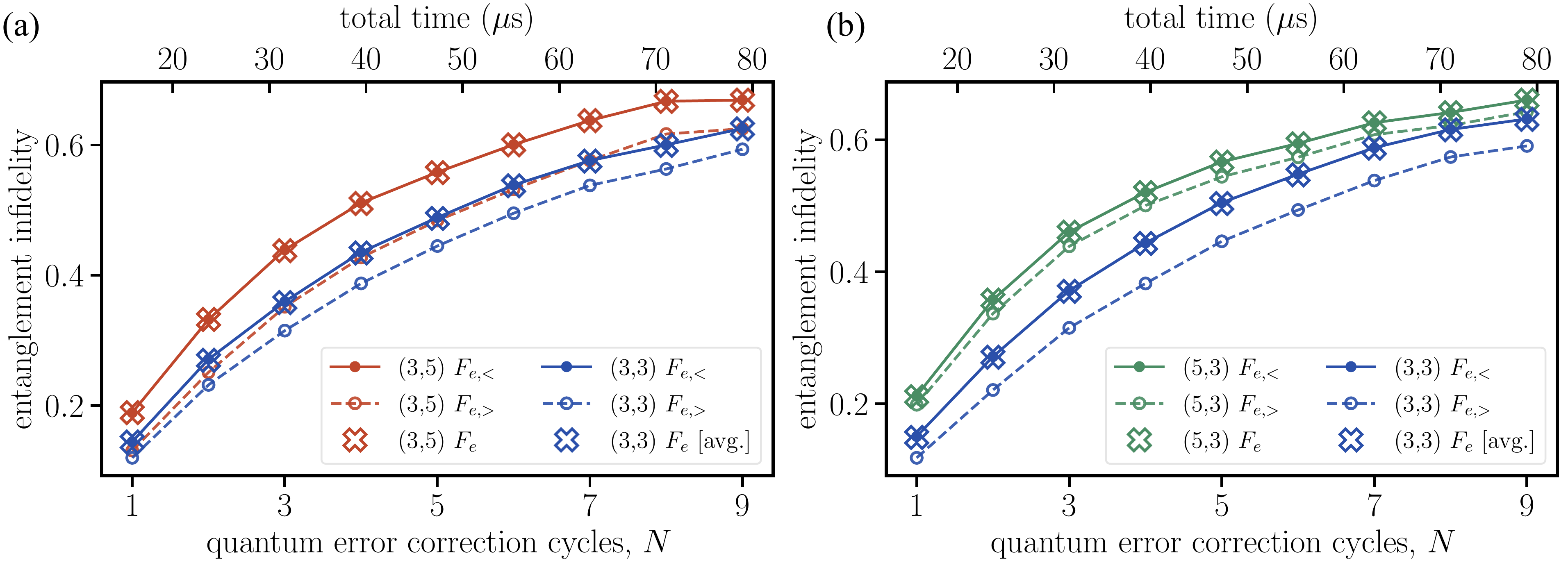}} 
\caption{The entanglement infidelity calculated from \cref{eq:ent-fidelity}, along with the lower [\cref{eq:Fe_bounds-lower}] and upper [\cref{eq:Fe_bounds-upper}] bounds for (a) scaling the $(3,5)$ code and (b) scaling the $(5,3)$ code. Same quantities are also shown for the average of (3,3) codes in each case. In both cases, the entanglement infidelity is very close to the lower bounds, suggesting that non-unital logical errors play a minor role.}
\label{fig-bounds}
\end{figure*}

The only quantity which we do not have enough information to determine from the decoder output (without measuring $\pm Y$) is the equilibrium Bloch-$z$ polarization $u_n$. To estimate it we would need to assume a relation between $\xi_n$ and $\zeta_n$, i.e., a relation between the unital contractions for $X$ and $Z$. 
Without assuming any relation between $\xi_n$ and $\zeta_n$, we can still bound $\Gamma^{(N)}$ from below. Using \cref{eq:lambdaN-muN-SPAM},
\beq
\Gamma^{(N)} = \prod_{n=1}^{N}(1 - \gamma_n) \ge \prod_{n=1}^{N}\xi_n(1 - \gamma_n) =\prod_{n=1}^{N} \lambda_n = \frac{\Sigma_z(N)}{\Sigma_z(0)}.
\eeq

Using this in \cref{eq:Fe-GAD-SPAM-2} and combining with \cref{eq:Fe-GAD-SPAM-ineq},
we obtain the upper and lower bounds
\beq
\label{eq:Fe-band}
\Felow(N) \leq \Fe(N) \leq \Fehigh(N),
\eeq
with
\beq
\label{eq:Fehigh}
\Fehigh(N) = \Felow(N) + \frac{1}{4}\left(\frac{\Sigma_z(0)}{\Sigma_z(N)} - 1\right)\frac{\Sigma_x(N)\Sigma_z(N)}{a_x a_z},
\eeq
which is \cref{eq:Fe_bounds-upper} of the main text.
This uses only SPAM-calibrated decoder outputs.

In practice, since we have no information about $u_n$, we compute the EF via \cref{eq:Fe-GAD-SPAM} by assuming that $u_n=1$, i.e., the AD model. We ensure, via a data regularization procedure, that \cref{eq:Fe-band} is always obeyed, despite shot noise. Namely, we replace 
\beq
(a_N, \lambda_N) \mapsto (\tilde{a}_N, \tilde{\lambda}_N)
\eeq
with
\beq
\tilde{a}_N = \max\{0, a_N\}, \quad \tilde{\lambda}_N = \mathrm{Clip}(\lambda_N, [0, 1]),
\eeq
where $\mathrm{Clip}[x,\{\min,\max\}]$ gives $x$ for $\min\le x\le \max$, $\min$ for $x<\min$ and $\max$ for $x>\max$. Also, if $\tilde{a}_N + \tilde{\lambda}_N > 1$, we scale both by
\beq
\text{scale} = (1 - 10^{-9})/(\tilde{a}_N + \tilde{\lambda}_N).
\eeq
The result of this regularization procedure is plotted in \cref{fig:ent-fidelity}. 

As shown in \cref{fig-bounds}, we find that in all cases $\Fe(N)$ is very close to its lower bound $\Felow(N)$, which means that the effect of logical AD errors is very small.

\subsubsection{Derivation of three-parameter fitting model from the stationary $X/Z+$GAD$+$SPAM model}
\label{sec:derive-3param-fit}

We now simplify the previous cycle-dependent model and assume that the per-cycle logical channel is constant and diagonal in the PTM formalism, with
\beq
\label{eq:r_true_N+1}
r_{N+1}^{\rm true}(i)=t_i+\eta_ir_{N}^{\rm true}(i),\quad i\in\{x,z\},
\eeq
where $r_N^{\rm true}(i)=\Tr(\sigma_i\rho_N)$ is the true Bloch component after $N$ QEC cycles, and
\bes
\label{eq:eta-t}
\begin{align}
\label{eq:etax-tx}
(\eta_x,t_x)&=\bigl(\zeta\sqrt{1-\gamma},0\bigr)\\
\label{eq:etaz-tz}
(\eta_z,t_z)&=\bigl(\xi(1-\gamma),\xi\gamma u\bigr)
\end{align}
\ees
for the standard $X/Z+$GAD composition (GAD followed by Pauli). The analog of \cref{eq:SPAM-affine} is
\beq
\label{eq:SPAM-affine-PTM}
m_i^{\mathrm{obs}(N)}(\pm i)=a_i r_N^{\rm true}(i)+b_i,
\eeq
with cycle‑independent coefficients $a_i$ and $b_i$. We note that this assumption can be in conflict with the strategy of consecutive stabilizer readouts and a final Pauli frame update, as the projected states at each intermediate QEC cycle can lie in different stabilizer subspaces and, in principle, be subject to different effects of noise, e.g., in the case of logical AD. 

For preparation in the $\pm1$ eigenstate of $\sigma_i$ we take $r_0^{\rm true}(i)=\pm 1$. 
Write \cref{eq:SPAM-affine-PTM} for $m_i^{\mathrm{obs}(N+1)}(\pm i)$, substitute \cref{eq:r_true_N+1}, and replace $r_N^{\rm true}(i)$ using \cref{eq:SPAM-affine-PTM}. This yields
\beq
\label{eq:m_i_N+1}
m_i^{\mathrm{obs}(N+1)}(\pm i)=\eta_i m_i^{\mathrm{obs}(N)}(\pm i)+\bigl[a_i t_i+(1-\eta_i)b_i\bigr].
\eeq
We can use this to derive a linear recurrence for $\overline{p}_{N,\alpha}$.

\Cref{eq:pbar-unified} still applies, and can be written as 
\beq
\label{eq:pbar-unified-2}
\overline{p}_{N,\alpha}=\frac{1-s(\alpha) M_{N}(\alpha)}{2}
\eeq
with 
\beq
\label{eq:M_N}
M_{N}(\alpha)\equiv m_{h(\alpha)}^{\mathrm{obs}(N)}(s(\alpha)h(\alpha)) .
\eeq
Thus, using \cref{eq:m_i_N+1}:
\bes
\begin{align}
M_{N+1}(\alpha)&\equiv m_{h(\alpha)}^{\mathrm{obs}(N+1)}(s(\alpha)h(\alpha))\\ 
&= \eta_{h(\alpha)} M_{N}(\alpha)+\bigl[a_{h(\alpha)} t_{h(\alpha)}+(1-\eta_{h(\alpha)})b_{h(\alpha)}\bigr]
\end{align}
\ees

Using \cref{eq:pbar-unified-2} in the form $\overline{p}_{N+1,\alpha}=\frac{1-s(\alpha) M_{N+1}(\alpha)}{2}$ and $M_{N}(\alpha) = (1-2\overline{p}_{N,\alpha})/s(\alpha)$, we obtain
\bes
\label{eq:lin-rec-p}
\begin{align}
\overline{p}_{N+1,\alpha}&=\frac12\Bigl(1-s(\alpha) \Bigl[\eta_{h(\alpha)} M_{N}(\alpha)\notag \\
 &\quad +a_{h(\alpha)} t_{h(\alpha)}+(1-\eta_{h(\alpha)})b_{h(\alpha)}\Bigr]\Bigr)\\
&= \eta_{h(\alpha)}\overline{p}_{N,\alpha}
+\varepsilon_\alpha .
\end{align}
\ees
This difference equation is in the standard inhomogeneous first‑order form, with 
\beq
\label{eq:vareps}
\varepsilon_\alpha \equiv \frac{1-\eta_{h(\alpha)}-s(\alpha)\left[a_{h(\alpha)} t_{h(\alpha)}+(1-\eta_{h(\alpha)})b_{h(\alpha)}\right]}{2}
\eeq
a constant independent of $N$.
The solution of \cref{eq:lin-rec-p} is
\bes
\label{eq:recur-sol}
\begin{align}
\label{eq:fit-three-param}
\overline{p}_{N,\alpha}&=a_\alpha+b_\alpha\Bigl(1-\frac{\varepsilon_\alpha}{a_\alpha}\Bigr)^{N}\\
\label{eq:recur-sol-2}
a_\alpha&=\frac{\varepsilon_\alpha}{1-\eta_{h(\alpha)}}\\
\label{eq:recur-sol-3}
b_\alpha&=\overline{p}_{0,\alpha}-a_\alpha,
\end{align}
\ees
which is exactly the three-parameter fitting model \cref{eq:3param-model} of the main text. We note that, in the following section, we derive this solution from 
the expected update of the logical error probabilities in consecutive QEC cycles. 

From here we can obtain the explicit parameters in the $X/Z+$GAD$+$SPAM case.
For $\alpha\in\{+,-\}$ we have $h(\alpha)=x$, $s(\pm)=\pm 1$, $\eta_x=\zeta\sqrt{1-\gamma}$, and $t_x=0$ [\cref{eq:etax-tx}], so using \cref{eq:recur-sol-2}
\beq
\varepsilon_\alpha=a_\alpha\bigl(1-\zeta\sqrt{1-\gamma}\bigr),
\eeq
and using \cref{eq:vareps}:
\bes
\label{eq:a_alpha-x}
\begin{align}
\label{eq:a_alpha-x-1}
a_\alpha&=\frac{\varepsilon_\alpha}{1-\eta_x}=\frac{1-\eta_x-s(\alpha)(1-\eta_x)b_x}{2(1-\eta_x)}\\
\label{eq:a_alpha-x-2}
&=\frac{1-s(\alpha)b_x}{2} .
\end{align}
\ees
For $\alpha\in\{0,1\}$ we have $h(\alpha)=z$, $s(0)=+1$, $s(1)=-1$,
$\eta_z=\xi(1-\gamma)$, and $t_z=\xi\gamma u$ [\cref{eq:etaz-tz}], so similarly, using \cref{eq:recur-sol-2}
\beq
\varepsilon_\alpha=a_\alpha\bigl[1-\xi(1-\gamma)\bigr] ,
\eeq
and using \cref{eq:vareps}
\bes
\label{eq:a-alpha-z}
\begin{align}
\label{eq:a-alpha-z-1}
a_\alpha&=\frac{\varepsilon_\alpha}{1-\eta_z}\\
\label{eq:a-alpha-z-2}
&=\frac{1-\eta_z-s(\alpha) \left[a_z t_z+(1-\eta_z)b_z\right]}{2(1-\eta_z)}\\
\label{eq:a-alpha-z-3}
&=\frac{1}{2} \left(1-s(\alpha) \left[b_z+\frac{a_z t_z}{1-\eta_z}\right]\right).
\end{align}
\ees

As for $b_\alpha$, at $N=0$ we take $m_i^{\mathrm{true}(0)}(\pm i)=\pm 1$.
Hence, using \cref{eq:SPAM-affine,eq:M_N}
\beq
M_0(\alpha)=m_{h(\alpha)}^{\mathrm{obs}(0)}\bigl(s(\alpha) h(\alpha)\bigr)
=a_{h(\alpha)} s(\alpha)+b_{h(\alpha)}.
\eeq
Substituting this into \cref{eq:pbar-unified-2} gives
\begin{align}
\overline{p}_{0,\alpha}
=\frac{1-s(\alpha) M_0(\alpha)}{2}=\frac{1-a_{h(\alpha)}-s(\alpha) b_{h(\alpha)}}{2},
\end{align}
where we used $s(\alpha)^2=1$.
By definition, $b_\alpha=\overline{p}_{0,\alpha}-a_\alpha$ [\cref{eq:recur-sol-3}],
so that in both cases
\beq
b_\alpha=\frac{1}{2}\Bigl[ 1-a_{h(\alpha)}-s(\alpha) b_{h(\alpha)} \Bigr]-a_\alpha,
\eeq
which is fixed by the measured $N=0$ point, and reduces to $b_\alpha=-a_\alpha$ in the absence of SPAM. More generally, $b_\alpha$ is the initial offset from the steady-state value $a_\alpha$ and sets the amplitude of the transient relaxation to this steady state.

\subsubsection{Recovering $(a_\alpha,b_\alpha,\varepsilon_\alpha)$ from the measured $\overline{p}_{N,\alpha}$}
\label{appsec:recover}

Let us write the stationary three-parameter model in the form
\begin{equation}
\label{eq:3p-model}
\overline{p}_{N,\alpha} = a_\alpha + b_\alpha q_\alpha^{N},
\quad
q_\alpha \equiv 1-\frac{\varepsilon_\alpha}{a_\alpha} .
\end{equation}
Assume that this holds for all $N\ge 0$ and note that since $\overline{p}_{N,\alpha}$ is a probability, we must have $0\le q_\alpha\le 1$. Define 
\beq
\delta p_N^{(\alpha)} \equiv \overline{p}_{N+1,\alpha}-\overline{p}_{N,\alpha}.
\eeq
Then $\delta p_N^{(\alpha)} = b_\alpha(q_\alpha-1) q_\alpha^{N}$, so the ratio of successive differences is constant for any $N\ge 0$ with $\delta p_{N}^{(\alpha)}\ne 0
$:
\begin{equation}
\label{eq:q-from-diffs}
q_\alpha = \frac{\delta p_{N+1}^{(\alpha)}}{\delta p_{N}^{(\alpha)}}
 = \frac{\overline{p}_{N+2,\alpha}-\overline{p}_{N+1,\alpha}}
  {\overline{p}_{N+1,\alpha}-\overline{p}_{N,\alpha}}.
\end{equation}
Once $q_\alpha$ is known [e.g., by averaging \cref{eq:q-from-diffs} over $N$ to reduce noise], $a_\alpha$ and $b_\alpha$ follow from any two consecutive $N\ge 0$ points:
\bes
\label{eq:ab-from-two-points}
\begin{align}
\label{eq:a-from-two-points}
a_\alpha
&=\frac{\overline{p}_{N+1,\alpha}-q_\alpha \overline{p}_{N,\alpha}}{1-q_\alpha}\\
\label{eq:b-from-two-points}
b_\alpha
&= \frac{\overline{p}_{N,\alpha}-a_\alpha}{q_\alpha^{N}}.
\end{align}
\ees
Finally, the logical error rate is
\begin{equation}
\label{eq:eps-from-aq}
\varepsilon_\alpha = a_\alpha(1-q_\alpha).
\end{equation}
Note that this implies $0\le \varepsilon_\alpha\le 1$, confirming that $\varepsilon_\alpha$ is a probability. Also note that by eliminating $b_\alpha$ between the $N=0,1$ versions of \cref{eq:3p-model}, we have
\beq
\overline{p}_{1,\alpha}=\varepsilon_\alpha + q_\alpha \overline{p}_{0,\alpha},
\eeq
so that, without SPAM (i.e., $\overline{p}_{0,\alpha}=0$), we obtain $\overline{p}_{1,\alpha}=\varepsilon_\alpha$. This means that the logical error rate $\varepsilon_\alpha$ is just the error probability after a single cycle in the absence of SPAM.

Because $\delta p_N^{(\alpha)}$ can be small and noisy, a robust estimator is to compute $q_\alpha$ by a linear fit of
$\log|\delta p_N^{(\alpha)}|$ versus $N$:
\beq
\log|\delta p_N^{(\alpha)}| = \mathrm{const} + N\log|q_\alpha|,
\eeq
using only indices with $\delta p_N^{(\alpha)}\ne 0$. We can then obtain $a_\alpha$ from
\cref{eq:a-from-two-points} by least squares over all $N$, and finally
$\varepsilon_\alpha$ from \cref{eq:eps-from-aq}. 

\Cref{eq:3p-model} imposes the constraint that for each $\alpha$, the quantities on the r.h.s. in \cref{eq:q-from-diffs,eq:ab-from-two-points} are independent of $N$. In fact, note that using \cref{eq:recur-sol-2}, we have
$q_\alpha\equiv 1-\frac{\varepsilon_\alpha}{a_\alpha}
=\eta_{h(\alpha)}$. Therefore,
\beq
\label{eq:q+-}
q_+=q_-=\eta_x,\quad q_0=q_1=\eta_z .
\eeq
Moreover, using \cref{eq:a_alpha-x-2} with $s(\pm)=\pm 1$, and similarly using \cref{eq:a-alpha-z-3} with $s(0)=+1$ and $s(1)=-1$, we have 
\beq
\label{eq:a-sum}
a_+ +a_- = a_0 + a_1 = 1 .
\eeq
These are testable predictions of the stationary (cycle-independent noise) model. A plot of the quantities on the r.h.s. in \cref{eq:ab-from-two-points,eq:eps-from-aq} is given in \cref{fig-a-alpha,fig-b-alpha,fig-epsilon-alpha}. The deviation from constancy and the constraints in \cref{eq:q+-,eq:a-sum} is a measure of the accuracy of the stationarity assumption. Clearly, the deviations are significant.

\begin{figure*}[ht]
\hspace{0cm}{\includegraphics[width=1\textwidth]{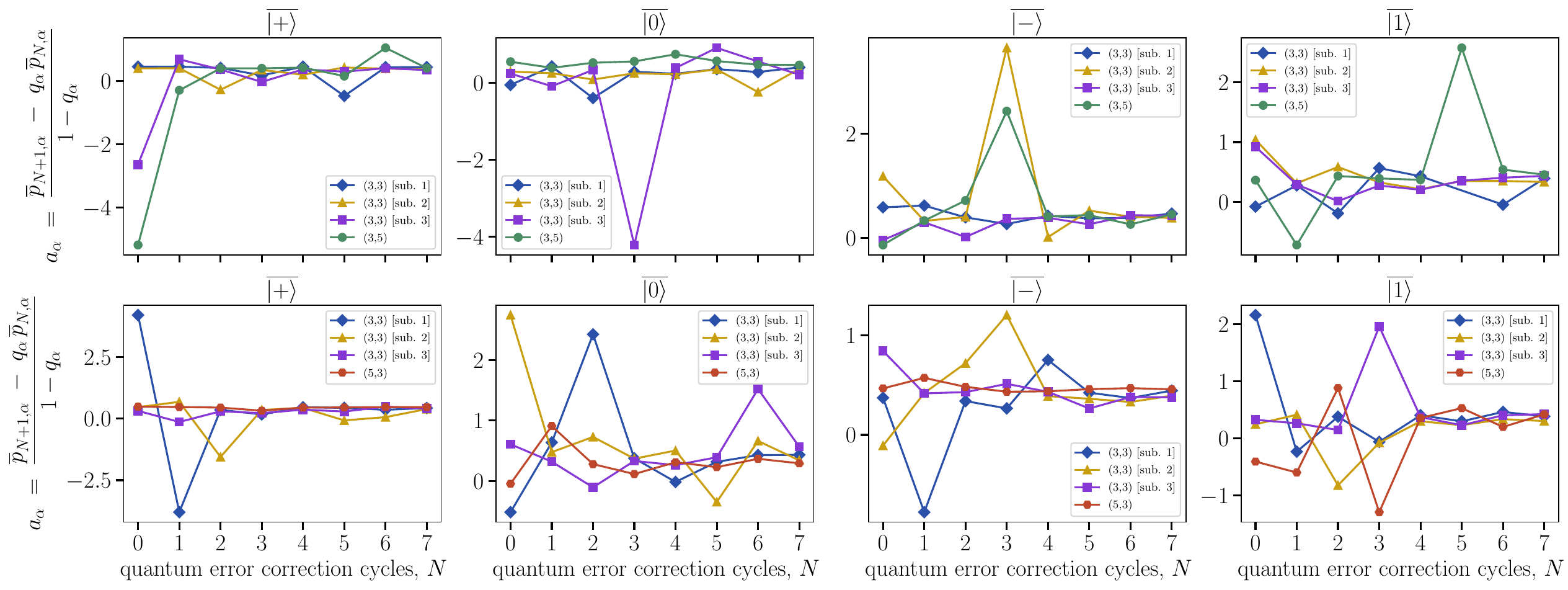}} 
\caption{The $a_\alpha$ quantity from \cref{eq:q-from-diffs} for all codes and sublattices. Any deviation from constancy is inconsistent with stationarity (cycle-independence).} 
\label{fig-a-alpha}
\end{figure*}
\begin{figure*}[ht]
\hspace{0cm}{\includegraphics[width=1\textwidth]{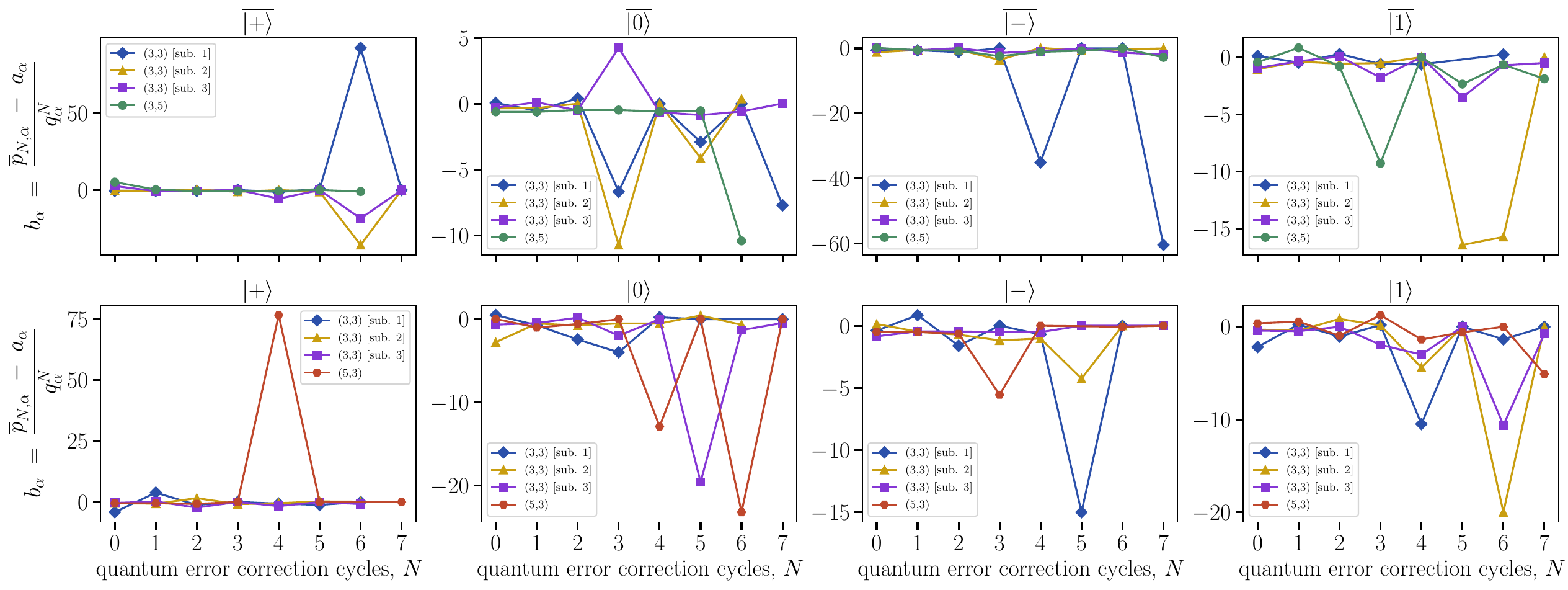}} 
\caption{The $b_\alpha$ quantity from \cref{eq:b-from-two-points} for all codes and sublattices. Any deviation from constancy is inconsistent with stationarity (cycle-independence)} 
\label{fig-b-alpha}
\end{figure*}
\begin{figure*}[ht]
\hspace{0cm}{\includegraphics[width=1\textwidth]{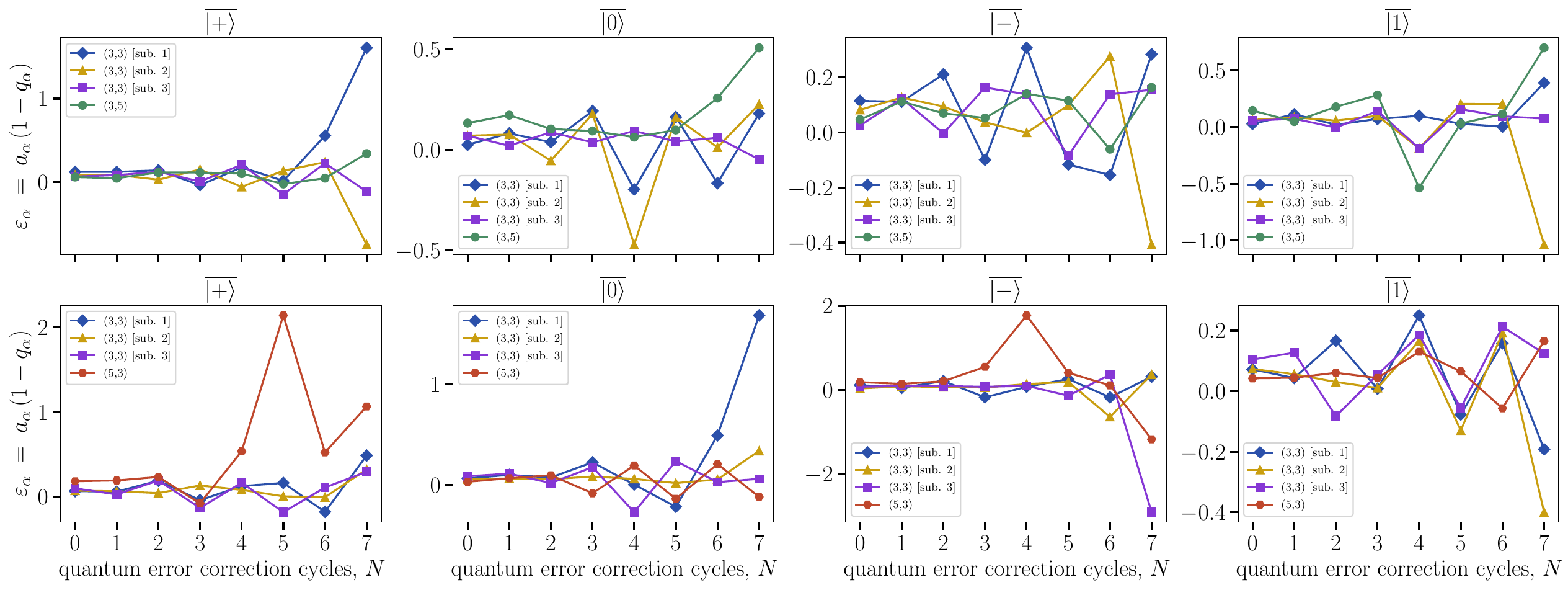}} 
\caption{The $\varepsilon_\alpha$ quantity from \cref{eq:eps-from-aq} for all codes and sublattices. Any deviation from constancy is inconsistent with stationarity (cycle-independence)} 
\label{fig-epsilon-alpha}
\end{figure*}

Recall that in the main text we defined the sublattice-dependent and average suppression factors for each basis state as $\Lambda^{\bd}_{\varepsilon,\alpha,s}\equiv\varepsilon_{\alpha,s}^{(3,3)}/\varepsilon_{\alpha}^{\bd}$ and $\Lambda^{\bd}_{\varepsilon,\alpha}\equiv\varepsilon_{\alpha}^{(3,3)}/\varepsilon_{\alpha}^{\bd}$, as well as the basis and sublattice-independent metric $\Lambda^{\bd}_{\varepsilon}$. These definitions presume stationarity, i.e., cycle-independence of the noise parameters, and in particular of the logical error rates. As \cref{fig-a-alpha,fig-b-alpha,fig-epsilon-alpha} show, the predictions of the stationary model are strongly violated in our case, which in turn renders the various suppression factors unreliable metrics.

\subsubsection{Interpretation of the traditional logical error rate $\varepsilon^{\bd}$}
\label{sec:vareps-interpretation}

The logical error rate reported in Refs.~\cite{GoogleAI2023QEC,GoogleAI2024Nature} is the average of the $X$ and $Z$ rates, i.e., $\varepsilon^d = (\varepsilon^d_+ + \varepsilon^d_-)/2$ for a code with uniform distance $\bd=(d,d)$. This assumes stationarity, unitality, and no SPAM. If we assume stationarity [$\eta_\alpha^{(n)}\equiv \eta_\alpha$ in \cref{eq:F_e-diag}] and the absence of SPAM, but not unitality, we have, using \cref{eq:q+-,eq:a-sum}, $\varepsilon_\pm=a_\pm(1-\eta_x)$, and $\varepsilon_{0/1}=a_{0/1}(1-\eta_z)$ [\cref{eq:recur-sol-2}], that $\varepsilon_+ + \varepsilon_- = 1-\eta_x$ and $\varepsilon_0 + \varepsilon_1 = 1-\eta_z$. Thus [as a special case of \cref{eq:Fe-GAD-SPAM,eq:Fe-GAD-SPAM-ineq}]:
\bes
\label{eq:Fe(1)}
\begin{align}
\label{eq:Fe(1)-1}
\Fe(1) &=\frac14\bigl(1+\eta_x+\eta_y+\eta_z\bigr)\\
\label{eq:Fe(1)-2}
&=\frac14\bigl(1+\eta_x+\eta_z\bigr) + \frac{\eta_x\eta_z}{4(1-\gamma_1)}\\
\label{eq:Fe(1)-3}
&\ge \frac14(1+\eta_x)(1+\eta_z) \\
\label{eq:Fe(1)-4}
&= \bigl(1-\frac{\varepsilon_+ + \varepsilon_-}{2}\bigr)\bigl(1-\frac{\varepsilon_0 + \varepsilon_1}{2}\bigr) .
\end{align}
\ees
This can be viewed as a justification for the ad hoc use of $\varepsilon^d = (\varepsilon^d_+ + \varepsilon^d_-)/2$, in the sense that expanding \cref{eq:Fe(1)-4} to first order in the $\varepsilon_\alpha$, and identifying $2\varepsilon^d$ as the infidelity after a single cycle:
\bes
\label{eq:1-Fe(1)}
\begin{align}
1-\Fe(1) &\approx (\varepsilon^d_+ + \varepsilon^d_- + \varepsilon^d_0 + \varepsilon^d_1)/2 \\
&\approx \varepsilon^d_+ + \varepsilon^d_0 = 2\varepsilon^d.
\end{align}
\ees
It is worth stressing that several assumptions and approximations, going well beyond what was needed to derive \cref{eq:Fe-GAD-SPAM,eq:Fe-GAD-SPAM-ineq}, are needed to obtain this result ($\varepsilon^d_+=\varepsilon^d_-$, $\varepsilon^d_0=\varepsilon^d_1$, unitality, no SPAM, etc.). Moreover, it does not hold beyond a single cycle: for $N>1$ we need to use 
\cref{eq:Fe-cycle-dep-exact-main,eq:bound_fe_AD-general} together with \cref{eq:3p-model}, and the simple interpretation of $(\varepsilon^d_+ + \varepsilon^d_-)/2$ as the average error rate is lost.


\begin{figure*}[ht]
\hspace{0cm}{\includegraphics[width=1\textwidth]{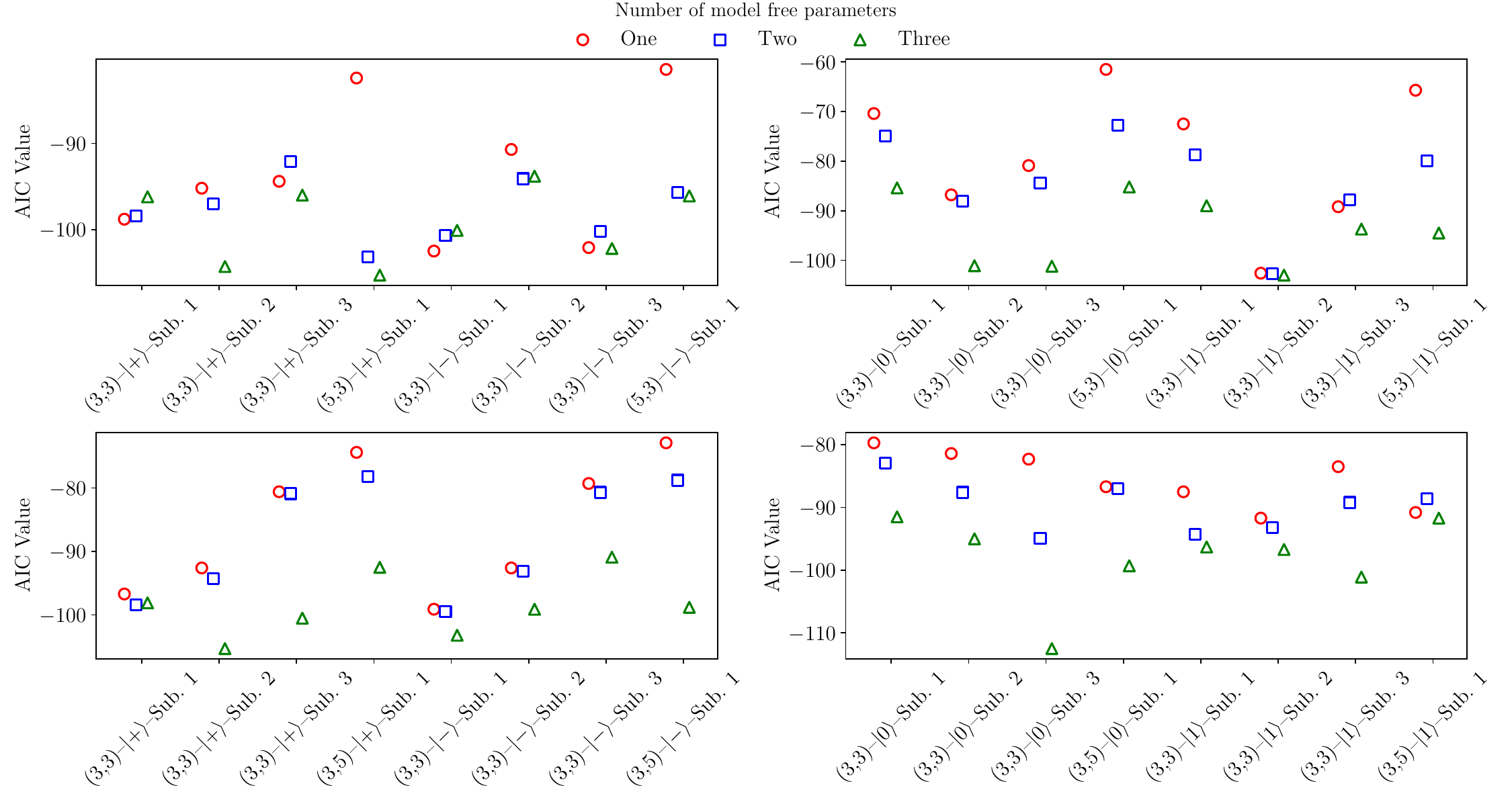}} 
\caption{Akaike Information Criterion (AIC) comparison of the different fitting models on the data from \aachen. The top row shows all the experiments (logical states) involved in scaling (5,3) code  and bottom row shows all the experiments invovled for the (3,5) code. The labels indicate, the code type, measurement basis, and the sublattice used, respectively. A lower AIC value indicates a better fit. We see that overall the three-parameter model has the best score. This means that the assumptions underlying the one- and two-parameter models (the absence of SPAM errors or the logical error limit of $\overline{p}_{N\to\infty}=0.5$) are less consistent with the data.} 
\label{fig-aic}
\end{figure*}

\subsection{Difference equation and three fitting models}
\label{sec:phen_model_supp}

\subsubsection{Derivation of \cref{eq:fit-three-param} via a difference equation}
As noted in the previous subsection, we can derive the three-parameter fitting model \cref{eq:fit-three-param} by generalizing the ad hoc assumptions of the difference equations for the logical error rates in consecutive QEC cycles briefly discussed in~\cite{GoogleAI2024Nature}. 

We assume that logical errors are incoherent and, moreover, independent between consecutive QEC cycles. This gives a recursive expression for the logical error probability after $N$ cycles, considering a quantum memory encoded in the $\alpha\in\{0,1,+,-\}$ basis: 
\begin{equation}
\label{eq:phen-model}
\overline{p}_{N+1,\alpha}=\bigl(1-\beta_\alpha\bigr)\overline{p}_{N,\alpha}+\varepsilon_{\alpha}\bigl(1-\overline{p}_{N,\alpha}\bigr),
\end{equation}
where $\varepsilon_\alpha$ is the per-cycle logical error rate, and $\beta_\alpha$ is the per-cycle recovery rate (we use the term ``rate'' since these are per-cycle quantities, even though both quantities are dimensionless probabilities). That is, $\varepsilon_\alpha$ represents the probability that a QEC cycle causes a logical error in the logical $\alpha$-basis state, and $\beta_\alpha$ is the probability of recovering from a previous logical error.
This inhomogeneous first-order linear difference equation can be solved exactly to yield \cref{eq:fit-three-param},
with $a_{\alpha}={\varepsilon_\alpha}/({\varepsilon_\alpha+\beta_\alpha})$ (assuming $\varepsilon_\alpha+\beta_\alpha<1$ to ensure convergence), and an integration constant $b_\alpha$. 

For example, if logical errors arise solely from a non-unital amplitude damping channel with dimensionless error rate $\gamma=1-\exp(-\Gamma t)=1-\exp(-t/T_1)$ at the logical level, then for an initial state $\overline{\ket{1}}$ that decays under the channel, we have $\beta_{1} = 0$ and $\varepsilon_1 = \gamma$, leading to $a_1=\lim_{N\to\infty}\overline{p}_{N,1} = 1$. In contrast, if the initial memory state is $\overline{\ket{0}}$, which is unaffected by logical amplitude damping, then $\beta_0 = \gamma$ and $\varepsilon_0 = 0$, resulting in $a_0=\lim_{N\to\infty}\overline{p}_{N,0} = 0$. 

Our derivation above is based on reasoning about the effect of non-Pauli errors such as amplitude damping at the logical level, as well as the generic effect of SPAM errors at the logical level.
The derivation of \cref{eq:fit-three-param} from a CPTP map model in \cref{sec:derive-3param-fit} is significantly longer and more elaborate, but it provides explicit expressions for the phenomenological parameters $(a_\alpha,\b_\alpha,\varepsilon_\alpha)$ in terms of the parameters of the CPTP map, which assumes $X$ and $Z$ Pauli errors, generalized amplitude damping, and SPAM errors.

\subsubsection{Assessment of the three fitting models using the Akaike Information Criterion}
\label{appsec:AIC-assessment}

We now identify the additional assumptions that need to be made to arrive at previous logical-error fitting models~\cite{GoogleAI2021Nature,GoogleAI2023QEC,GoogleAI2024Nature}. 
If logical SPAM errors are neglected, then $\overline{p}_{N=0,\alpha}=0$, since in that case the quantum memory at $N=0$ cycles has not had time to accumulate any errors from the environment or the controls. This sets $b_\alpha=-a_\alpha$ and therefore, without SPAM, the logical error probability reduces to a two-parameter model
\begin{equation} 
\label{eq:fit-two-param}
 \overline{p}_{N,\alpha} = a_\alpha(1- \bigl(1 -\varepsilon_\alpha/a_\alpha\bigr)^N).
\end{equation}

 The asymptotic limit $a_\alpha=\lim_{N\to\infty}\overline{p}_{N,\alpha}$ is entirely determined by the ratio between $\varepsilon_\alpha$ and $\beta_\alpha$, which itself depends on the noise channel and the initial state. Assuming a logical Pauli channel, then if the same Pauli error occurs twice, it cancels out. Therefore, the recovery and error rates must be equal regardless of the initial encoded state, i.e., $\beta_\alpha=\varepsilon_\alpha$. In this case, $a_\alpha=\lim_{N\to\infty}\overline{p}_{N,\alpha}=1/2$ giving
\begin{equation} 
\label{eq:fit-single-param}
 \overline{p}_{N,\alpha} = \frac{1}{2}- \frac{1}{2}\bigl(1 -2\varepsilon_\alpha\bigr)^N,
\end{equation}
which coincides with the single-parameter exponential fitting model of Refs.~\cite{GoogleAI2021Nature,GoogleAI2023QEC,GoogleAI2024Nature}. Hence, 
the latter assume a logical Pauli channel with independent logical errors per QEC cycle and no SPAM. 
\Cref{eq:fit-single-param} corresponds to the probability of having an odd number of errors after $N$ rounds~\cite{obrien2017npj}, and also agrees with the (infinite-temperature) steady state of unital channels~\cite{Lidar:2019aa}.

\label{sec:fitting_supp}

An obvious requirement for estimating the error suppression factor $\Lambda_\varepsilon^d$ from the per-cycle logical error rate, $\varepsilon_d$, is to infer the latter from the experimental data of consecutive QEC cycles on a quantum memory benchmark. This can be achieved through a fitting procedure considering some of the above expressions; and different experiments have employed varying fitting strategies, focusing on the single-parameter fitting model. For example, Ref.~\cite{GoogleAI2021Nature} fits the logical error rate to an exponential decay, while Ref.~\cite{GoogleAI2023QEC} fits the logical fidelity to an exponential. In contrast, Ref.~\cite{GoogleAI2024Nature} applies a logarithmic-scale fit to the logical fidelity. These differences reflect a clear dependence on the nature of the data, with fitting choices somewhat arbitrarily tailored to optimize interpretability or accuracy in a given experimental context. 

Rather than making an arbitrary choice, we fit our data to the one-, two-, and three-parameter versions of \cref{eq:fit-single-param,eq:fit-two-param,eq:fit-three-param}, respectively. To make an unbiased selection, we then evaluate the quality of each fit by computing and comparing their Akaike Information Criterion (AIC)~\cite{Akaike1974} scores. The AIC estimates the relative quality of statistical models by balancing the accuracy of fit against model complexity, penalizing models with more free parameters to avoid overfitting. A lower AIC value indicates a more parsimonious fit. \cref{fig-aic} shows the resulting AIC scores of the three different models. 

We find that the AIC score for each fit varies depending on the code, measurement basis, and sublattice. Although in some instances the AIC scores for the three models fall within a similar range, in most cases the three-parameter model clearly yields the best result. This implies the presence of both logical SPAM errors and logical AD, which violates the assumptions underlying the single- and two-parameter models. 
While AIC prefers a stationary three‑parameter model (indicative of non‑unital features and SPAM when stationarity is enforced), the EF analysis, which allows cycle‑dependence and removes SPAM, shows the EF tracks its Pauli lower bound closely (\cref{fig-bounds}). Thus, any effective non‑unital component (e.g., logical AD) is subdominant for the EF in our datasets. We therefore rely on EF, rather than stationary fits, to assess subthreshold behavior.

\cref{tab-fits} summarizes the suppression factors $\Lambda_{\varepsilon}^{\bd}$ we obtain from each model. Here, as in \cref{tab-res-aachen},
$\Lambda^{\bd}_{\varepsilon} = \frac{\varepsilon^{(3,3)}}{\varepsilon^{\bd}}$ where $\varepsilon^{(3,3)} = \frac14\sum_{\alpha\in\{+,-,0,1\}}\varepsilon_\alpha^{(3,3)}$ and $\varepsilon^{\bd}=\frac14\sum_{\alpha\in\{+,-,0,1\}}\varepsilon_\alpha^{\bd}$, with $\bd\in\{(3,5),(5,3)\}$.
The results clearly show that the values strongly depend on the chosen fitting model, underscoring the need for an evidence-based choice of the most appropriate model.

\begin{table*}
\centering
\renewcommand{\arraystretch}{1.5}
\setlength{\tabcolsep}{12pt}
\begin{adjustbox}{max width=\textwidth}
\begin{tabular}{c|ccc||ccc}
\hline\hline
\multirow{2}{*}{\text{backend}}
& \multicolumn{3}{c||}{$d_x = 3,~d_z = 5$}
 & \multicolumn{3}{c}{$d_x = 5,~ d_z = 3$} \\
\cline{2-7}
 & 1 parameter
 & 2 parameters
 & 3 parameters
 & 1 parameter
 & 2 parameters
 & 3 parameters \\
\hline\hline
\texttt{aachen} & $0.79$ & $0.80$ & $0.77$ & $0.70$ & $0.67$ & $0.67$ \\
\texttt{marrakesh} & $0.80$ & $0.82$ & $0.75$ & $0.73$ & $0.75$ & $0.70$ \\
\texttt{pittsburgh} & $0.73$ & $0.70$ & $0.74$ & $0.64$ & $0.65$ & $0.62$ \\
\hline
\end{tabular}
\end{adjustbox}
\caption{Comparison of $\Lambda_\varepsilon^{\bd}$ metrics for two different QPUs used for scaled surface codes using error models with a single [\cref{eq:fit-single-param}], two [\cref{eq:fit-two-param}], and three [\cref{eq:fit-three-param}] parameters. These results show that the $\Lambda$ value depends on the number of fitting parameters. Moreover, the values obtained from the three-parameter model are the lowest or tied for lowest [except for \pittsburgh\ with the $(3,5)$ code], showing that the model with the better fit is the least supportive of subthreshold scaling. All uncertainties (omitted) are $<2\times 10^{-3}$.} 
 \label{tab-fits}
\end{table*}

\begin{figure}[ht]
\includegraphics[width=.9\columnwidth]{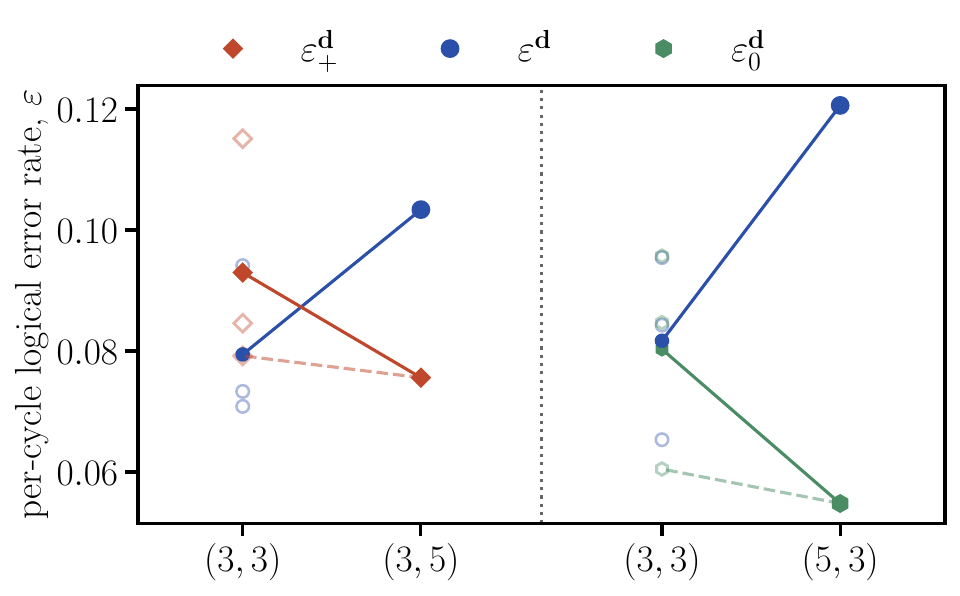}
\caption{Per-cycle logical error rate.
The values shown are obtained by fitting the data for the $\overline{\ket{+}}$ and $\overline{\ket{0}}$ states separately. 
The green hexagons and red diamonds represent the logical error rates $\varepsilon^{(3,3)}_\alpha$ (averaged over the three sublattices), $\varepsilon^{(5,3)}_\alpha$, and $\varepsilon^{(3,5)}_\alpha$, for each $(d_x,d_z)$ code and basis state $\overline{\ket{\alpha}}$, for $\alpha=0$ (green) and $\alpha=+$ (red). 
The negative slopes of the solid red and green lines show that each error type is suppressed for the sublattice average by increasing the corresponding code distance by $2$. 
This remains true with respect to the best $(3,3)$ sublattice of each scaled code, as shown by the dashed lines. 
Also shown, in blue, is the average logical error rate $\varepsilon^{\bd}=\frac14\sum_{\alpha\in\{+,-,0,1\}}\varepsilon_\alpha^{\bd}$.
Blue open markers represent $\varepsilon^{(3,3)}_s$, the corresponding rate for sublattice $s$; the left blue filled marker represents $\varepsilon^{(3,3)}$, the rate after an average over the three sublattices. The right blue filled markers represent $\varepsilon^{(3,5)}$ and $\varepsilon^{(5,3)}$. The positive slope shows that there is no code scaling advantage according to this metric.}
\label{fig:error-rates}
\end{figure}

\begin{figure*}[ht]
\hspace{0cm}{\includegraphics[width=\textwidth]{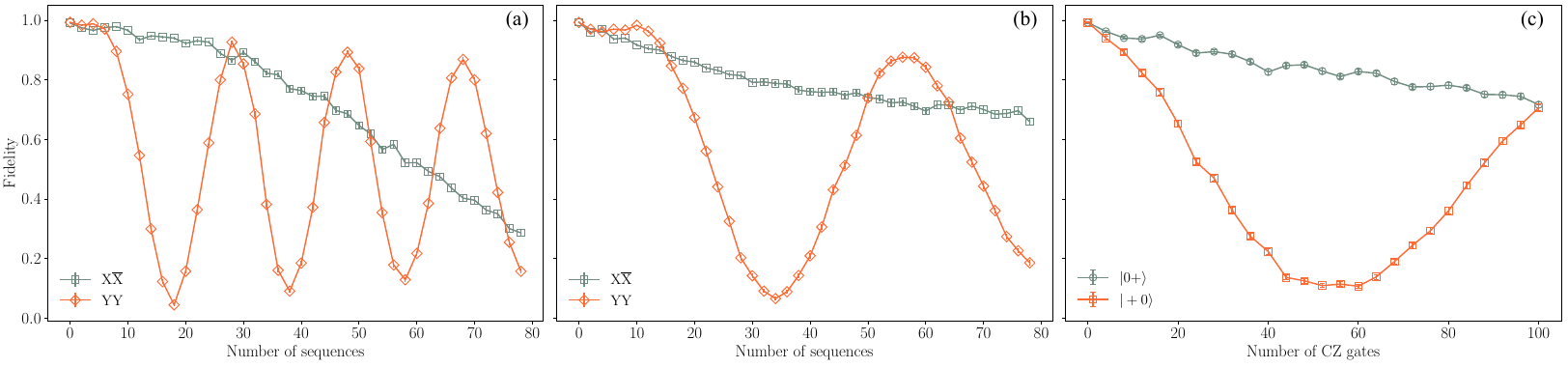}} 
\caption{Diagnostic tests were performed on qubits (a) 21, (b) 22, and (c) the coupler implementing the CZ gate between them. The sequences run in (a) $X\overline{X}$ and (b) $YY$ are sensitive to phase errors and rotation errors, respectively. The presence of pronounced oscillations indicates significant coherent errors. In (c), the CZ gate is between qubits 21 and 22 is applied to two different initial states, $\ket{0+}$ and $\ket{+0}$, which should have yielded identical fidelity curves. The difference seen reveals the presence of a significant coherent error originating from the CZ gate itself. } 
\label{fig-db}
\end{figure*}

The results of our three-parameter fits are used to calculate the error rates $\varepsilon_+^{(3,5)}$ and $\varepsilon_0^{(5,3)}$ shown in \cref{fig:error-rates} in red and green, respectively. We fit each bootstrapped $\overline{p}_N$ dataset separately and take the mean of the fitted parameters across all replicates. Therefore, $\varepsilon$ is computed as the mean over all the fitted values from each bootstrapped sample and the associated uncertainty is computed as the standard deviation across all bootstrapped samples. Using this methodology, we find the errors in $\varepsilon$ are on the order of $10^{-5}$ across all datasets, bases, and codes. We find that increasing the code distance along the direction of the error affecting the memory state, namely, using the $(5,3)$ code for $\overline{\ket{0}}$ 
and the $(3,5)$ code for $\overline{\ket{+}}$, leads to improved performance with respect to both the average and the best $(3,3)$ code. The improvement is quantified by the suppression factors reported in \cref{tab-res-aachen} of the main text. However, there is no improvement in terms of the state-averaged logical error rate $\varepsilon^{\bd}$, which coincides with the entanglement infidelity after a single cycle, computed using a simple model of independent, constant Pauli errors [\cref{eq:1-Fe(1)}].


\section{Error budget from coherent errors}
\label{appsec:coherent-errors}

In addition to decoherence, other error mechanisms, particularly coherent errors introduced by single- and two-qubit gates used in QEC circuits, can significantly impact code performance. These coherent errors have been extensively studied in the context of surface codes~\cite{Mrton2023Quantum}. Recent experimental work has shown that for single-qubit gates, such errors can be effectively characterized using \textit{deterministic benchmarking} (DB)~\cite{tripathi2024DB}, a technique that amplifies coherent errors using a small set of DD-like experiments, drastically reducing the effort required to make such errors detectable compared to randomized benchmarking~\cite{Knill:2008aa,Magesan:2011kx}, which often misses them entirely.

For a single-qubit transmon, coherent errors are most simply modeled as arising from a driven two-level system, analyzed in the drive frame under the rotating wave approximation. Letting $\{\sigma_i\}$ represent the set of Pauli matrices, the time-dependent system Hamiltonian generating single-qubit $X$ rotation gates is given by:

\bes
 \label{eq:Hsys}
 \begin{align}
 H(t) &= \epstot(t) \frac{\sigma_x}{2} + \Herr,\\
 \Herr &= \epserr \frac{\sigma_x}{2} + \delerr \frac{\sigma_z}{2}.
 \end{align}
\ees
Here, $\epstot(t)$ is the intended time-dependent control field and $\epserr$ and $\delerr$ are errors. Ideally, $\epserr = \delerr = 0$. In reality, both are present and give rise to rotation and phase errors

\begin{equation}
\dth\equiv \epserr t_g\ , \quad \dphi \equiv \frac{\delerr}{\bar{\eps}} = \frac{1}{\theta}\delerr t_g
\label{eq:dth-dphi}
\end{equation}
respectively, with $t_g$ denoting the gate duration, 
\beq
\label{eq:theta}
\theta \equiv \int_0^{t_g}\epstot(t)dt, 
\eeq
and $\bar{\varepsilon} = \theta/t_g$ the 
average pulse amplitude. An open system single-qubit gate includes both rotation and phase errors, as well as a system-bath interaction term that is always present while the gate is being generated. Ref.~\cite{tripathi2024DB} 
identified simple DD sequences that are susceptible either to phase errors
or to rotation errors. Such sequences can be used to deterministically detect the presence of these errors (in contrast to randomized benchmarking). 

Two-qubit gates, in our case, are \texttt{CZ} gates implemented using tunable couplers~\cite{Stehlik2021PRL}. The effective Hamiltonian governing these gates is given by:

\begin{equation}
 H_{\mathrm{CZ}}=\frac{\pi}{4} (ZZ-IZ-ZI).
\end{equation}
Each component of this Hamiltonian may be susceptible to coherent errors, and a full investigation of their interplay in \texttt{CZ} gates is beyond the scope of this work. For the present discussion, it suffices to note that repeated application of the \texttt{CZ} gate amplifies coherent errors, leading to characteristic oscillation patterns, as demonstrated in Ref.~\cite{Kjaergaard2022PRX}.

As an example, we study the effects of these single- and two-qubit coherent errors on qubits 21 and 22 of \marrakesh.
The results are shown in \cref{fig-db}. Panels (a) and (b) display deterministic benchmarking sequences $X\overline{X}$ and $YY$ run on each qubit, which are designed to amplify phase and rotation errors, respectively. Here $X$ and $\overline{X}$ denote $\theta=\pi$ and $\theta=-\pi$, respectively, in \cref{eq:theta}. The pronounced oscillations observed in both cases indicate strong coherent errors. Panel (c) shows the results of applying the \texttt{CZ} gate between qubits 21 and 22 from two different initial states, revealing coherent errors in gate terms that simultaneously affect both qubits. As seen from the results, while the $\ket{0,+}$ initial state exhibits no oscillations on the timescale of the experiment, the $\ket{+,0}$ state does, suggesting that a coherent error affects the first qubit as a result of the \texttt{CZ} gate operation. While most of the DD sequences used in our experiments (such as the UR$_m$ and RGA sequences~\cite{Genov2017PRL,Quiroz2013PRA}) are inherently robust against coherent errors and, as such, slow down the buildup of coherent noise, 
as the number of QEC rounds increases, this noise becomes more pronounced due to amplification, thus degrading performance and causing logical errors, as observed in \cref{fig-DD-ratio}(a) in the main text.

\begin{figure*}[ht]
\hspace{0cm}{\includegraphics[width=1\textwidth]{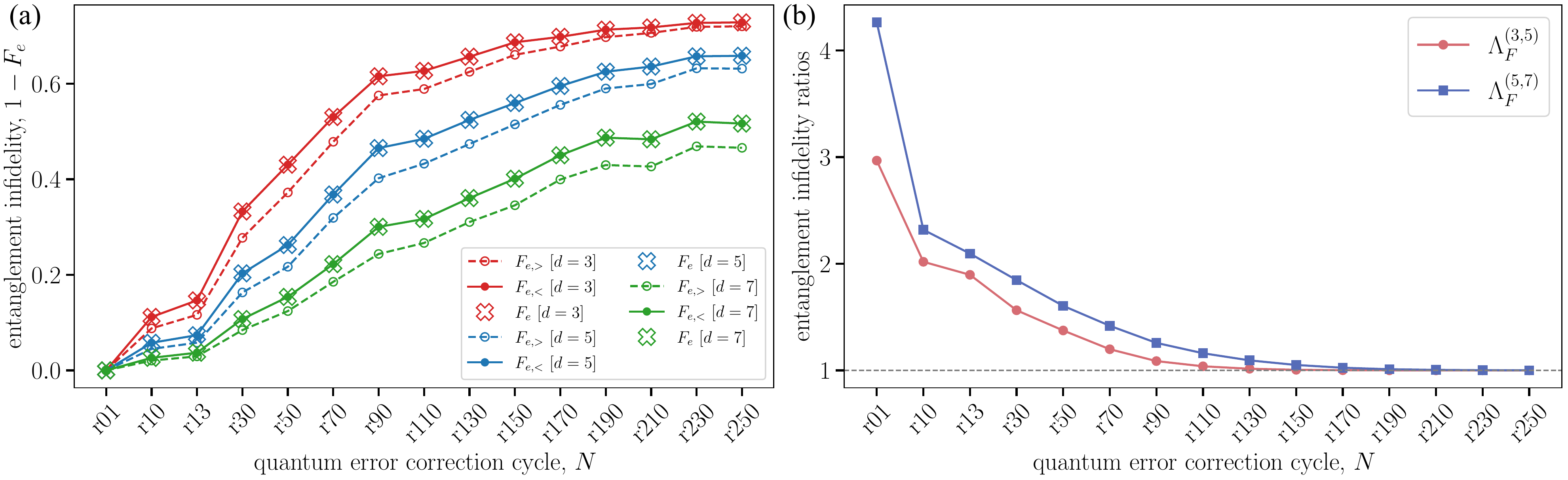}} \caption{Entanglement fidelity analysis applied to the Willow processor data~\cite{GoogleAI2024Nature}. (a) The result from \cref{eq:ent-fidelity}, combining the data from $X$ and $Z$ bases (i.e., just the $\overline{\ket{+}}$ and $\overline{\ket{0}}$ initial states, as data for $\overline{\ket{-}}$ and $\overline{\ket{1}}$ is not available). Our EF method avoids model fitting and arithmetic averaging of error rates. In this case, the EF is computed by assuming $\overline{p}_{N,0}=\overline{p}_{N,1}$ and $\overline{p}_{N,-}=\overline{p}_{N,+}$, in which case the EF coincides with its lower bound, \cref{eq:Fe_bounds-lower}. (b) The EF metric from \cref{eq:EF-metric}, for when the code is scaled from $d=3$ to $d=5$ and $d=5$ to $d=7$. A per cycle ratio $>1$ indicates an advantage achieved by the scaled code for that particular cycle. Interestingly, the $5\to 7$ ratio maintains a deeper advantage than the $3\to 5$ ratio. Note that for Ref.~\cite{GoogleAI2024Nature}, $1$ stabilizer round = $1$ QEC cycle, hence the label ``r'' on the horizontal axis.} 
\label{fig-willow-res}
\end{figure*}

\begin{figure*}[ht]
\hspace{0cm}{\includegraphics[width=.9\textwidth]{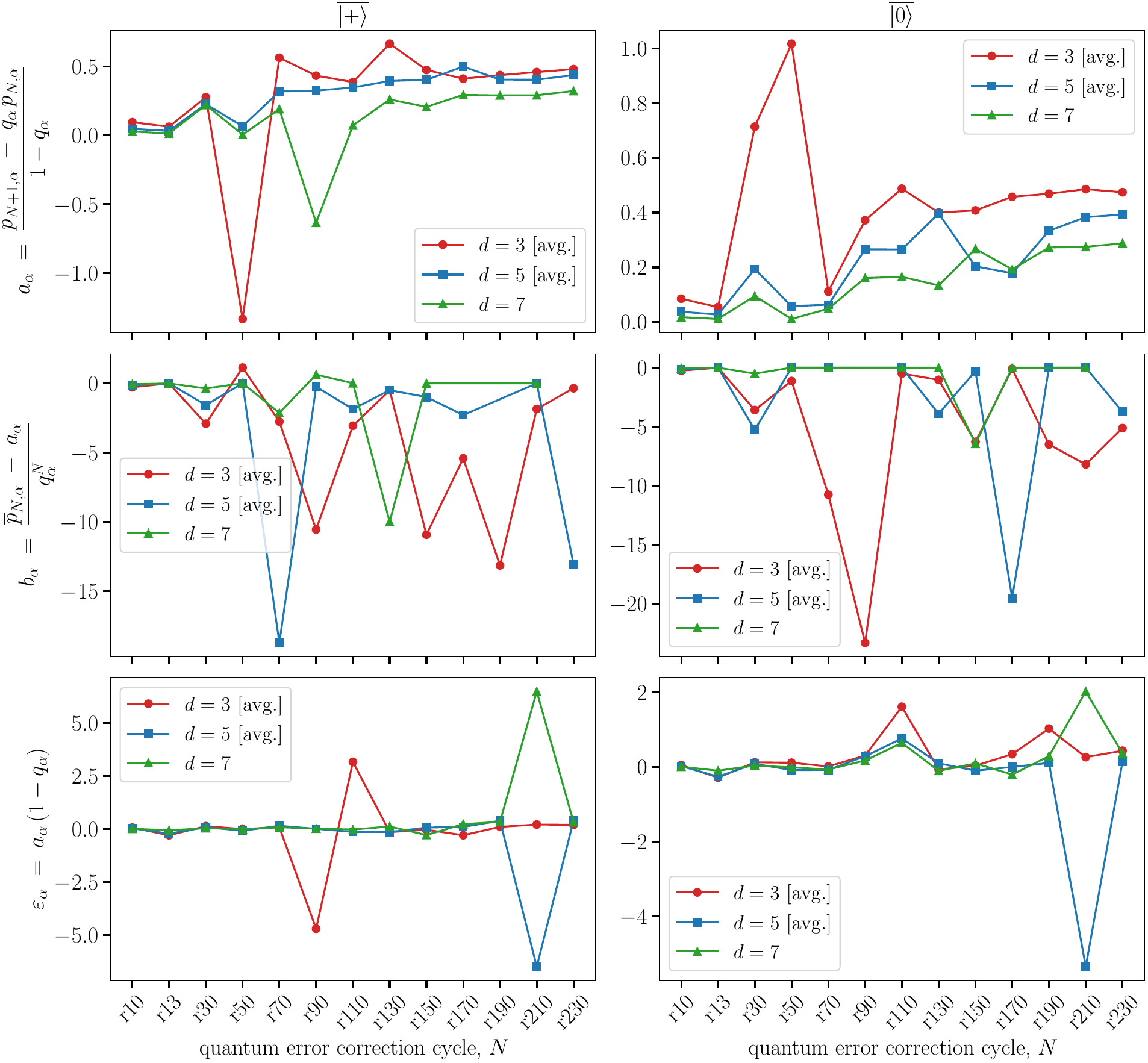}} 
\caption{The $a_\alpha$, $b_\alpha$, and $\varepsilon_\alpha$ quantities from \cref{eq:ab-from-two-points,eq:eps-from-aq} for all three codes and averaged over sublattices for $d=3$ and $d=7$. Any deviation from constancy is inconsistent with stationarity (cycle-independence)} 
\label{fig-epsilon-alpha-Willow}
\end{figure*}

\begin{figure*}[ht]
\hspace{0cm}{\includegraphics[width=.9\textwidth]{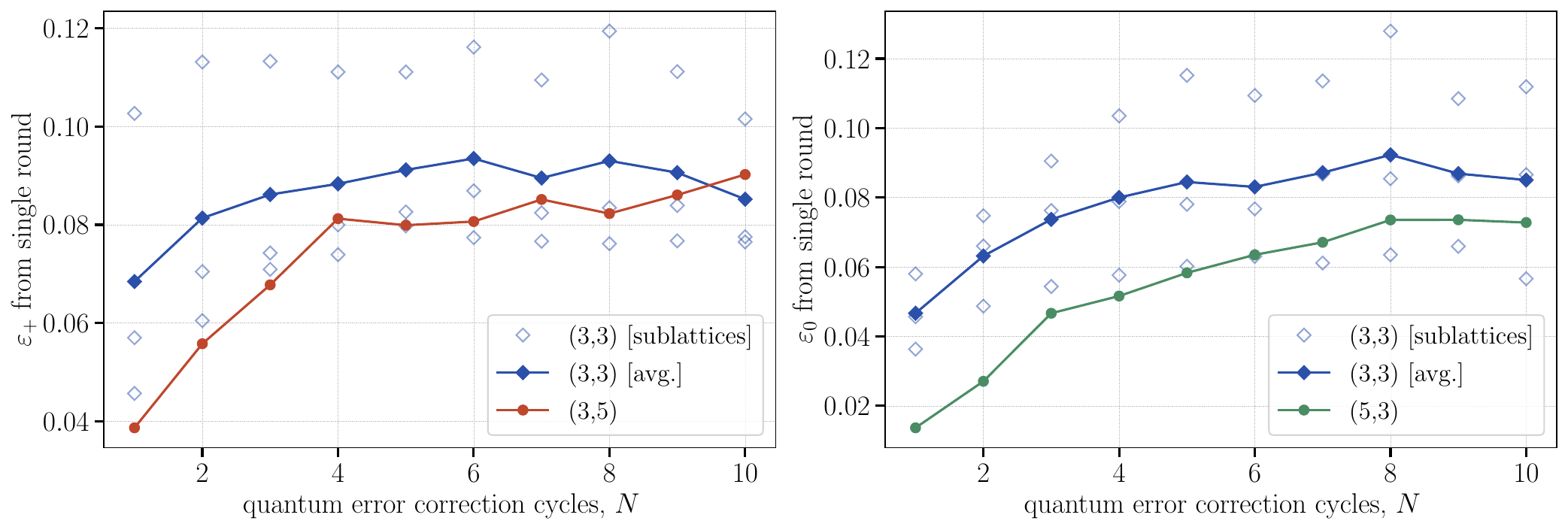}} \caption{Logical error rate from a single point, as given in \cref{eq:lepr-single-p} for scaling to the $(3,5)$ code with $\overline{\ket{+}}$ (left) and the $(5,3)$ code with $\overline{\ket{0}}$ (right). Open markers show the corresponding $(3,3)$ sublattices results. The $(5,3)$ code exhibits a more consistent advantage than the $(3,5)$ code.} 
\label{fig-lepr-single-p}
\end{figure*}

\section{Comparisons to Willow processor experiments}
\label{appsec:Willow}

\subsection{heavy-hex versus 2D square lattice}

The mismatch between the surface code and the heavy-hex lattice results in numerous idle gaps, significantly increasing the duration of each QEC cycle compared to square-type lattices such as Willow's, where operations can proceed in parallel. Consequently, our total of 
10 QEC cycles lasting $\approx80~\mu$s, corresponds to roughly 72 cycles on the Willow processor (with each Willow cycle $\approx1.1~\mu$s). In comparison, Ref.~\cite{GoogleAI2023QEC} performed 25 cycles (with a best reported $\Lambda_{\epsilon}\approx1.01$), while Ref.~\cite{GoogleAI2024Nature} performed 250 cycles (with a best reported $\Lambda_{\epsilon}\approx2.2$). We perform four times as many stabilizer measurements (20 rounds) as in Ref.~\cite{HetenyiPRXQ2024} on \texttt{ibm\_torino}.

\subsection{Entanglement fidelity analysis of Willow processor data}

To establish their claim of subthreshold code scaling, Ref.~\cite{GoogleAI2024Nature} used the fitting procedures discussed above and an arithmetic average of their $X$- and $Z$-basis experiments. We already discussed the various problematic aspects of this approach when applied to IBM QPUs in our experiments. In this section, we revisit their claim by applying our entanglement-fidelity analysis to their available open-source data.

Ref.~\cite{GoogleAI2024Nature} provides datasets for multiple decoders; the top-performing neural-network decoder was not open source at the time of writing, so we use the \textit{Libra} decoder data instead~\cite{GoogleAIData}. \cref{fig-willow-res}(a) shows the entanglement infidelity for up to $250$ QEC cycles on Willow, combining the logical error probabilities in the $X$- and $Z$-basis, using the error-channel analysis we introduced above. \cref{fig-willow-res}(b) shows per-cycle ratios for scaling $d:3\to5$ and $d:5\to7$. A ratio $>1$ indicates a per-cycle QEC advantage for the scaled code. The $5\to7$ scaling maintains an advantage for $N\approx 170$ cycles, whereas $3\to5$ maintains it for $N\approx 130$. However, note that Ref.~\cite{GoogleAI2024Nature} only reports one dataset in terms of DD inclusion, so in computing the EF metric using \cref{eq:EF-metric} we did not optimize for DD for the results shown here. Despite this, due to the absence of idle gaps in the surface code implementation of Ref.~\cite{GoogleAI2024Nature}, we do not believe that the observed scaling advantage is spurious. In fact, since we use Libra (because the neural‑network decoder was not public), our Willow-EF results are quite likely conservative.

\begin{table*}[t]
\centering
\squeezetable
\renewcommand{\arraystretch}{1.2}
\setlength{\tabcolsep}{6pt}
\begin{tabular}{lccc}
\hline\hline
 \multicolumn{1}{c}{$\Lambda_{\varepsilon}^{(d)}$} & \multicolumn{1}{c}{$3\to5$} & \multicolumn{1}{c}{$5\to7$} \\
\hline
\multicolumn{1}{c}{Ref. \cite{GoogleAI2024Nature}}  & $7.12\pm0.06\times10^{-3}/3.49\pm0.04\times10^{-3}=2.04 \pm 0.03$   & $3.49\pm0.04\times10^{-3}/1.71\pm0.03\times10^{-3}=2.04 \pm 0.03$ \\
\cref{eq:fit-single-param} & $7.24\pm0.03\times10^{-3}/3.42\pm0.03\times10^{-3}=2.11 \pm 0.02$   & $3.42\pm0.03\times10^{-3}/1.89\pm0.05\times10^{-3}=1.81 \pm 0.05$ \\
\hline\hline
\end{tabular}
  \caption{The average logical error probabilities per QEC cycle, $\varepsilon_{d}$, and the corresponding error suppression factors $\Lambda_{\varepsilon}^{(d)} \equiv \varepsilon_{d}/\varepsilon_{d+2}$ for the Willow data of Ref.~\cite{GoogleAI2024Nature}. Top row: the data from Ref.~\cite{GoogleAI2024Nature}, using the log-fit procedure. Bottom row: the error rates and suppression factors calculated in this work by fitting the data using \cref{eq:fit-single-param}.}
  \label{tab-Willow-eps}
\end{table*}

The $\varepsilon_d$ values reported for the \textit{Libra} decoder in Ref.~\cite{GoogleAI2024Nature} are based on the one-parameter model given in \cref{eq:fit-single-param} and are found through a log-fit procedure. That is, the log of the measured logical error probabilities is fitted to a linear model. The reported per-cycle logical error rates, along with the corresponding error suppression factors, are summarized in the top row of \cref{tab-Willow-eps}. This leads to equal suppression factors of $2.04$ for both code scalings. We find that the single- and two-parameter models have similarly lowest AIC values for the Willow data. As such, we directly fit the single-parameter model in its exponential form to the data. The per-cycle logical error rates and the suppression factors computed using this method are reported in the bottom row of \cref{tab-Willow-eps}. Unlike the log-fit results, we find that the suppression factors are different for the two code scalings, once again highlighting the fit-dependent ambiguity of reporting suppression factors.

The coarse-grained nature of the output of the fitting procedure, summarized by a single number per code pair, leaves out important information visible in \cref{fig-willow-res}, in particular the dependence on the QEC cycle number and the different resulting behaviors of the $3\to 5$ versus the $5\to 7$ scaling. 

\Cref{fig-epsilon-alpha-Willow} reports the results of our stationarity test of the Willow data, using the analysis of \cref{appsec:recover}, and in analogy to \cref{fig-a-alpha,fig-b-alpha,fig-epsilon-alpha}. Clearly, stationarity does not hold. We already discussed, in the context of our own data, the adverse implications of this result for the validity of the suppression factor metric.

\subsection{Logical error per cycle from one point}
\label{sec:rate-from-single-point}

Ref.~\cite{GoogleAI2024Nature} explores the concept of calculating the logical error rate from a single value of the logical error probability $\overline{p}_N$. That is, assuming the single-parameter error model [\cref{eq:fit-single-param}], we may invert it to find 
\begin{equation} 
\label{eq:lepr-single-p}
  \varepsilon_{N,\alpha} = \frac{1}{2}- \frac{1}{2}\bigl(1 -2\overline{p}_{N,\alpha}\bigr)^N,
\end{equation}
where now, the logical error rate becomes cycle dependent, which means that we can define a cycle-dependent suppression factor: $\Lambda^{\bd}_{\varepsilon_{N,\alpha}}=\varepsilon_{N,\alpha}^{(3,3)}/\varepsilon_{N,\alpha}^{\bd}$. Note that this is fundamentally different from the suppression factor we introduced in the main text, $\Lambda^{\bd}_{{p},\alpha}(N) = \overline{p}_{N,\alpha}^{(3,3)}/\overline{p}_{N,\alpha}^{\bd}$. The single-parameter model in \cref{eq:lepr-single-p} implicitly imposes several assumptions about the experimental setup (see \cref{sec:phen_model_supp}). In contrast, $\Lambda^{\bd}_{p,\alpha}(N)$ depends solely on empirical data from each experiment. 

Despite this, we apply this methodology to our data for comparison. The results are shown \cref{fig-lepr-single-p}. We find that the variation in the logical error rate is high compared to the relatively uniform behavior seen in Ref.~\cite[Fig.~S24]{GoogleAI2024Nature}. For the $(5,3)$ code, the scaling advantage appears consistent across all QEC cycles, whereas for the $(3,5)$ code it tapers off as $N$ increases. 

Ref.~\cite{Bluvstein2025arxiv}, whose experiment includes four QEC cycles, reports $\Lambda^{\bd}_{\varepsilon_{4,\alpha}}=2.14(13)$, extracted from the fourth cycle. They further simulate additional rounds and find that the ratio decreases with more QEC cycles. Given the ambiguity such an analysis introduces, as we have argued in detail above, using $\Lambda^{\bd}_{p,\alpha}(N)$ and the EF metric appears preferable.

\begin{table}[]
\begin{tabular}{l|c|c|c|c|c|c|c|c}
\hline\hline
 \multirow{2}{*}{Scaling factor $\chi$} & \multicolumn{4}{c|}{DD} & \multicolumn{4}{c}{no DD} \\ 
 & $\overline{\ket{+}}$ & $\overline{\ket{-}}$ & $\overline{\ket{0}}$ & $\overline{\ket{1}}$ & $\overline{\ket{+}}$ & $\overline{\ket{-}}$ & $\overline{\ket{0}}$ & $\overline{\ket{1}}$\\
 \hline\hline
$(5,3)$ code 
& 2.45 & 2.5 & 1.85 & 1.9 & 3.6 & 3.7 & 1.85 & 1.9 \\
$(3,5)$ code
& 2.3 & 2.3 & 2.25 & 2.3 & 4.0 & 4.0 & 2.25 & 2.35 \\
\hline\hline
\end{tabular}
\caption{Scaling factor  $\chi$ that needs to be applied to the error rates of the noise model, initially obtained from randomized benchmarking, in order for our simulations to match the experimental performance of the surface code. Shown are the required $\chi$ factors with and without DD, for the $X$ and $Z$ bases and the two scaled codes.}
\label{tab-noise-scaling}
\end{table}

\section{Simulating a distance-5 code}
\label{sec:d_5_simulations}

As we reported, we observed subthreshold behavior when scaling the surface code for a single error type. We evaluated the logical error probabilities of a \dxdz{3}{3} code against a \dxdz{3}{5} (\dxdz{5}{3}) code and found enhanced performance against phase flip (bit flip) errors. On the other hand, we found an expected decrease of the protection against bit flip (phase flip) errors, as the \dxdz{3}{5} (\dxdz{5}{3}) code is not scaled in that direction, and hence degrades from the increased complexity of the circuits. To achieve universal subthreshold behavior, it is necessary for logical error rates to be concurrently reduced for both error types, thereby increasing the protection of any logical state. This necessitates a \dxdz{5}{5} code, which is infeasible with the qubit-count and connectivity of current IBM QPUs.

In light of this, we instead turn to numerical simulations of a \dxdz{5}{5} code, comparing its performance to that of the \dxdz{3}{3} code. This task requires an accurate noise model that matches the experimental results. However, individual per-qubit error probabilities are currently unavailable at the scale of the \dxdz{5}{5} code. Instead, we can use a noise model in which all physical qubits have identical performance, set to the median performance of all qubits in the device, as per \cref{eq:mediannoise}. However, we find that the noise model defined in this way does not agree with our experimental data for the smaller codes. To mitigate this, we multiply all error probabilities appearing in the Pauli channels of the noise model by a scaling factor $\chi$.
In doing so, we obtain a noise-scaled model with $p_i\rightarrow\chi p_i$, $p_{ij}\rightarrow\chi p_{ij}$ for all noise channels. In order to determine $\chi$, we perform numerical simulations to find the optimal value that maximizes the agreement with the experimental data. 
\begin{figure*}
\includegraphics[width=0.9\linewidth]{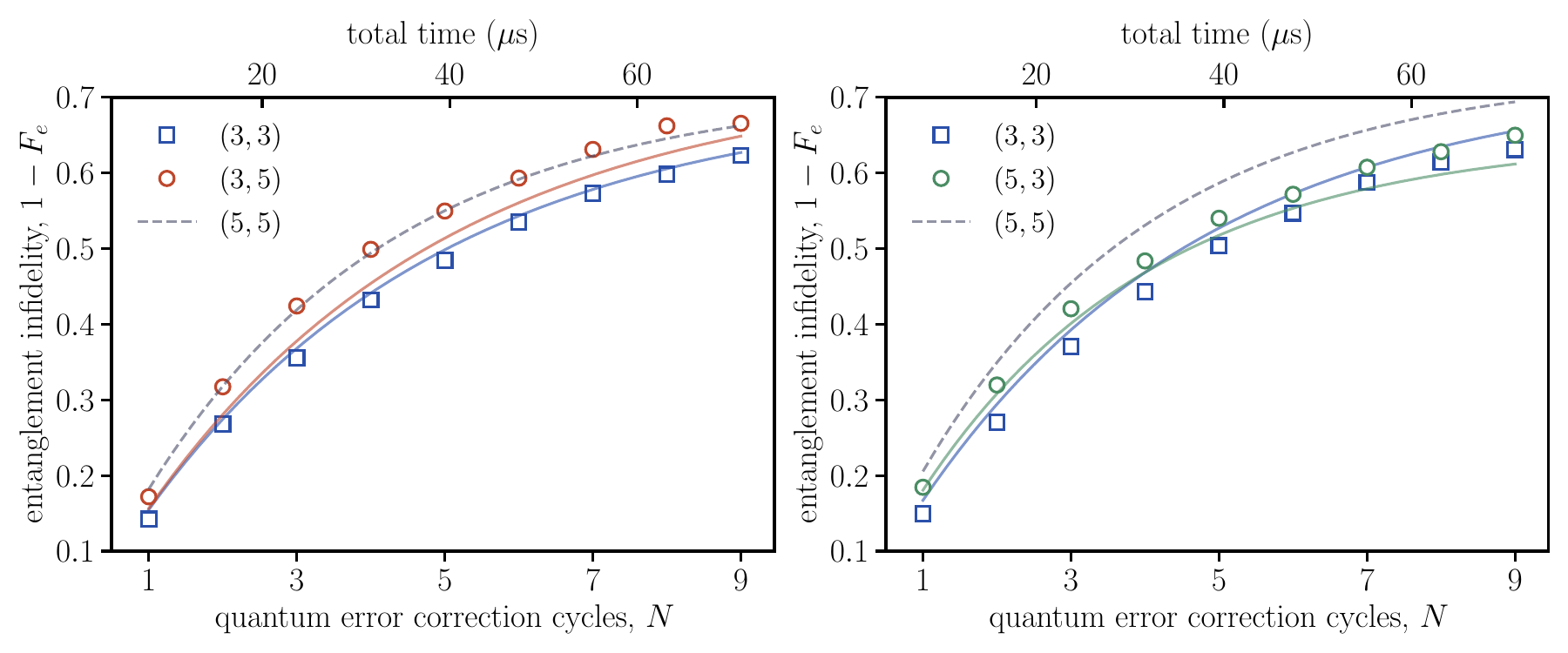}
\caption{Entanglement infidelity $1-F_e$ computed using \cref{eq:ent-fidelity} as a function of the number of QEC cycles $N$. Data points correspond to the $(3,5)$ (left) and the $(5,3)$ (right) experiments we ran on \aachen\ with DD, along with the sublattice average of their respective $(3,3)$ codes (identical to the data in \cref{fig:ent-fidelity}), and solid lines to the corresponding simulations using the $\chi$ factors with DD indicated in \cref{tab-noise-scaling}. The predicted performance of the isotropic $d=5$ code is indicated by the dashed lines.}
\label{fig-distance-5-sim}
\end{figure*}

\begin{figure}
\includegraphics[width=\linewidth]{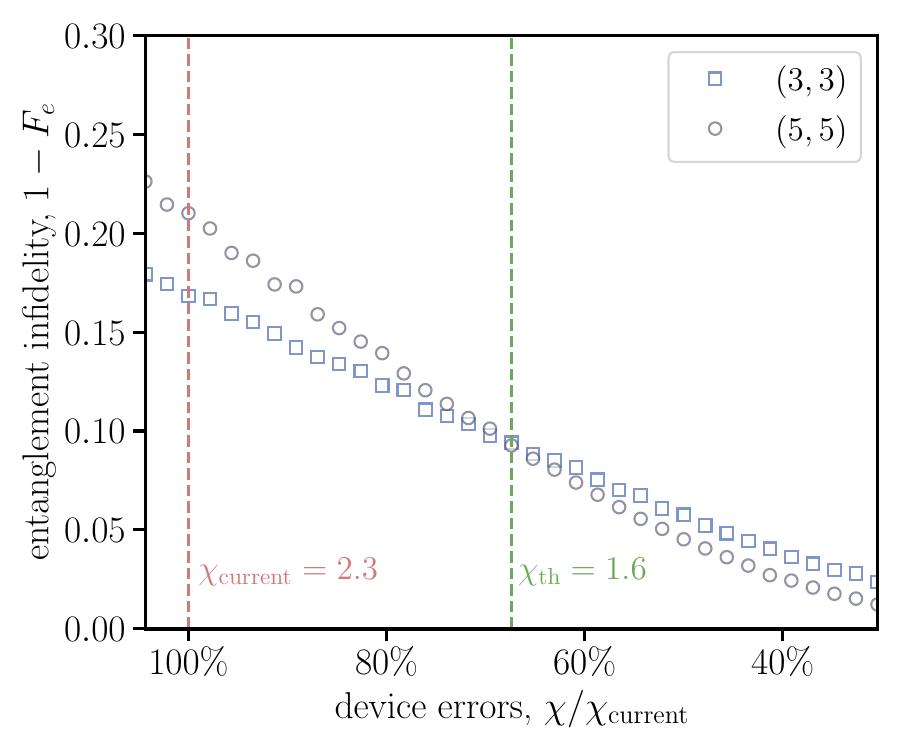}
\caption{Simulated entanglement infidelity of a single QEC cycle ($N=1$) of $d=3$ and $d=5$ surface codes as a function of the system improvements $\chi/\chi_\text{current}$. Below $\chi_\text{th}=1.6$, using the \dxdz{5}{5} code increases the fidelity over the \dxdz{3}{3} code. }
\label{fig-distance-5-threshold}
\end{figure}

In \cref{tab-noise-scaling}, we present the noise scaling factors required to match the experimental data, including those for the $X$ and $Z$ bases, as well as for the two spatial directions along which the code is grown. We report the scaling factors both with and without DD, allowing us to assess its effect on the physical error rates. The primary improvement due to DD is observed for phase-flip errors, which are reduced by $30\%$-$45\%$ (the $\ket{+}$ and $\ket{-}$ columns), with the effectiveness depending on the subset of physical qubits that are used in the circuits (not shown). In contrast, as expected, DD does not result in an improvement for bit-flip errors (the $\ket{0}$ and $\ket{1}$ columns). With DD, the $\chi$ factor for dephasing errors becomes closer to that of bit-flip errors, indicating that DD indeed reduces $Z$-type coherent and non-Markovian errors. After applying DD, the remaining error sources that contribute to the $\chi$ factor but are not captured by our noise model affect both bases equally.

The scaling factors we find, of $\chi\approx 2.3$ with DD, indicate that the noise model built from calibration data only captures $~50\%$ of the total QPU noise, while the rest of the error budget comes from error channels not included in the model. Some important noise sources that go beyond the multi-parameter Markovian Pauli channel are incoherent cross-talk~\cite{zhou2025arxiv}, which (if Markovian) cannot be canceled by DD, and coherent errors during two-qubit gates, which are not captured by randomized benchmarking~\cite{tripathi2024DB}, but are promoted to incoherent errors after syndrome measurement due to error discretization.

In \cref{fig-distance-5-sim}, we show the entanglement fidelity metric from simulations and experiments for the \dxdz{3}{3}, \dxdz{3}{5} and \dxdz{5}{3} codes, as well as the predicted performance of the \dxdz{5}{5} code. We observe that the agreement between the simulations and the data is incomplete, which indicates that using a single scaling factor is insufficient to capture the deviation from the calibrated error rates. The fact that we have to consider a median qubit performance in order to extrapolate to $d=5$ may also have an impact. However, we note that the agreement between simulations and experiments depends on the experiment run. For older data on \aachen\ (not shown)
we still found $\chi\sim2.3$, but with closer agreement. This suggests that device noise also drifts away from the calibrated values. 

The simulation results for the \dxdz{5}{5} code provide useful insights about its expected performance if run using a real device.
In particular, we find that the surface code is still above threshold for present-day IBM QPUs, but breakeven scaling could be achieved with minor system improvements: we observe that the \dxdz{5}{5} code is below threshold up to $\chi_\text{th}=1.6$. That is, a $30\%$ reduction of the experimental noise rates would suffice to reach breakeven. The simulated performance of a surface code QEC cycle for different $\chi$ factors is displayed in \cref{fig-distance-5-threshold}. We recall that $\chi=1$ corresponds to the unmodified error rates provided by IBM calibration data. Since a $\chi=1$ device would be below threshold, we can conclude that the error sources that prevent existing devices from showing a scaling advantage are precisely those not captured by the current benchmarking procedure.

\putbib[biblo]
\end{bibunit}

\end{document}